\documentclass[manuscript,screen]{acmart}

\AtBeginDocument{%
  }

\setcopyright{acmlicensed}
\copyrightyear{2018}
\acmYear{2018}
\acmDOI{XXXXXXX.XXXXXXX}

\acmConference[Conference acronym 'XX]{Make sure to enter the correct
  conference title from your rights confirmation emai}{June 03--05,
  2018}{Woodstock, NY}
\acmISBN{978-1-4503-XXXX-X/18/06}

\usepackage{amsmath}
\usepackage{proof} 
\usepackage{booktabs} 
\usepackage{mathtools} 
\usepackage{tikz} 
\usetikzlibrary{automata}
\usetikzlibrary{shapes}
\usetikzlibrary{arrows,calc,fit}
\usetikzlibrary{snakes}
\usepackage{color,soul} 
\setul{0.5ex}{0.15ex}
\setulcolor{red}
\usepackage{comment} 
\usepackage{hyperref} 
\usepackage{enumerate} 

    \newcommand{\rNum}[1]{\expandafter{\romannumeral #1\relax}}
    \newcommand{\rNUM}[1]{\uppercase\expandafter{\romannumeral #1\relax}}

    \newcommand{\bff}[1]{\mathbf{#1}}
    \newcommand{\mbb}[1]{\mathbb{#1}}
    \newcommand{\mcl}[1]{\mathcal{#1}}
    \newcommand{\spc}[1]{\begin{spacing}{#1}}
    \newcommand{\spce}{\end{spacing}}
    
    \newcommand{\la}[0]{\langle}
    \newcommand{\ra}[0]{\rangle}
    \newcommand{\mfr}[1]{\mathfrak{#1}}

    \newcommand{\Conf}[0]{\bff{L}}
    \newcommand{\Prop}[0]{\bff{Prop}}
    \newcommand{\Prog}[0]{\bff{P}}

    \newcommand{\DLp}[0]{\mbox{DL$\mfr{p}$}}
    \newcommand{\abort}[0]{\uparrow}
    \newcommand{\ter}[0]{\downarrow}
    \newcommand{\true}[0]{\mathit{true}}
    \newcommand{\false}[0]{\mathit{false}}
    \newcommand{\Fmla}[0]{\bff{F}}
    
    \newcommand{\Sem}[1]{[| #1 |]}
    \newcommand{\cnt}[0]{\mathit{C}}

    \newcommand{\trans}[0]{\longrightarrow}

    \newcommand{\GDL}[0]{\textit{GDL}}

    \newcommand{\suf}[0]{\preceq_s}
    \newcommand{\psuf}[0]{\prec_s}

    \newcommand{\sufeq}[0]{\approx_s}
    \newcommand{\mult}[0]{\preceq_m}
    \newcommand{\pmult}[0]{\prec_m}
    \newcommand{\multr}[0]{\succeq_m}
    \newcommand{\pmultr}[0]{\succ_m}

    \newcommand{\TA}[0]{\bff{TA}}
    \newcommand{\app}[0]{\mfr{I}}
    \newcommand{\free}[0]{\textit{free}}
    
    \newcommand{\DLF}[0]{\mfr{F}_{dlp}}

    \newcommand{\cfeq}[0]{=_{\textit{cf}}}
    \newcommand{\Eval}[0]{\bff{Eval}}

    \newcommand{\pfDLp}[0]{{Pr_\textit{dlp}}}

    \newcommand{\Oper}[0]{\Lambda}

    \newcommand{\FV}[0]{\textit{FV}}
    
    \newcommand{\CT}[0]{\textit{CT}}
    
    \newcommand{\termi}[0]{\Downarrow}
    \newcommand{\Terminate}[0]{\Omega}
    \newcommand{\Domain}[1]{{(\TA_{#1}, \Eval_{#1}, \app_{#1})}}
    \newcommand{\Extra}[0]{E}

    \newcommand{\WP}[0]{{\textit{WP}}}
    \newcommand{\E}[0]{E}
    \newcommand{\fodl}[0]{\textit{fodl}}
    \newcommand{\se}[0]{\leadsto}
    \newcommand{\sepc}[0]{{*}}
    \newcommand{\sepi}[0]{{-\!*}}
    
    \newcommand{\partto}[0]{\rightarrowtail}
    \newcommand{\Alloc}[0]{\bff{cons}}
    \newcommand{\disAlloc}[0]{\bff{disp}}
    \newcommand{\dom}[0]{\textit{dom}}
    \DeclareMathOperator{\disj}{\bot}
    \newcommand{\allocto}[0]{\dashrightarrow}
    \newcommand{\Sep}[0]{{\textit{SP}}}
    \newcommand{\std}[0]{\textit{std}}

    \newcommand{\Sub}[0]{{\textit{Sub}}}

    \newcommand{\LDL}[0]{\mbox{DL$\mfr{p}$}}
    \newcommand{\Kr}[0]{\mcl{K}}
    \newcommand{\Krp}[0]{M}
    \newcommand{\I}[0]{\mcl{I}}
    \newcommand{\lm}[0]{\mfr{m}}
    
    \newcommand{\EX}[0]{\textit{EX}}
    
    \newcommand{\mex}[0]{\textit{mex}}
    
    \newcommand{\PLK}[0]{PLK}
    \newcommand{\LM}[0]{\bff{M}}
    \newcommand{\DLpF}[0]{\mfr{F}_{\textit{ldlp}}}
    \newcommand{\PT}[0]{\mfr{F}_{\textit{pt}}}
    \newcommand{\PTer}[0]{\mfr{F}_{\textit{ter}}}
    \newcommand{\pfPT}[0]{Pr_{\textit{pt}}}
    \newcommand{\pfPTer}[0]{Pr_{\textit{ter}}}
    \newcommand{\pfLDLp}[0]{Pr_{\textit{ldlp}}}
    
    \newcommand{\pfOper}[0]{\textbf{Pr}_{\textit{op}}}
    \newcommand{\Wd}[0]{\mcl{S}}
    \DeclareMathOperator{\seq}{;}
    \DeclareMathOperator{\cho}{\cup}
    \newcommand{\lup}[0]{*}
    \newcommand{\upd}[0]{\mcl{U}}
    \newcommand{\pdl}[0]{\textit{pdl}}
    \newcommand{\afo}[0]{\textit{afo}}
    \newcommand{\DLpWP}[0]{\mbox{DL$\mfr{p}$-WP}}
    \newcommand{\rel}[0]{R}
    \newcommand{\DLpFODL}[0]{\mbox{DL$\mfr{p}$-FODL}}
    \newcommand{\DLpPL}[0]{\mbox{DL$\mfr{p}$-PL}}
    \newcommand{\DLpSP}[0]{\mbox{DL$\mfr{p}$-SP}}
    \DeclareMathOperator{\f}{\bff{f}}
    \DeclareMathOperator{\Suf}{\bff{suf}}
    \DeclareMathOperator{\n}{\bff{n}}
    \newcommand{\pl}[0]{\textit{pl}}
    
    \newcommand{\dddef}[0]{=_{df}}

    \newcommand{\Var}[0]{\mathit{Var}}
    \newcommand{\Etrap}[0]{\textit{trap}}
    \newcommand{\Eloop}[0]{\textit{loop}}
    \newcommand{\Eend}[0]{\textit{end}}
    \newcommand{\Eif}[0]{\textit{if}}
    \newcommand{\Ethen}[0]{\textit{then}}
    
    \newcommand{\Eemit}[0]{\textit{emit}}
    \newcommand{\Eexit}[0]{\textit{exit}}
    
    \newcommand{\Epause}[0]{\textit{pause}}
    \DeclareMathOperator{\Split}{|}
    
    \newcommand{\Wwhile}[0]{\textit{while}}
    \newcommand{\Wdo}[0]{\textit{do}}
    \newcommand{\Wif}[0]{\textit{if}}
    \newcommand{\Wthen}[0]{\textit{then}}
    \newcommand{\Welse}[0]{\textit{else}}
    \newcommand{\Wend}[0]{\textit{end}}

\begin{document}

\title{On A Parameterized Theory of Dynamic Logic for Operationally-based Programs}

\author{Yuanrui Zhang}
\email{yuanruizhang@nuaa.edu.cn}
\orcid{0000-0002-0685-6905}
\affiliation{%
  \institution{College of Software, Nanjing University of Aeronautics and Astronautics}
  \city{Nanjing}
  \state{Jiangsu}
  \country{China}
}


\renewcommand{\shortauthors}{Y. Zhang}

\begin{abstract}
Applying dynamic logics to program verifications is a challenge, because their axiomatic rules for regular expressions can be difficult to be adapted to different program models.    
We present a novel dynamic logic, called $\DLp$, which supports reasoning based on programs' operational semantics. 
For those programs whose transitional behaviours are their standard or natural semantics, $\DLp$ makes their verifications easier since one can directly apply the program transitions for reasoning, without the need of re-designing and validating new rules as in most other dynamic logics. 
$\DLp$ is parametric.   
It provides a model-independent framework consisting of a relatively small set of inference rules, which depends on a given set of trustworthy rules for the operational semantics. 
These features of $\DLp$ let multiple models easily compared in its framework and makes it compatible with existing dynamic-logic theories. 
$\DLp$ supports cyclic reasoning, providing an incremental derivation process for recursive programs, making it more convenient to reason about without prior program transformations. 
We analyze and prove the soundness and completeness of $\DLp$ under certain conditions. 
Several case studies illustrate the features of $\DLp$ and fully demonstrate its potential usage. 
\end{abstract}


\ifx Old one

\fi



\maketitle

\section{Introduction}
\label{section:Summery}

Dynamic logic~\cite{Harel00} has proven to be a valuable program logic for specifying and reasoning about different types of programs. 
As a ``multi-modal'' logic that integrates both programs and formulas into a single form, it is more expressive than traditional Hoare logic~\cite{Hoare63} (cf.~\cite{Ahrendt2025-rh}). 
Dynamic logic has been successfully applied to domains such as process algebras~\cite{Benevides10}, programming languages~\cite{Beckert2016}, synchronous systems~\cite{Zhang21,Zhang22}, hybrid systems~\cite{Platzer07b,Platzer18} and probabilistic systems~\cite{Pardo22,kozen85}. 
These theories have inspired the development of related verification tools for safety-critical systems, 
such as KIV~\cite{Reif95}, KeY~\cite{Rustan07}, KeYmaera~\cite{Platzer08}, and the tools developed in~\cite{Zhang21,Minh24}. 
Recent work in dynamic logic features a variety of extensions designed to tackle modern problems, including guaranteeing the correctness of blockchain protocols~\cite{Icard24}, formalizing hyper-properties~\cite{Gutsfeld20}, and verifying quantum computations~\cite{Takagi24,Minh24}.
It also has attracted attention as a promising framework for incorrectness reasoning, as explored recently in~\cite{Hearn19,Zilberstain23}.

The theories of most dynamic logics, as well as Hoare-style logics, are built up based on the denotational semantics of programs: The behaviour of a program is interpreted as a mathematical object (e.g. a set of traces), and a set of inference rules are constructed to match this object.  
Adapting these theories to other programs can be difficult. 
This is true especially for the programming languages such as Java, C and Esterel~\cite{Berry85}, whose semantics is very complex. 
Building logic theories for them requires carefully designing a large set of rules specific in their program domains. 
Moreover, these rules are often error prone, thus requiring validation of their soundness (and even completeness), which can also be costly. 
For example, in KeY~\cite{Rustan07}, to apply first-order dynamic logic~\cite{Pratt76} (FODL) to the verification of Java programs, more than 500 inference rules are proposed for the primitives of Java (cf.~\cite{Murk2007-mi}). Their correctness is hard to be guaranteed. 
Another example is that the tool Verifiable C~\cite{Appel_Dockins_Hobor_Beringer_Dodds_Stewart_Blazy_Leroy_2014}
spends nearly 40,000 lines of Rocq code to define and validate its logic theory for C based on separation logic~\cite{Reynolds02}. 


Another main issue is that to reason about some types of programs, one has to first transform them into some sorts of ``standard forms'', in order to apply suitable axiomatic rules. 
These beforehand transformations are unnecessary and can be expensive. And they usually mean breaking the original program structures and thus can cause loss of program information. 
A typical example is imperative synchronous programming languages such as Esterel~\cite{Berry85} or Quartz~\cite{Schneider17}. 
In~\cite{Gesell12}, it shows that how a synchronous program must be transformed into a so-called ``STA program'' in order to apply the right Hoare-logic rules to it. 

\ifx
Another main issue is that for some languages, their denotational semantics are not compositional w.r.t. the language operators. 
For these languages, additional program transformations are required to design suitable axiomatic rules. 
These transformations must be performed beforehand, thus may break the original program structures and cause loss of information. 
A typical example is imperative synchronous programming languages such as Esterel~\cite{Berry85} or Quartz~\cite{Schneider17}. 
\cite{Gesell12} shows how a synchronous program must be transformed into a so-called ``STA program'' in order to apply the right Hoare-logic rules to it. 
\fi

Different from denotational semantics, structural operational semantics~\cite{Plotkin81} describes how a program is transitioned to another program under some configuration. 
It is the standard semantics for concurrent models, such as CCS~\cite{Milner82} and $\pi$-calculus~\cite{Milner92}. 
For executable programs such as those mentioned above, providing an operational semantics is straightforward, since their executions are intended to directly transform program configurations. 
For this reason, in most cases, the operational semantics can be trusted as given, without additional validations (cf., e.g., \cite{Ellison12,Bogdanas15}). 

\textbf{In this paper}, we propose a dynamic-logic-based theory aimed for easing the reasoning of those programs whose operational semantics is in their nature. 
We propose a so-called \emph{parameterized dynamic logic}, abbreviated as $\DLp$~\footnote{To avoid the conflict with the famous Propositional Dynamic Logic (PDL)~\cite{Fischer79}}. 
It supports a directly reasoning based on operational semantics. 
Unlike previous work such as~\cite{DLForCCS,Benevides10} which focus on particular calculi, our framework is parametric and can be adapted to arbitrary programs and formulas. 
We present a proof system for $\DLp$ based on a cyclic proof approach (cf.~\cite{Brotherston07}).  
It provides a set of ``kernel rules'' for deriving $\DLp$ formulas, accompanied with an assumed set of rules for programs' transitional behaviours. 
We study and prove the soundness and completeness of $\DLp$ under a general setting. 

Compared to the traditional approaches based on dynamic logics, $\DLp$ has several advantages: 
(1) For those programs for which the operational semantics is easy to obtain and can be trusted, it reduces the burden of the target-model adaptations and consistency validations for unreliable rules. 
Compared to the previous work (like~\cite{Beckert2016}), our set of ``kernel rules'' is very small. 
(2) Its parameterization of formulas provides a model-independent framework, in which different program theories can be easily embedded through a lifting process, and multiple models can be easily compared. 
(3) Its support of cyclic reasoning is a natural solution for infinite symbolic executions caused by recursive programs, which also allows an incremental derivation process by avoiding prior program transformations for certain types of ``non-standard'' program models.  

Previous work mostly related to ours 
have addressed this issue in different mathematical logics or theories (e.g.~\cite{Rosu12,Rosu13,Stefanescu14,X.Chen19,Moore18,X.Li21,Hennicker19,Acclavio2024-uv,Teuber2025-et}, see Section~\ref{section:Related Work} for a detailed comparison). 
Except a few~\cite{Hennicker19,Acclavio2024-uv,Teuber2025-et}, most of them has not yet concerned with an efficient logical calculus for deriving dynamic-logic formulas.   
To our best knowledge, $\DLp$ is the first dynamic logic to provide a cyclic verification framework for direct operationally-based reasoning of different programs. 

This paper is a non-trivial extension of the previous work~\cite{zhang2025parameterizeddynamiclogic}, from which we take an entire different method to build up the theory of $\DLp$ that is based on Kripke structures and is independent from explicit signatures. 
In this work, we further make a full analysis of the soundness and completeness of $\DLp$, as well as a much richer cases analysis for different features of $\DLp$.

\ifx
The methodology of reasoning based on operational semantics has been proposed and studied for years within different mathematical theories, among them the work based on rewriting logics~\cite{Rosu12,Rosu13,Stefanescu14,X.Chen19}, set theory and coinduction~\cite{Moore18,X.Li21}, and program updates\cite{Platzer07b,Beckert13,Beckert2016} are closest to our work (see Section~\ref{section:Related Work} for a detailed comparison). 
Except a few work such as~\cite{MossakowskiEA09,Hennicker19}, to the best of our knowledge, most of the previous work has not yet addressed a similar approach under dynamic-logic settings, i.e., to provide an efficient logical calculus for deriving dynamic formulas. 
In our opinion, it is valuable to fill in this gap. 
\fi

\section{An Overview}
\label{section:An Overview}

In dynamic logic, a \emph{dynamic formula} is of the form: $[\alpha]\phi$ (cf.~\cite{Harel00}), where $[\cdot]$ is a modal operator, $\alpha$ is a program model (or simply ``program''), $\phi$ is a logical formula. Intuitively, it means that after the terminations of all executions of program $\alpha$, formula $\phi$ holds. 
When $\alpha$ is deterministic, formula $\phi\to [\alpha]\psi$ exactly captures the partial correctness of the triple $\{\phi\}\alpha\{\psi\}$ in Hoare-style logics (e.g.~\cite{Hoare63,Reynolds02}). 
$\phi \to \la\alpha\ra\psi$ captures the total correctness of $\{\phi\}\alpha\{\psi\}$, with $\la\cdot \ra$ the dual modal operator of $[\cdot]$. $\la\alpha\ra\psi$ is defined such that $\la\alpha\ra\psi = \neg[\alpha]\neg\psi$, meaning that there exists an execution of $\alpha$ satisfying that it terminates and after its termination, formula $\psi$ holds.   
As a combination of both programs and formulas, a dynamic formula allows multiple and nested modalities in forms like $[\alpha]\phi\to \la\beta\ra\psi$, $[\alpha]\la\beta\ra\phi$, $[\alpha](\phi\to \la\beta\ra\psi)$, etc., making it strictly more expressive than Hoare logic (cf.~\cite{Ahrendt2025-rh}). 

The rest of this section gives an outline of the main work on $\DLp$, focusing on the main ideas illustrated through examples. They are introduced in details in the following sections of this paper. 

\ifx
Traditional dynamic logics, like PDL~\cite{Fischer79} \& FODL~\cite{Pratt76}, take regular programs as their program models (cf. Section~\ref{section:PDL FODL}). 
When applying them to new program models, new inference rules need to be checked if their soundness matches the denotational semantics of programs. 
For instance, in FODL, we have a rule for a sequential program $\alpha\seq\beta$ as: $$\begin{aligned}\infer[^{([\seq])}]{\Gamma\Rightarrow [\alpha\seq \beta]\phi, \Delta}{\Gamma\Rightarrow [\alpha][\beta]\phi}\end{aligned},$$ which means that to prove that program $\alpha\seq\beta$ satisfies $\phi$ after its executions, it is sufficient to prove that program $\alpha$ satisfies $[\beta]\phi$ after its executions. 
The soundness of this rule corresponds to the denotational semantics $[|\alpha\seq \beta|]$ of the program $\alpha\seq \beta$, which is defined as $[|\alpha\seq \beta|] = [|\alpha|]\circ [|\beta|]$ (cf.~Section~\ref{section:PDL FODL}). 

The consistency between the axiomatic semantics and denotational semantics of programs requires that the validity of an inference rule for a program is according to its denotational semantics. 
For example, in first-order dynamic logic (FODL)~\cite{Harel00}, we have a rule for a sequential program $\alpha\seq\beta$ as: $$\begin{aligned}\infer[^{([\seq])}]{\Gamma\Rightarrow [\alpha\seq \beta]\phi, \Delta}{\Gamma\Rightarrow [\alpha][\beta]\phi}\end{aligned},$$ which means that to prove that program $\alpha\seq\beta$ satisfies $\phi$ after its executions, it is sufficient to prove that program $\alpha$ satisfies $[\beta]\phi$ after its executions. 
The validity of this rule relies on the denotational semantics $[|\alpha\seq \beta|]$ of the program $\alpha\seq \beta$, which is defined as $[|\alpha\seq \beta|] = [|\alpha|]\circ [|\beta|]$, intuitively meaning that a trace of $\alpha\seq\beta$ consists of a trace of $\alpha$ concatenated (through the operator $\circ$) by a trace of $\beta$. 

[Deleted from the above: 

The behaviour of a program is interpreted as a mathematical object (e.g. a set of traces), and a set of inference rules are constructed to match this object. 

\emph{operational semantics} describes how a program $\alpha$ under a configuration $\sigma$, denoted by $(\alpha, \sigma)$, is transitioned to another program $\alpha'$ under a configuration $\sigma'$ (denoted by $(\alpha', \sigma')$). 
]
\fi

\subsection{Labeling and Parameterization of Dynamic Logic}
In order to directly reason about programs via their operational semantics, 
in $\DLp$, we introduce a ``label'' $\sigma$ (Definition~\ref{def:Labels and Label Mappings}) to capture explicit program structures as the current program configurations for symbolic executions. 
$\sigma$ attaches a dynamic formula $[\alpha]\phi$, yielding a labeled formula of the form $\sigma : [\alpha]\phi$. Intuitively, it means that under configuration $\sigma$, program $\alpha$ can be executed and formula $\phi$ holds after all terminating executions of $\alpha$. 
The introduction of labels allows us to derive labeled formulas $\sigma : [\alpha]\phi$ by the following program transitions: $$(\alpha, \sigma)\trans (\alpha', \sigma'),$$ 
using inference rules conceptually explained as the form:
$$
\begin{aligned}
    \infer[^{([\alpha])}]
    {\sigma : [\alpha]\phi}
    {\sigma' : [\alpha']\phi\mbox{, for all $(\alpha', \sigma')$ such that $(\alpha, \sigma)\trans (\alpha', \sigma')$}}
\end{aligned}.
$$
These rules (corresponding to the rules $([\alpha]R)$ and $([\alpha]L)$ in Table~\ref{table:General Rules for LDL}) reduce the deduction of $\sigma : [\alpha]\phi$ to the deductions of all successor formulas $\sigma' : [\alpha']\phi$ corresponding to the one-step program transitions.

Due to the universal form of program transitions, this framework applies for arbitrary program models and configurations. 
Consequently, in a labeled formula $\sigma : [\alpha]\phi$ of $\DLp$, we parameterize the program $\alpha$, the logical formula $\phi$ and the label $\sigma$, to allow them to have any algebraic structures. 
$\sigma : [\alpha]\phi$ turns out to be a more general form than $[\alpha]\phi$. 
When $\sigma$ is ``free'' (cf.~Definition~\ref{def:free configurations}) w.r.t. $[\alpha]\phi$, $\sigma : [\alpha]\phi$ has the same meaning as $[\alpha]\phi$.  

Consider a formula $\phi_1 \dddef (x\ge 0\to [x := x + 1]x > 0)$ in FODL~\cite{Pratt76}, 
where $x$ is a variable ranging over integers $\mbb{Z}$.  
Intuitively, formula $\phi_1$ means that if $x\ge 0$ holds, then $x>0$ holds after assigning the expression $x + 1$ to $x$. 
In FODL, to derive $\phi_1$, 
we apply the assignment rule: 
$$
\begin{aligned}
\label{equ:assignRule}
\infer[^{(x := e)}]
{[x := e]\phi}
{\phi[e / x]}
\end{aligned}
$$
on the part $[x:=x+1]x > 0$. 
It substitutes $x$ of $x > 0$ with $x + 1$, yielding expression $x + 1 > 0$. 
After the derivation we obtain formula $\phi'_1\dddef (x\ge 0 \to x + 1 > 0)$, which is true for any $x\in \mbb{Z}$. 

In $\DLp$, on the other hand, we can express $\phi_1$ as an equivalent labeled formula:  $\psi_1\dddef (t\ge 0\to \{x\mapsto t\} : [x := x + 1]x > 0)$, where the label $\{x\mapsto t\}$ means that variable $x$ stores value $t$ (with $t$ a fresh variable). 
With the program configurations explicitly showing up, to derive formula $\psi_1$, we instead directly apply the above rule $([\alpha])$ on the part $\{x\mapsto t\} : [x:=x+1]x>0$ according to the program transition
$$(x:=x+1, \{x\mapsto t\})\trans (\ter, \{x\mapsto t + 1\}),\ \ \ \ \ (op\ x:=e)$$
which assigns the value $t+1$ to $x$ afterwards. Here $\ter$ indicates a program termination (cf. Definition~\ref{def:Progarms and Formulas}). 
After the derivation, we obtain the formula $\psi'_1\dddef (t\ge 0\to \{x\mapsto t + 1\} : x > 0)$, 
where formula $\{x\mapsto t + 1\} : x > 0$ exactly means $t + 1 > 0$ if we replace $x$ with its current value $t+1$ in formula $x > 0$. 
So from $\psi'_1$, we obtain formula $t\ge 0\to t + 1 > 0$, which is exactly formula $\phi'_1$ (modulo free-variable renaming). 

\ifx No saying
Reasoning programs whose natural semantics is operational in $\DLp$ does not require the design of extra inference rules or program transformations. 
For a sequence Quartz program $\alpha\seq\beta$ we have seen, 
we can directly reason about it based on a program transition $(\alpha\seq \beta, \sigma)\trans(\alpha'\seq \beta, \sigma')$ for some configurations $\sigma, \sigma'$ and program $\alpha'$, despite the fact that $\alpha\seq \beta$ may not be in an STA form. 
\fi

From this example, we see that the above rule $([\alpha])$ can be directly applied to other languages (by just choosing a different set of program transitions) while rule $(op\ x:=e)$ may not. It cannot be applied to, e.g., a Java statement $x := new\ C(...)$, which creates a new object of class $C$ (cf.~\cite{Beckert2016}). 
Throughout this paper (from Example~\ref{example:While program} - \ref{example:Proof system for program behaviours}, in Section~\ref{section:Case Study} and~\ref{section:Instantiation of FODL in DLp}), 
we show that how $\DLp$ can be adapted to different theories of programs through two instantiations of $\DLp$: $\DLpWP$ and $\DLpFODL$. 
Section~\ref{section:Instantiation of FODL in DLp} also displays the capability of $\DLp$ to derive multiple program models in a single framework. 

Section~\ref{section:Encoding of Complex Configurations} and Appendix~\ref{section:An Encoding of Separation Logic in DLp} further give two instantiations: $\DLpPL$ and $\DLpSP$ separately to show that in $\DLp$ not only static properties (i.e. the properties holding on a state) can be expressed, but also more complex properties, like temporal properties and spatial properties. 

The entire process of labeling and parameterization is fully introduced in Section~\ref{section:Dynamic Logic LDL}. In Section~\ref{section:A Proof System for LDL}, a proof system $\pfDLp$ for $\DLp$ is built.

\subsection{Lifting Process and Compatibility of $\DLp$}
As a dynamic logic extended with the extra structure labels, 
$\DLp$ is compatible with the existing theories of dynamic logics in the sense that every inference rule for non-labeled dynamic formulas can be lifted as a rule for their labeled counterparts in $\DLp$. 
Section~\ref{section:Lifting Process From Program Domains} discusses this technique in detail, where
we introduce a notion called ``free labels'' (Definition~\ref{def:free configurations}), and show that attaching a free label to a formula does not affect the validity of this formula (Theorem~\ref{prop:lifting process 2}). 

For instance, from the rule $(op\ x:=e)$ above, one can obtain a sound lifted rule by attaching to each formula the label $\{x\mapsto t\}$: 
$$
\infer[^{(\textit{lf}\ (x:=e))}]
{\{x\mapsto t\} : [x:=e]\phi}
{\{x\mapsto t\} : \phi[x/e]}. 
$$
$\{x\mapsto t\}$ is free as $t$ is a fresh variable. Trivially, replacing any free occurrence of variable $x$ with variable $t$ in the formulas $[x:=e]\phi$ and $\phi[x/e]$ does not change their meanings.   
From the formula $\psi_1$ above, by applying the rule $(\textit{lf}\ (x:=e))$ on the part $\{x\mapsto t\} : [x:=x+1]x>0$, we obtain the formula 
$\psi''_1 \dddef t\ge 0\to \{x\mapsto t\} : x + 1 > 0$. 
It is just $\psi'_1$ we have seen above if we replace $x$ of $x + 1 > 0$ with its current value $t$. 

Lifting process provides a type of flexibility by directly making use of the rules special in different domains. 
In Section~\ref{section:Lifting Process in While Programs}, we illustrate in detail how this technique can be beneficial during derivations.

\subsection{Cyclic Reasoning of $\DLp$ Formulas}
In an ordinary deductive procedure we usually expect a finite proof tree. 
However, in the proof system $\pfDLp$ of $\DLp$, a branch of a proof tree does not always terminate, because the process of symbolically executing a program via rule $([\alpha]R)$ or/and rule $([\alpha]L)$ might not stop. 
This is well-known when a program has an explicit/implicit loop structure that may run infinitely. 
For example, in the instantiated theory $\DLpWP$ of $\DLp$ (cf. Example~\ref{example:While program} - \ref{example:Proof system for program behaviours}), a while program $$W \dddef \textit{while}\ \textit{true}\ \textit{do}\ x := x + 1\ \textit{end}$$ proceeds infinitely as the following program transitions:
$$(W, \{x \mapsto 0\})\trans (W, \{x\mapsto 1\})\trans ....$$ 
This yields the following infinite derivation branch in $\DLp$ when deriving, for example, a formula $\{x\mapsto 1\} : [W]\phi$: 
$$
\infer[^{([\alpha])}]
{\{x\mapsto 1\} : [W]\phi}
{
    \infer[^{([\alpha])}]
    {\{x\mapsto 2\} : [W]\phi}
    {
        \infer*[]
        {...}
        {
            \infer[]
            {...}
            {
                \infer[]
                {\{x\mapsto n\} : [W]\phi}
                {...}
            }
        }
    }
}. 
$$

To solve this problem, we propose a cyclic proof system for $\DLp$ (Section~\ref{section:Construction of A Cyclic Preproof Structure}). 
Cyclic proof approach (cf.~\cite{Brotherston07}) is a technique to admit a certain type of infinite deductions, called ``cyclic proofs'' (cf.~Section~\ref{section:Proof Preproof Cyclic Preproof}). 
A cyclic proof is a finite proof tree augmented with some non-terminating leaf nodes, called ``buds'', which are identical to some of their ancestors. Call a bud and one of its identical ancestors a ``back-link''.  

We propose a cyclic derivation approach special for $\DLp$. This mainly consists of the following two steps. 

In the first step, we construct a cyclic structure by identifying suitable buds and back-links, where the most critical work is to design the substitution rule $(\Sub)$ of labels (Definition~\ref{def:Substitution of Labels}). 
For example, by performing the substitution rule $(\textit{Sub})$ given in~Section~\ref{section:Case Study} on the labels $\{x\mapsto 1\}$ and $\{x\mapsto t+1\}$, we can obtain a cyclic derivation for the formula $\{x\mapsto 1\} : [W]\phi$ on the left below: 
$$
\begin{aligned}
    \infer[^{(\Sub)}]
    {\{x\mapsto 1\} : [W]\phi}
    {
        \infer[^{([\alpha])}]
        {1:\ \{x\mapsto t\} : [W]\phi}
        {
            \infer[^{(\Sub)}]
            {\{x\mapsto t + 1\} : [W]\phi\}}
            {2:\ \{x\mapsto t\} : [W]\phi}
        }
    }
\end{aligned}, \ \ \ 
\begin{aligned}
    \infer[^{(\Sub)}]
    {\{x\mapsto 1\} : \la W\ra\phi}
    {
        \infer[^{([\alpha])}]
        {\{x\mapsto t\} : \la W\ra\phi}
        {
            \infer[^{(\Sub)}]
            {\{x\mapsto t + 1\} : \la W\ra\phi\}}
            {\{x\mapsto t\} : \la W\ra\phi}
        }
    }
\end{aligned}
$$
where node 2 is a bud and it back-links to node 1. 
$t$ is a fresh variable w.r.t. $x$, $W$ and $\phi$. 
The label $\{x\mapsto 1\}$ equals to $\{x\mapsto t\}[1/t]$ (i.e. the label obtained by substituting $t$ with $1$) and the label $\{x\mapsto t + 1\}$ equals to $\{x\mapsto t\}[t + 1/t]$ (i.e. the label obtained by substituting $t$ with expression $t+1$). 

However, not all cyclic structures are cyclic proofs. 
Consider the cyclic derivation on the above right for formula $\{x\mapsto 1\} : \la W\ra\phi$. 
According to the semantics of modality $\la\cdot \ra$ (cf.~Section~\ref{section:Syntax and Semantics of DLp Formulas}), 
$\{x\mapsto 1\} : \la W\ra\phi$ is invalid for any formula $\phi$ because $W$ never terminates. 
Therefore, in the second step, we need to check whether these cyclic structures are legal cyclic proofs, where the key step is to define suitable ``progressive derivation traces'' special for system $\pfDLp$ (Definition~\ref{def:Progressive Step/Progressive Derivation Trace}). 

We give two examples in Section~\ref{section:Case Study} and Appendix~\ref{section:Example Two: A Synchronous Loop Program} respectively to show how cyclic reasoning in $\DLp$ can be carried out and  
how we can benefit from the incremental reasoning of recursive programs by cyclic graphs. 
Especially, the example given in Appendix~\ref{section:Example Two: A Synchronous Loop Program} is an Esterel program, from which we can see that how cyclic reasoning based on operational semantics can prevent extra prior program transformations in synchronous languages. 



\ifx
: (1) constructing a cyclic structure by identifying suitable buds and back-links, where the most critical work is to design the substitution rule $(\Sub)$ of labels (Definition~\ref{def:Substitution of Labels}); And (2) checking whether these cyclic structures are legal cyclic proofs, where 

In Section~\ref{section:Construction of A Cyclic Preproof Structure}, we propose a cyclic condition for 
\fi

\ifx
Cyclic proof approach (cf.~\cite{Brotherston07}) is a powerful technique to admit a certain type of infinite deductions (i.e. ``cyclic proofs'') caused by symbolic executions of programs with loop structures (Section~\ref{section:Construction of A Cyclic Preproof Structure}). 
It has been attracting more and more attention and recently has been applied to many logic theories such as~\cite{Gadi20,Jones22,Afshari22}. 
We investigate this approach and adapt it to the theory of $\DLp$. 
The main challenges of this part are (1) identifying a sound cyclic proof system, where the most critical work is to design rule $(\Sub)$ and the ``substitutions of labels'' (Definition~\ref{def:Substitution of Labels}); and (2) constructing a cyclic proof structure, where the key step is to define ``progressive derivation traces'' (Definition~\ref{def:Progressive Step/Progressive Derivation Trace}). 
At last, as an important contribution, we prove the soundness of our cyclic proof system (Section~\ref{section:Proof of Theorem - theo:Soundness of A Cyclic Preproof}). 
\fi

\subsection{Soundness and Completeness of $\DLp$}
We analyze and prove the soundness and completeness of $\DLp$ w.r.t. arbitrary programs and formulas under certain restriction conditions  (Section~\ref{section:Proof of Theorem - theo:Soundness of A Cyclic Preproof}). 

The soundness of $\DLp$ states that a cyclic proof always leads to a valid conclusion. 
In Section~\ref{section:Conditional Soundness of DLp}, we prove it under a condition (Definition~\ref{def:Program Properties}) that restricts how a program can terminate. 
Even though, the types of restricted programs is still very rich, enough to include all deterministic programming languages (cf.~Section~\ref{section:Conditional Soundness of DLp}).   

The idea of proving the soundness is by contradiction (cf.~\cite{Brotherston08}). We assume the conclusion, e.g. $\{x\mapsto 1\} : [W]\phi$, is invalid, then it leads a sequence of invalid formulas in some proof branch, e.g. $\{x\mapsto 1\} : [W]\phi$, $\{x\mapsto t\} : [W]\phi$, $\{x\mapsto t + 1\} : [W]\phi$,... in the above derivation.
This sequence of invalid formulas then causes the violation of a well-founded set (Definition~\ref{def:Relation pmult}) of a type of metrics (Definition~\ref{def:counter-example set}) that relate these invalid formulas. More details is given in Section~\ref{section:Conditional Soundness of DLp}. 


The completeness of $\DLp$ states that for any valid labeled dynamic formula, there is a cyclic proof  for it. 
We prove the completeness of $\DLp$ under a sufficient assumption about the so-called ``loop programs'' (Definition~\ref{def:Expression Finiteness Property}, \ref{def:Well-behaved Loop Programs}). 
The main idea and the details of the proof are given in Section~\ref{section:Conditional Completeness of DLp} and Appendix~\ref{section:Other Propositions and Proofs}. 
This completeness result is useful because the restriction condition is general: any instantiation of $\DLp$ is complete once its program models satisfy this condition.

\subsection{Main Contributions \& Content Structure}
The main contributions of this paper can be summarized as follows:
\begin{itemize}
    \item We define the syntax and semantics of DLp formulas.
    \item We construct a labeled proof system and develop a lifting process for $\DLp$. 
    \item We propose a cyclic proof approach tailored for $\DLp$.
    \item We analyze and prove the soundness and completeness of $\DLp$ under certain conditions. 
\end{itemize}

The rest of the paper is organized as follows. 
Section~\ref{section:PDL FODL} gives a brief introduction to PDL and FODL, necessary for understanding the main content. 
In Section~\ref{section:Dynamic Logic LDL}, we define the syntax and semantics of $\DLp$. 
In Section~\ref{section:A Cyclic Proof System for LDL}, we propose a cyclic proof system for $\DLp$. 
In Section~\ref{section:Case Studies}, we analyze some case studies.  
Section~\ref{section:Proof of Theorem - theo:Soundness of A Cyclic Preproof} analyzes the soundness and completeness of $\DLp$. 
Section~\ref{section:Related Work} introduces related work, while Section~\ref{section:Discussions and Future Work} makes a conclusion and discusses about future work.

\section{Prerequisites : PDL \& FODL}
\label{section:PDL FODL}

In propositional dynamic logic (PDL)~\cite{Fischer79}, 
the syntax of a formula $\phi$ is given by simultaneous inductions on both programs and formulas as follows in BNF form:
$$
\begin{aligned}
\alpha\dddef&\ a\ |\ \phi?\ |\ \alpha\seq\alpha\ |\ \alpha \cho \alpha\ |\ \alpha^\lup,\\
\phi\dddef&\ p\ |\ \neg\phi\ |\ \phi\wedge \phi\ |\ [\alpha]\phi. 
\end{aligned}
$$
In the above definition, $\alpha$ is a regular expression with tests, often called a \emph{regular program}. 
$a\in A$ is an atomic action of a symbolic set $A$. 
$\phi?$ is a test. If $\phi$ is true, then the program proceeds, otherwise, the program halts;
$\alpha\seq \beta$ is a sequential program, meaning that after program $\alpha$ terminates, $\beta$ proceeds. 
$\alpha\cho \beta$ is a choice program, it means that either $\alpha$ or $\beta$ proceeds non-deterministically. 
$\alpha^\lup$ is a star program, which means that $\alpha$ proceeds for an arbitrary number $n\ge 0$ of times. 
$p$ is an atomic formula, including the boolean $\true$. 
We call a formula having the modality $[\cdot]$ a \emph{dynamic formula}. 
Intuitively, formula $[\alpha]\phi$ means that after all executions of program $\alpha$, formula $\phi$ holds. 

The semantics of PDL is given based on a Kripke structure (cf.~\cite{Harel00}) $\Krp = (S, \to, I)$, where $S$ is a set of \emph{worlds}; $\to\subseteq S\times A\times S$ is a set of transitions labeled by atomic programs; $I : P\to \mcl{P}(S)$ interprets each atomic PDL formula of set $P$ to a set of worlds.  

Given a Kripke structure $\Krp$, the semantics of PDL is based on the denotational semantics $\Sem{\cdot}$ of regular programs, given as a satisfaction relation $\Krp, w\models \phi$ between a world $w$ and a PDL formula $\phi$.
It is defined as follows by simultaneous inductions on both programs and formulas: 
\begin{enumerate}[a.]
    \item $\Sem{a}\dddef \{(w, w')\ |\ \mbox{$w\xrightarrow{a}w'$ on $\Krp$}\}$;
    \item $\Sem{\phi?}\dddef \{(w, w)\ |\ \Krp, w\models \phi\}$;
    \item $\Sem{\alpha\seq\beta}\dddef \{(w, w')\ |\ \exists w''. (w, w'')\in \Sem{\alpha}\wedge (w'', w')\in \Sem{\beta}\}$; 
    \item $\Sem{\alpha\cho\beta}\dddef \Sem{\alpha}\cup \Sem{\beta}$;
    \item $\Sem{\alpha^\lup}\dddef \bigcup^{\infty}_{n=0}\Sem{\alpha^n}$, where 
    $\alpha^0\dddef \true ?$, $\alpha^n \dddef \alpha\seq \alpha^{n-1}$ for any $n\ge 1$. 
\end{enumerate}
\begin{enumerate}[1.]
    \item $\Krp, w\models p$, if $w\in I(p)$;
    \item $\Krp, w\models \neg\phi$, if $\Krp, w\not\models \phi$;
    \item $\Krp, w\models \phi\wedge \psi$, if $\Krp, w\models \phi$ and $\Krp, w\models \psi$;
    \item $\Krp, w\models [\alpha]\phi$, if for all $(w, w')\in \Sem{\alpha}$, $\Krp, w'\models \phi$. 
\end{enumerate}

PDL studies the formulas that are valid w.r.t. all Kripke structures. Its proof system is complete (cf.~\cite{Harel00}). 

First-order dynamic logic (FODL)~\cite{Pratt76} is obtained from PDL by specializing the atomic actions $a$ and atomic formulas $p$ in PDL in some special domains. 
In FODL, an atomic action is an assignment of the form $x := e$, where $x\in \Var_{\fodl}$ is a variable and $e$ is an expression. 
Usually, we consider $e$ as an arithmetical expression of integer domain $\mbb{Z}$, e.g., $x + 5$ and $x - 2 * y$, where $x, y\in \Var_{\fodl}$. $+, -, *, /$ are the usual arithmetical operators.   
An atomic formula is an arithmetical relation $e_1 \bowtie  e_2$ with $\bowtie\in \{=, <, \le, >, \ge\}$, such as $x + 5 < 0$ and $x - 2 * y = 1$. 
The non-dynamic formulas in FODL are thus the usual arithmetical first-order formulas linked by the logical connectives $\neg, \wedge$ and the quantifier $\forall$.  

\ifx
A formula $\phi$ in the test $\phi?$ is usually a first-order arithmetical formula, where the atomic formulas are arithmetical relations like $x  + 5 > 0$ and $x - 2 * y = 1$ and composed formulas are linked by the logical connectives $\neg, \wedge, \vee, \to$ and quantifiers $\forall$ and $\exists$, such as $x  + 5 > 0 \wedge x - 2 * y = 1$ and $\forall x. (x + 5) > 0$. 
Given an FODL formula $\phi$, a \emph{substitution} $\phi[e/x]$ returns a formula in which each free occurrence of variable $x$ is replaced by an expression $e$. 
One can refer to~\cite{Harel00} for more formal definitions. 
\fi

In FODL, a \emph{state} $w : \Var_{\fodl}\to \mbb{Z}$ maps each variable to an integer. 
For an expression $e$, $w(e)$ returns the value obtained by replacing all the free occurrences of each variable $x$ in $e$ with the value $w(x)$. 
The Kripke structure $\Krp_\fodl = (S_\fodl, \to_\fodl, I_\fodl)$ of FODL is defined such that $S_\fodl$ is the set of all states w.r.t. $\Var_\fodl$ and $\mbb{Z}$. 
For each assignment $x:=e$, $w\xrightarrow{x:=e} w'$ is a relation on $\Krp_\fodl$ iff $w' = w[x\mapsto e]$, where $w[x\mapsto e]$ returns a state that maps $x$ to value $w(e)$ and maps other variables to the value the same as $w$. 
$I_\fodl$ interprets each atomic FODL formula as the set of states in which it is satisfied. 
Or formally, for a state $w\in I_\fodl(e_1\bowtie e_2)$, $w(e_1) \bowtie w(e_2)$ is true.  
Based on these, the semantics of FODL can be defined in a similar way as shown above. 
\ifx
for which we only show the cases for the assignments $x:=e$ and the atomic formulas as below:
$$
\begin{aligned}
    &\Sem{x := e}\dddef \{(s, s')\ |\ s\xrightarrow{x:=e} s', s' = s[x\mapsto e]\};\\
    &......\\
    &\Krp_\fodl, s\models e_1\bowtie e_2,\mbox{ if }w\in \I_\fodl(e_1\bowtie e_2);\\
    &......
\end{aligned}
$$
where $s[x\mapsto e]$ returns a state that maps $x$ to value $s(e)$ and maps other variables to the value the same as $s$. 
\fi
One can refer to~\cite{Harel00} for a more formal definition of FODL. 

FODL forms the language basis of many existing dynamic-logic theories~\cite{Benevides10,Beckert2016,Zhang21,Zhang22,Platzer07b,Platzer18,Pardo22,kozen85,Feldman84} as mentioned in Section~\ref{section:Summery}. 
Some are the extensions of FODL by adding new primitives, while the others can be expressed by FODL. 
For example, for the traditional while programs $\alpha$: 
$$
\alpha \dddef x := e\ |\ \alpha\seq \alpha\ |\ \Wif\ \phi\ \Wthen\ \alpha\ \Welse\ \beta\ \Wend\ |\ \Wwhile\ \phi\ \Wdo\ \alpha\ \Wend, 
$$
their special statements can be captured by FODL as follows (cf.~\cite{Harel00}): 
$$
\begin{aligned}
    \Wif\ \phi\ \Wthen\ \alpha\ \Welse\ \beta\ \Wend\dddef&\ \phi?\seq \alpha \cho \neg\phi?\seq \beta, \\
    \Wwhile\ \phi\ \Wdo\ \alpha\ \Wend \dddef&\  (\phi?\seq \alpha)^\lup\seq \neg\phi?. 
\end{aligned}
$$

\ifx
of usually integer domain $\mbb{Z}$. 
Many existing dynamic-logic theories~\cite{Benevides10,Beckert2016,Zhang21,Zhang22,Platzer07b,Platzer18,Pardo22,kozen85,Feldman84} as mentioned above were developed based on FODL by extending regular programs with special language primitives and adding inference rules for these primitives.    
\fi

\section{Dynamic Logic $\DLp$}
\label{section:Dynamic Logic LDL}

\subsection{Syntax and Semantics of $\DLp$ Formulas}
\label{section:Syntax and Semantics of DLp Formulas}

The theory of $\DLp$ extends PDL by permitting the program $\alpha$ and formula $\phi$ in modalities $[\alpha]\phi$ to take arbitrary forms, only subject to
some restriction conditions when discussing its soundness and completeness in Section~\ref{section:Proof of Theorem - theo:Soundness of A Cyclic Preproof}.   

In the relations defined in this section, we use $\cdot$ to express an ignored object whose content does not really matter. For example, we may write $w\xrightarrow{} \cdot$, $w\xrightarrow{\alpha/\cdot}w'$, $(\alpha, \sigma)\trans (\alpha', \cdot)$ and so on.

\begin{definition}[Programs \& Formulas]
\label{def:Progarms and Formulas}
 In $\DLp$ we assume two pre-defined disjoint sets $\Prog$ and $\Fmla$. 
$\Prog$ is a set of programs, in which we distinguish a special program $\ter\in \Prog$ called the ``termianal program''. 
$\Fmla$ is a set of formulas.   
\end{definition}

\begin{definition}[$\DLp$ Formulas]
\label{def:DLp Formulas}
A dynamic logical formula
$\phi$ w.r.t. the parameters $\Prog$ and $\Fmla$, called a ``parameterized dynamic logic'' ($\DLp$) formula, is defined as follows in BNF form:
$$
\begin{aligned}
    \phi \dddef&\  F\ |\ \neg \phi\ |\ \phi\wedge \phi\ |\ [\alpha]\phi,\\
\end{aligned}
$$
where $F\in \Fmla$, $\alpha\in \Prog$; $[\cdot]$ is a new operator that does not appear in any formula of $\Fmla$. 

We denote the set of $\DLp$ formulas as $\DLF$. 
\end{definition}

A $\DLp$ formula is called a \emph{dynamic formula} if it contains a modality $[\cdot]$ within it. 
Intuitively, formula $[\alpha]\phi$ means that after the terminations of all executions of program $\alpha$, formula $\phi$ holds. 
$\la\cdot\ra$ is the dual operator of $[\cdot]$. 
Formula $\la\alpha \ra\phi$ is expressed as $\neg [\alpha]\neg\phi$. 
Other formulas with logical connectives such as $\vee$ and $\to$ can be expressed by formulas with $\neg$ and $\wedge$ accordingly.  

Following the convention of defining a dynamic logic (cf.~\cite{Harel00}), we introduce a novel Kripke structure to capture the parameterized program behaviours in $\Prog$. 


\begin{definition}[Program-labeled Kripke Structures]
\label{def:Program-labeled Kripke Structure}
    A ``program-labeled'' Kripke (\PLK) structure w.r.t. parameters $\Prog$ and $\Fmla$ is a triple $$K(\Prog, \Fmla) \dddef (\Wd, \trans, \I),$$ where 
    $\Wd$ is a set of worlds; $\trans\subseteq \Wd\times (\Prog\times \Prog)\times \Wd$ is a set of relations labeled by program pairs, in the form of $w_1\xrightarrow{\alpha/\alpha'}w_2$ for some $w_1, w_2\in S$, $\alpha, \alpha'\in \Prog$; $\I: \Fmla\to \mcl{P}(\Wd)$
    is an interpretation of formulas in $\Fmla$ on the power set of worlds. 
    Moreover, $K(\Prog, \Fmla)$
    satisfies that
    $w\not\xrightarrow{\ter/\alpha} \cdot$ for any $w\in \Wd$ and $\alpha\in \Prog$. 
\end{definition}

Definition~\ref{def:Program-labeled Kripke Structure} differs from the Kripke structures $M$ of PDL (Section~\ref{section:PDL FODL}) in the following aspects: (1) It introduces a program-labeled relation of the form: $w_1\xrightarrow{\alpha/\alpha'}w_2$; (2) It introduces an additional condition for the terminal program $\ter$. 
The program-labeled relations describe
programs' transitional behaviours, which is usually captured by their operational semantics (as we see in Section~\ref{section:A Proof System for LDL}).
This is unlike the relations in $M$, where a relation only captures the behaviours of an atomic program. 
Intuitively, $w_1\xrightarrow{\alpha/\alpha'}w_2$ means that from world $w_1$, program $\alpha$ is transitioned to program $\alpha'$, ending with world $w_2$. 
The condition for $\ter$ exactly captures the meaning of program termination. 

Below in this paper, \textbf{our discussion is always based on an assumed \PLK\ structure namely $\Kr(\Prog, \Fmla) = (\Wd, \trans, \I)$}.

Starting from Example~\ref{example:While program} below, through Example~\ref{example:Labels and Formulas of labeled Sequents}, \ref{example:Label Mappings} and~\ref{example:Proof system for program behaviours} we gradually instantiate the theory $\DLp$ in the setting of the special theory FODL, where we restrict the program models of FODL to a simpler one --- the while programs. 
To do this, we give explicit definitions for the parameters $\Prog, \Fmla, \Conf, \LM$ and the proof system $\pfOper$ for the operational semantics in $\DLp$ ($\Conf, \LM$ and $\pfOper$ are introduced below). 
The logical theory after instantiated is called $\DLpWP$. 

\begin{example}[An Instantiation of Programs and Formulas]
\label{example:While program}
    Consider instantiating $\Prog$ by the set of while programs defined in Section~\ref{section:PDL FODL}, namely $\Prog_W$. 
    Consider a program $\WP$ in $\Prog_W$:
    $$
    \WP\dddef \{
            \textit{while}\
            (n > 0)\
            \textit{do}\
            s := s + n\ ;\
            n := n - 1\
            \textit{end}\
            \}.
    $$
    Given an initial value of variables $n$ and $s$, program $\textit{WP}$ computes the sum from $n$ to $1$ stored in the variable $s$. 
    The \PLK\ structure $\Kr_W = (\Wd_W, \trans_W, \I_W)$ of while programs satisfies that $\Wd_W = S_\fodl$ and $\I_W = I_\fodl$. 
    $\trans_W$ describes the transitional behaviours of while programs, captured by the operational semantics of $\Prog_W$ (see Table~\ref{table:An Example of Inference Rules for Program Behaviours} in Section~\ref{section:A Cyclic Proof System for LDL}). 
    $\trans_W$ coincides with $\to_\fodl$ on atomic programs. 
    For example, a relation $w\xrightarrow{x := x + 1/\ter}w[x\mapsto w(x) + 1]$ is on $\Kr_W$ iff a relation $w\xrightarrow{x := x + 1}w[x\mapsto w(x) + 1]$ is on $\Krp_\fodl$. 
    
    We instantiate $\Fmla$ by the arithmetic first-order formulas in integer domain $\mbb{Z}$ (Section~\ref{section:PDL FODL}), namely $\Fmla_\afo$.  

    \ifx
    and each relation $w\xrightarrow{\alpha/\alpha'}w'\in \trans_W$ iff 
    there is an assignment $x := e$ such that $\alpha = x := e\seq \alpha'$, 
    
    and that each relation 
    a world $w\in \Wd_\fodl$ is a mapping $w: \Var_\fodl\to \mbb{Z}$ from variables to integers. 
    As an example, in $\Krp_\fodl$, there is a relation 
    $w\xrightarrow{x := x + 1/\ter}w[x\mapsto w(x) + 1]$, where 
    $w[x\mapsto v]$ is a mapping that only differs from $w$ on mapping $x$ to value $v$.  
    The formulas in while programs namely $\Fmla_W$ are the usual arithmetic first-order formulas in integer domain $\mbb{Z}$. 
    \fi
\end{example}


\begin{definition}[Execution Paths]
\label{def:Execution Path}
    An ``execution path'' on $\Kr$ is a finite sequence of 
    relations on $\trans$: $w_1\xrightarrow{\alpha_1/\beta_1}...\xrightarrow{\alpha_{n}/\beta_{n}}w_{n+1}$ ($n\ge 0$) satisfying that 
    $\beta_n\in \{\ter\}$, and $\beta_i = \alpha_{i+1} \notin\{\ter\}$ for all $1\le i < n$. 
\end{definition}

In Definition~\ref{def:Execution Path}, the execution path is sometimes simply written as a sequence of worlds: $w_1...w_{n+1}$. 
When $n = 0$, the execution path is a single world $w_1$ (without any relations on $\trans$). 

Given a path $tr$, we often use $tr_b$ and $tr_e$ to denote its first and last element (if there is). 
For two paths $tr_1 \dddef w_1...w_n$ and $tr_2\dddef w'_1w'_2...w'_m...$ ($n, m\ge 0$), $tr_1$ is finite. The \emph{concatenation} $tr_1\cdot tr_2$ is defined as the path: $w_1...w_nw'_2...w'_m...$, if $w_n = w'_1$ holds. 
We use relation $tr_1\suf tr_2$ to represent that $tr_1$ is a suffix of $tr_2$. 
Write $tr_1\psuf tr_2$ if $tr_1$ is a proper suffix of $tr_2$.

\begin{definition}[Semantics of $\DLp$ Formulas]
\label{def:Semantics of dynamic logical formulas}
Given a $\DLp$ formula $\phi$, 
the satisfaction of $\phi$ by a world $w\in \Wd$ under $\Kr$, denoted by $\Kr, w\models \phi$, is inductively defined as follows:
\begin{enumerate}
    \item $\Kr, w\models F$ where $F\in \Fmla$, if $w\in \mcl{I}(F)$;
    \item $\Kr, w\models \neg \phi$, if $\Kr, w\not\models \phi$;
    \item $\Kr, w\models \phi\wedge \psi$, if $\Kr, w\models \phi$ and $\Kr, w\models \psi$;
    \item $\Kr, w\models [\alpha]\phi$, if for all execution paths of the form: $w\xrightarrow{\alpha/\cdot}...\xrightarrow{\cdot/\ter}w'$ for some $w'\in \Wd$, 
    $\Kr, w'\models \phi$. 
\end{enumerate}
\end{definition}

According to the definition of operator $\la\cdot \ra$, 
its semantics is defined such that
$\Kr, w\models \la\alpha\ra\phi$, if there exists an execution path of the form $w\xrightarrow{\alpha/\cdot}...\xrightarrow{\cdot/\ter}w'$ for some $w'\in \Wd$ such that $\Kr, w'\models \phi$. 

A $\DLp$ formula $\phi$ is called \emph{valid} w.r.t. $\Kr$, denoted by $\Kr\models \phi$ (or simply $\models \phi$), if $\Kr, w\models \phi$ for all  $w\in \Wd$. 

Compared to the semantics of PDL (Section~\ref{section:PDL FODL}), where to capture the semantics of a regular program one only has to record the beginning and ending worlds, we have to record the whole execution path from the beginning to the ending node.
This is because the semantics of a program in $\DLp$ is operational, not denotational defined according to its syntactic structures.

\begin{example}[$\DLp$ Specifications]
\label{example:DLp specifications}
    A property of program $\textit{WP}$ (Example~\ref{example:While program})
    is described as the following formula
    $$
    (n \ge 0 \wedge n = N\wedge s = 0) \to [\textit{WP}] (s = ((N+1)N)/2), 
    $$
    which means that 
    given an initial condition of $n$ and $s$, after the execution of $\textit{WP}$, $s$ equals to $((N+1)N)/2$, which is the sum of $1 + 2 + ... + N$, with $N$ a free variable in $\mbb{Z}$.  
    We prove an equivalent labeled version of this formula in $\DLp$ in Section~\ref{section:Case Study}. 
    \ifx
    Recall that $\textit{WP}$ is defined as:
        $$
         \begin{aligned}
            \textit{WP}
             \dddef&
            \{
            \textit{while}\
            (n > 0)\
            \textit{do}\
            s := s + n\ ;\
            n := n - 1\
            \textit{end}\
            \}.
         \end{aligned}
         $$
    \fi
\end{example}

\ifx
In this paper, our current construction of the proof system of $\DLp$ (Section~\ref{section:A Cyclic Proof System for LDL}) confines us to only consider a certain type of programs whose behaviours satisfy the following property. 
\fi

\ifx
In this paper,  proving the soundness of the cyclic proof system $\pfDLp$ (Section~\ref{section:Proof of Theorem - theo:Soundness of A Cyclic Preproof}) requires to confine the discussed programs  to a certain type, which satisfy the following property. 

\begin{definition}[Termination Finiteness]
\label{def:Program Properties}
In \PLK\ structure $\Kr$,  for a world $w\in \Wd$ and a program $\alpha\in \Prog$, 
        there is only a finite number of minimum execution paths starting from relations of the form: $w\xrightarrow{\alpha/\cdot}\cdot$. 
\ifx
In \PLK\ structure $\Kr$,
we assume the following properties:
    \begin{enumerate}
        \item\label{item:Branching Finiteness} Branching Finiteness. For a world $w\in \Wd$ and a program $\alpha\in \Prog$, there exists only a finite number of transitions of the form: $w\xrightarrow{\alpha/\cdot}\cdot$.

        \item\label{item:Termination Finiteness} Termination Finiteness. 
        For a world $w\in \Wd$ and a program $\alpha\in \Prog$, 
        there is only a finite number of minimum execution paths starting from transitions of the form: $w\xrightarrow{\alpha/\cdot}\cdot$.   
    \end{enumerate}
\fi
\end{definition}
\fi

\ifx
Branching finiteness comes from the restrictions made for the rules $\pfPT$ of system $\pfDLp$ in Definition~\ref{def:Matching Proof System}. 
Termination finiteness is directly required for proving the soundness of the cyclic proof system $\pfDLp$ (Section~\ref{section:Proof of Theorem - theo:Soundness of A Cyclic Preproof}). 

The currently discussed programs restricted by Definition~\ref{def:Program Properties} are actually a rich set, 
including, for example, all deterministic programs (i.e., for a program there is only one transition from a world) and programs with finite world spaces (i.e., for a program, from each world, only a finite number of worlds can be reached). 
\textit{While} programs shown in Example~\ref{example:While program} are deterministic programs. 
\fi

\subsection{Labeled $\DLp$ Formulas}

\begin{definition}[Labels \& Label Mappings]
\label{def:Labels and Label Mappings}
In \DLp, we assume two pre-defined sets $\Conf$ and $\LM$. 
$\Conf$ is a set of ``labels''. $\LM\subseteq \Conf\to\Wd$ is a set of label mappings. 
Each mapping $\lm\in \LM$ maps a label of $\Conf$ to a world of set $\Wd$. 
\end{definition}

Labels usually denote the explicit data structures that capture program configurations, for example, storage, heaps, substitutions, etc. 
Label mappings associate labels with the worlds, acting as the semantic functions of the labels.

\begin{example}[An Instantiation of Labels]
\label{example:Labels and Formulas of labeled Sequents}
    In while programs, we consider a type of labels namely $\Conf_W$ that capture the meaning of the program configurations of the form: 
    $$
    \begin{aligned}
        \{x_1\mapsto e_1,...,x_n\mapsto e_n\}\mbox{ ($n\ge 0$)}
    \end{aligned}
    $$
    where 
    each variable $x_i\in \Var_W$ stores a unique value of arithmetic expression $e_i$ ($1\le i\le n$). 
    To make it simple, we restrict that variables $x_1,...x_n$ must appear in the discussed programs and any free variable in $e_1,..., e_n$ cannot be any of $x_1,...,x_n$. 
    For any expression $e$, $\sigma(e)$ returns an expression by replacing each free variable $x_i$ in $e$ with its expression $e_i$ in $\sigma$. 
    A ``configuration update'' $\sigma^x_e$ returns a configuration that stores variable $x$ as a value of expression $\sigma(e)$, while storing other variables as the same value as $\sigma$.

  For example, in program $\WP$ (Example~\ref{example:While program}), $\{n \mapsto N, s\mapsto 0\}$ can be a configuration that maps $n$ to value $N$ (as a free variable) and $s$ to $0$. 
    
    \ifx
    For a non-dynamic formula $\phi\in \Fmla_W$, 
    a labeled formula $\sigma : \phi$ can be interpreted as $\sigma(\phi)$, which returns the formula by substituting each free variable $x$ of $\phi$ by its value in $\sigma$. 
    For example, $\{n \mapsto 5, s\mapsto 0\}$ denotes a configuration that maps $n$ to $5$ and $s$ to $0$. 
    For formula $n > 0$, 
    $\{n \mapsto 5, s\mapsto 0\}(n > 0)$ returns the formula $5 > 0$ by substituting $n$ with its value $5$ in $\{n \mapsto 5, s\mapsto 0\}$. 
    \fi

    \ifx
    A configuration $\sigma_\WP$ of $\textit{WP}$ is of the form: $\{x_1\mapsto e_1, x_2\mapsto e_2,...,x_n\mapsto e_n\}$ ($n\ge 0$) as a variable storage that maps a variable $x_i$ to a unique value of arithmetic expression $e_i$ ($1\le i\le n$).   
    For example, $\{n \mapsto 5, s\mapsto 0\}$ denotes a configuration that maps $n$ to $5$ and $s$ to $0$. 
    Formulas in this particular domain are first-order arithmetical formulas over integer numbers, with propositions are those ``closed formulas'' without free variables. 
    The interpretation $\app_{\textit{WP}}(\sigma, \phi)$ is defined as a a substitution that replaces each free variable $x$ of formula $\phi$ with its value-expression stored in $\sigma$. 
    For instance, 
    we have $\app_{\textit{WP}}(\{n\mapsto 5, s\mapsto 0\}, n > 0) = (5 > 0)$, which is true in the theory of integer numbers. 
    An evaluation $\rho_\WP$ maps each variable to an integer number. 
    $\rho_\WP(t)$ given a term $t$ is defined in the usual sense that it replaces each free variable $x$ of term $t$ with an integer number $\rho_\WP(x)$. 
    \fi
\end{example}

\begin{example}[An Instantiation of Label Mappings]
\label{example:Label Mappings}
    In while programs, 
     we consider a set $\LM_W$ of label mappings. 
     $\LM_W\subseteq \Conf_W\to \Wd_W$. 
     Each label mapping in $\LM_W$ is associated to a world, denoted by $\lm_w$ for some $w\in \Wd_W$. 
     Given a configuration $\sigma$ that captures the meaning of $\{x_1\mapsto e_1,...,x_n\mapsto e_n\}$ ($n\ge 1$), $\lm_w(\sigma)$ is defined as a world such that 
     \begin{enumerate}[(1)]
         \item $\lm_w(\sigma)(x_i) = w(e_i)$ for each $x_i$ ($1\le i\le n$); 
         \item $\lm_w(\sigma)(y) = w(y)$ for other variable $y\in \Var_W$.  
     \end{enumerate}
     Where as explained in Section~\ref{section:PDL FODL}, $w(e_i)$ returns a value by substituting each free occurrences of variable $x$ of $e_i$ with the value $w(x)$. 
     
     For example, let $w$ be a world with $w(N) = 5$, then we have $\lm_w(\{n\mapsto N, s\mapsto 0\})(n) = w(N) = 5$,  $\lm_w(\{n\mapsto N, s\mapsto 0\})(s) = 0$, and  $\lm_w(\{n\mapsto N, s\mapsto 0\})(y) = w(y)$ for any other variable $y\notin \{n, s\}$.  
\end{example}

\begin{definition}[Labeled $\DLp$ Formulas]
\label{def:labeled DLp Formulas}
    A ``labeled formula'' in $\DLp$ belongs to one of the following types of formulas defined as follows:
    $$
    \phi \dddef \sigma : \psi\ |\ (\alpha, \sigma)\trans(\alpha', \sigma')\ |\ \sigma\termi \alpha, 
    $$
    where $\sigma, \sigma'\in \Conf$, $\alpha, \alpha'\in \Prog$, $\psi\in \DLF$. 
\end{definition}
We use $\DLpF$, $\PT$ and $\PTer$ to represent the sets of labeled formulas of the forms: 
$\sigma : \psi$, $(\alpha, \sigma)\trans(\alpha', \sigma')$ and $\sigma\termi \alpha$ respectively. We often use $\tau$ to represent a labeled formula in $\DLpF\cup \PT\cup \PTer$. 

In $\DLp$, relation $(\alpha, \sigma)\trans(\alpha', \sigma')$ is called a \emph{program transition}, which indicates an execution from a so-called \emph{program state} $(\alpha, \sigma)$ to another program state $(\alpha', \sigma')$. 
Relation $\sigma\termi \alpha$ is called a \emph{program termination}, which describes the termination of a program $\alpha$ under a label $\sigma$.

\ifx
We denote the set of all labeled formulas of the form $\sigma : \psi$ as $\DLpF$. 

In $\DLp$, relation $(\alpha, \sigma)\trans(\alpha', \sigma')$ is called a \emph{program transition}, which indicates an execution from a so-called \emph{program state} $(\alpha, \sigma)$ to another program state $(\alpha', \sigma')$. 
Relation $\sigma\termi \alpha$ is called a \emph{program termination}, which describes the termination of a program $\alpha$ under a label $\sigma$.
We use $\PT$ and $\PTer$ to represent the sets of all program transitions and terminations respectively. 
\fi

\ifx
In Definition~\ref{def:labeled DLp Formulas}, from the program-labeled relations defined in Definition~\ref{def:Program-labeled Kripke Structure}, we introduce \emph{program transitions}. symbolic executions of programs as a type of abstract transitions on labels and program terminations. 
A \emph{program transition} is a relation of the form $\sigma \xrightarrow{\alpha/\alpha'}\sigma'$ (also written as $(\alpha, \sigma)\trans (\alpha', \sigma')$ below) between labels and labeled by a program pair, with $\sigma, \sigma'\in \Conf$ and $\alpha, \alpha'\in \Prog$. 
Call pair $(\alpha, \sigma)$ a \emph{program state}. 
We use $\PT$ to represent the set of all program transitions. 
A \emph{program termination} is a relation of the form $\sigma\termi \alpha$ between a label and a program, where $\sigma\in \Conf$ and $\alpha\in \Prog$. 
The set of all program terminations is denoted by $\PTer$. 
\fi

\ifx
\begin{definition}[Substitution of labeled Dynamic Formulas]
A ``substitution function'' $\eta(\eta_p, \eta_f, \eta_l) : \DLpF\to \DLpF$, where $\eta_p : \Prog\to \Prog$, $\eta_f: \Fmla\to \Fmla$, and $\eta_l: \Conf\to \Conf$ are three functions on $\Prog, \Fmla$ and $\Conf$ respectively, is a point-wise function defined as follows based on the structure of labeled dynamic formulas: 
\begin{enumerate}
    \item $\eta(\eta_p, \eta_f, \eta_l)(\eta(F))\dddef \eta_f(F)$, if $F\in \Fmla$;
    \item $\eta(\eta_p, \eta_f, \eta_l)(\neg\phi)\dddef \neg \eta(\eta_p, \eta_f, \eta_l)(\phi)$;
    \item $\eta(\eta_p, \eta_f, \eta_l)(\phi\wedge \psi)\dddef \eta(\eta_p, \eta_f, \eta_l)(\phi)\wedge \eta(\eta_p, \eta_f, \eta_l)(\psi)$;
    \item $\eta(\eta_p, \eta_f, \eta_l)([\alpha]\phi)\dddef [\eta_p(\alpha)]\eta(\eta_p, \eta_f, \eta_l)(\phi)$
    \item $\eta(\eta_p, \eta_f, \eta_l)(\sigma : \phi)\dddef \eta_l(\sigma) : \eta(\eta_p, \eta_f, \eta_l)(\phi)$. 
\end{enumerate}

A substitution function $\eta : \DLpF\to \DLpF$ is called a ``substitution'', if it satisfies that 
for any mapping $\lm : \Conf\to \Wd$, there exists a mapping $\lm'$ (which is only determined only by $\lm$ and $\eta$) such that for any labeled formula $\tau\in \DLpF$, 
$\lm\models \eta(\tau)$ iff $\lm'\models \tau$. 
\end{definition}
\fi

\begin{definition}[Substitution of Labels]
\label{def:Substitution of Labels}
A ``substitution'' $\eta : \Conf\to \Conf$ is a function on $\Conf$ satisfying that 
    for any label mapping $\lm\in \LM$, there exists a label mapping $\lm'(\lm, \eta)$ (determined only by $\lm$ and $\eta$) such that 
    $\lm'(\sigma) = \lm(\eta(\sigma))$ for all labels $\sigma\in \Conf$. 
\end{definition}

Definition~\ref{def:Substitution of Labels} is used in the rule $(\Sub)$ (Table~\ref{table:General Rules for LDL}) and in the proof of soundness of rules $\pfLDLp$ and the cyclic proof system of $\DLp$.

\ifx
\begin{definition}[Substitutions of labeled $\DLp$ Formulas]
    Let $\eta(f_1, f_2) : \DLpF\to \DLpF$ be a pair-wised function defined as follows: 
    for a formula $\sigma : \phi\in \DLpF$, 
    $$\eta(f_1,f_2)(\sigma : \phi)\dddef f_1(\sigma) : f_2(\phi), $$
    where $f_1 : \Conf\to \Conf$ and $f_2 : \DLF\to \DLF$ are two functions. 
    $\eta$ is called a ``subsitution'' if for any mapping $\lm : \Conf\to \Wd$, there exists a mapping $\lm'$ (which is only determined only by $\lm$ and $\eta$) such that 
    $\lm'(\psi) = \lm(\eta(\psi))$ for all formulas $\psi\in \DLpF$.
\end{definition}
\fi

\begin{definition}[Semantics of Labeled $\DLp$ Formulas]
\label{definition:Semantics of labeled Dlp Formulas}
    Given a label mapping $\lm\in \LM$ and a labeled formula $\tau\in \DLpF\cup\PT\cup\PTer$, 
    the satisfaction relation $\Kr, \LM, \lm\models \tau$ of a formula $\tau$ by $\Kr$, $\LM$ and $\lm$ (simply $\lm\models \tau$) is defined as follows according to the different cases of $\tau$: 
    \begin{enumerate}
        \item $\Kr, \LM, \lm\models \sigma : \phi$, if $\Kr, \lm(\sigma)\models \phi$;
        \item $\Kr, \LM, \lm\models (\alpha, \sigma) \trans (\alpha', \sigma')$, if $\lm(\sigma)\xrightarrow{\alpha/\alpha'}\lm(\sigma')$ is a relation on $\Kr$;
        \item $\Kr, \LM, \lm\models \sigma \termi \alpha$, if there exists an execution path 
$\lm(\sigma)\xrightarrow{\alpha/\cdot}...\xrightarrow{\cdot/\ter}w$ on $\Kr$ for some world $w\in \Wd$. 
    \end{enumerate}
\end{definition}

A formula $\tau\in \DLpF\cup \PT\cup \PTer$ is \emph{valid}, denoted by $\Kr\models \tau$ (or simply $\models \tau$), if $\Kr, \LM, \lm\models \tau$ for all $\lm\in \LM$. 

\section{A Cyclic Proof System for $\DLp$}
\label{section:A Cyclic Proof System for LDL}

We propose a cyclic proof system for $\DLp$. 
We firstly propose a labeled proof system $\pfDLp$ to support reasoning based on operational semantics (Section~\ref{section:A Proof System for LDL}). 
Then we construct a cyclic proof structure for system $\pfDLp$, which support deriving infinite proof trees under certain conditions (Section~\ref{section:Construction of A Cyclic Preproof Structure}). 
Section~\ref{section:labeled Sequent Calculus} and~\ref{section:Proof Preproof Cyclic Preproof} introduce the notions of labeled sequent calculus and cyclic proof respectively. 

\subsection{Labeled Sequent Calculus}
\label{section:labeled Sequent Calculus}
\ifx
We assume a set $\Conf$ of labels as a \emph{parameter} of $\DLp$. A label mapping $\lm: \Conf\to \Wd$ maps each label of $\Conf$ to a world of set $\Wd$. 
Denote the set of all label mappings as $\LM$. 
A \emph{labeled $\DLp$ formula} is of the form $\sigma : \phi$, where $\sigma\in \Conf$ and $\phi\in \DLF$. 
Denote the set of all labeled $\DLp$ formulas as $\DLpF$. 

From program-labeled relations defined in Definition~\ref{def:Program-labeled Kripke Structure} we introduce symbolic executions of programs as a type of abstract transitions on labels and program terminations. 
A \emph{program transition} is a relation of the form $\sigma \xrightarrow{\alpha/\alpha'}\sigma'$ (also written as $(\alpha, \sigma)\trans (\alpha', \sigma')$ below) between labels and labeled by a program pair, with $\sigma, \sigma'\in \Conf$ and $\alpha, \alpha'\in \Prog$. 
Call pair $(\alpha, \sigma)$ a \emph{program state}. 
We use $\PT$ to represent the set of all program transitions. 
A \emph{program termination} is a relation of the form $\sigma\termi \alpha$ between a label and a program, where $\sigma\in \Conf$ and $\alpha\in \Prog$. 
The set of all program terminations is denoted by $\PTer$. 
\fi

A \emph{sequent} is a logical argumentation of the form: 
$\Gamma\Rightarrow \Delta,$ where $\Gamma$ and $\Delta$ are finite multi-sets of formulas, called the \emph{left side} and the \emph{right side} of the sequent respectively. 
We use dot $\cdot$ to express $\Gamma$ or $\Delta$ when they are empty sets. 
Intuitively, a sequent $\Gamma\Rightarrow \Delta$ means that if all formulas in $\Gamma$ hold, then one of formulas in $\Delta$ holds. 
We use $\nu$ to represent a sequent. 

A \emph{labeled sequent} is a sequent in which each formula is a labeled formula in $\DLpF\cup \PT\cup \PTer$. 

According to the meaning of a sequent above, 
a labeled sequent $\Gamma\Rightarrow \Delta$ is \emph{valid}, if for every $\lm\in \LM$, 
$\lm\models \tau$ for all $\tau\in \Gamma$ implies $\lm\models \tau'$ for some $\tau'\in \Delta$. 
For a multi-set $\Gamma$ of formulas, 
we write $\lm\models \Gamma$ to mean that $\lm\models \tau$ for all $\tau\in \Gamma$.

\ifx
\begin{definition}[Semantics of Formulas in labeled Sequents]
    Given a labeled sequent $\nu$ and a label mapping $\lm\in \LM$, 
    the satisfaction relation $\lm\models \tau$ of a formula $\tau$ in $\nu$ under $\Kr$ is defined as follows according to the different cases of $\tau$: 
    \begin{enumerate}
        \item $\Kr, \lm\models \sigma : \phi$, if $\Kr, \lm(\sigma)\models \phi$;
        \item $\Kr, \lm\models \sigma \xrightarrow{\alpha/\alpha'}\sigma'$, if $\lm(\sigma)\xrightarrow{\alpha/\alpha'}\lm(\sigma')$ is a relation on $\Kr$;
        \item $\Kr, \lm\models \sigma \termi \alpha$, if there exists an execution path 
$\lm(\sigma)\xrightarrow{\alpha/\cdot}...\xrightarrow{\cdot/\ter}w$ on $\Kr$ for some world $w\in \Wd$. 
    \end{enumerate}
\end{definition}

A formula $\tau$ in a labeled sequent is \emph{valid}, denoted by $\Kr\models \tau$ (or simply $\models \tau$), if $\Kr, \lm\models \tau$ for all $\lm\in \LM$. 
According to the meaning of a sequent above, 
a labeled sequent $\Gamma\Rightarrow \Delta$ is \emph{valid}, if for every $\lm\in \LM$, 
$\Kr, \lm\models \tau$ for all $\tau\in \Gamma$ implies $\Kr, \lm\models \tau'$ for some $\tau'\in \Delta$. 

For a multi-set $\Gamma$ of formulas, 
we write $\Kr, \lm\models \Gamma$ to mean that $\Kr, \lm\models \tau$ for all $\tau\in \Gamma$. 
\fi

\subsection{Proofs \& Preproofs \& Cyclic Proofs}
\label{section:Proof Preproof Cyclic Preproof}

An \emph{inference rule} is of the form
$
\begin{aligned}
\infer[]
{
    \nu
}
{
    \nu_1
    &
    ...
    &
    \nu_n
}
,\end{aligned}$
where each of $\nu,\nu_i$ ($1\le i\le n$) is also called a \emph{node}. 
Each of $\nu_1,...,\nu_n$ is called a \emph{premise}, and $\nu$ is called the \emph{conclusion}, of the rule. 
The semantics of the rule is that 
the validity of sequents $\nu_1, ..., \nu_n$ implies the validity of sequent $\nu$. 
A formula $\tau$ of node $\nu$ is called the \emph{target formula} if 
except $\tau$ other formulas are kept unchanged in the derivation from $\nu$ to some node $\nu_i$ ($1\le i\le n$). 
And in this case other formulas except $\tau$ in node $\nu$ are called the \emph{context} of $\nu$.  
A formula pair $(\tau_1, \tau_2)$ with $\tau_1$ in $\nu$ and $\tau_2$ in some $\nu_i$ is called a \emph{conclusion-premise} (CP) pair of the derivation from $\nu$ to $\nu_i$. 

In this paper, we use a double-lined inference form:
$$
\begin{aligned}
   \infer=[]
{\phi}
{\phi_1 & ... & \phi_n} 
\end{aligned}
$$
to represent both rules
\begin{center}
$
\begin{aligned}
    \infer[]
{
    \Gamma\Rightarrow \phi, \Delta
}
{
    \Gamma\Rightarrow \phi_1, \Delta
    &
    ...
    &
    \Gamma\Rightarrow \phi_n, \Delta
}
\mbox{ and }
\infer[]
{
    \Gamma, \phi\Rightarrow \Delta
}
{
    \Gamma, \phi_1\Rightarrow  \Delta
    &
    ...
    &
    \Gamma, \phi_n\Rightarrow \Delta
}, 
\end{aligned}
$
\end{center}
provided any context $\Gamma$ and $\Delta$. 



A \emph{proof tree} (or \emph{proof}) is a finite tree structure formed by making derivations backward from a root node.  
In a proof tree, a node is called \emph{terminal} if it is the conclusion of an axiom. 

In the cyclic proof approach (cf.~\cite{Brotherston07}), 
a \emph{preproof} is an infinite proof tree (i.e. some of its derivations contain infinitely many nodes) 
in which there exist non-terminal leaf nodes, called \emph{buds}. 
Each bud is identical to one of its ancestors in the tree. 
A bud and one of its identical ancestors together is called a \emph{back-link}. 
A \emph{derivation path} in a preproof is an infinite sequence of nodes $\nu_1\nu_2...\nu_m...$ ($m\ge 1$) starting from the root node $\nu_1$, 
where each node pair $(\nu_i, \nu_{i+1})$ ($i\ge 1$) is a CP pair of a rule. 
A proof tree is \emph{cyclic}, if it is a preproof in which there exists a ``progressive derivation trace'', whose definition depends on specific logic theories (see Definition~\ref{def:Progressive Step/Progressive Derivation Trace} later for $\DLp$), over every derivation path. 

A \emph{proof system} $Pr$ consists of a finite set of inference rules. 
We say that a node $\nu$ can be derived from $Pr$, denoted by $Pr\vdash \nu$, if a proof tree can be constructed (with $\nu$ the root node) by applying the rules in $Pr$, which satisfies 
either (1) all of its leaf nodes terminate or (2) it is a cyclic proof. 

\subsection{A Proof System for $\DLp$}
\label{section:A Proof System for LDL}

The labeled proof system $\pfDLp$ for $\DLp$ consists of two parts: a finite set $\pfLDLp$ of \emph{kernel rules} for deriving $\DLp$ formulas, as listed in Table~\ref{table:General Rules for LDL}, and a finite set $\pfOper$ of the rules, as a parameter of $\DLp$, for capturing the operational semantics of the programs in $\Prog$. 
The rules in $\pfLDLp$ do not rely on the explicit structures of the programs in $\Prog$, but  
depend on the derivations of program transitions according to $\pfOper$ as their side-conditions. 
$\pfDLp$ provides a universal logical framework but modulo different program theories: i.e., $\pfOper$. 

The content of $\pfOper$ depends on the explicit structures of the programs in $\Prog$ and can vary from case to case. 
Through the rules in $\pfOper$, in the system $\pfDLp$ we can derive the program transitions $\PT$ and terminations $\PTer$ respectively. 
However, to coincide with the \PLK\ structure $\Kr(\Prog, \Fmla)$ as well as to fulfill the soundness and completeness of $\DLp$, we need to make the following assumptions on $\pfOper$ as described in the next definition.

\ifx
The functionalities of the sets $\pfPT$ and $\pfPTer$ are for deriving program transitions $\PT$ and terminations $\PTer$ respectively. 
Program terminations are required when we build the cyclic proof structure in system $\pfDLp$ in Section~\ref{section:Construction of A Cyclic Preproof Structure}. 
To coincide with the PL Kripke structure $\Kr(\Prog, \Fmla)$, sets $\pfPT$ and $\pfPTer$ are assumed to satisfy the conditions in the following definition. 
\fi

\ifx
Each rule in $\pfPT$ (resp. $\pfPTer$) has a restricted form: 
$\begin{aligned}
    \infer[]
    {\Gamma\Rightarrow \tau, \Delta}
    {\Gamma_1\Rightarrow\Delta_1
    &...&
    \Gamma_n\Rightarrow\Delta_n} 
\end{aligned},$
where $n \ge 0$, $\tau\in \PT$ (resp. $\tau\in \PTer$), $\tau$ is the target formula of the rule. 


\fi


\ifx
$\pfPT$ and $\pfPTer$ are assumed to be \emph{sound and complete} for the operational semantics of $\Prog$, in the sense that (a) for every program transition $\sigma\xrightarrow{\alpha/\alpha'}\sigma'$ in $\PT$ derived according to system $\pfDLp$, $\lm(\sigma)\xrightarrow{\alpha/\alpha'}\lm(\sigma')$ is a relation on $K$ for any $\lm\in \LM$; and (b) every relation on $K$ can be derived in system $\pfDLp$.  
\fi

\ifx
For sets $\pfPT$ and $\pfPTer$ to capture the operational semantics of $\Prog$, 
assumptions should be made to guarrantee the consistency between system $\pfDLp$ and \PLK\ structure $\Kr$. 
In other words, system $\pfDLp$ should exactly derive all relations on $\Kr$. 
\fi

\ifx
To guarantee the soundness of system $\pfDLp$ (Theorem~\ref{theo:soundness of rules for LDL}), $\pfPT$ is assumed to be in an appropriate form so that system $\pfDLp$ can completely capture the program behaviours indicated by \PLK\ structure $\Kr$. Besides, it also needs to satisfy a finiteness property so that  the derivations of system $\pfDLp$ are maintained finite. 
\fi



\begin{definition}[Assumptions on Set $\pfOper$ ($\pfDLp$)]
\label{def:Matching Proof System}
    The parameter $\pfOper$ (as a part of $\pfDLp$) satisfies that
    \begin{enumerate}
    \ifx
        \item \label{item:assump 1} \textit{Restricted forms} of $\pfPT$ (resp. $\pfPTer$). Each rule in $\pfPT$ (resp. $\pfPTer$) has a form of 
$$\begin{aligned}
    \infer[]
    {\Gamma\Rightarrow \tau, \Delta}
    {\Gamma_1\Rightarrow\Delta_1
    &...&
    \Gamma_n\Rightarrow\Delta_n} 
\end{aligned},$$
where $n \ge 0$, $\tau\in \PT$ (resp. $\tau\in \PTer$), and $\tau$ is the target formula of the rule. 
    \fi
    \item \label{item:coincidence}\textit{Coincidence with $\Kr(\Prog, \Fmla)$}. For any $\lm\in \LM$, $\Gamma$ such that $\lm\models \Gamma$, and for any $\sigma\in \Conf$, if $\lm(\sigma)\xrightarrow{\alpha/\alpha'}w$ is a relation on $\Kr$ for some $\alpha, \alpha'\in \Prog$ and $w\in \Wd$, then there exists a label $\sigma'\in \Conf$ such that $\lm(\sigma') = w$ and $\models(\Gamma\Rightarrow (\alpha, \sigma)\trans (\alpha', \sigma'))$. 
    
    \item \label{item:soundness}\textit{Soundness w.r.t. $\PT$ and $\PTer$}. 
    For any derivation $\pfDLp\vdash (\Gamma\Rightarrow (\alpha, \sigma)\trans(\alpha', \sigma'))$ (resp. $\pfDLp\vdash (\Gamma\Rightarrow \sigma\termi \alpha)$) in system $\pfDLp$, $\models (\Gamma\Rightarrow (\alpha, \sigma)\trans(\alpha', \sigma'))$ (resp. $\models (\Gamma\Rightarrow \sigma\termi \alpha)$. 
    
    \item \label{item:assump 3} \textit{Completeness w.r.t. $\PT$ and $\PTer$}. 
    For any valid sequent $\Gamma\Rightarrow (\alpha, \sigma)\trans(\alpha', \sigma')$ (resp. $\Gamma\Rightarrow \sigma\termi \alpha$), 
    $\pfDLp\vdash (\Gamma\Rightarrow (\alpha, \sigma)\trans(\alpha', \sigma'))$ (resp. $\pfDLp\vdash (\Gamma\Rightarrow \sigma\termi \alpha)$). 
    \ifx
    \item \label{item:assump 4} \textit{Completeness of $\pfPTer$ w.r.t. $\Kr(\Prog, \Fmla)$} (needed ???). For any $\sigma\in \Conf$ and $\lm\in \LM$, if there is a path $\lm(\sigma)\xrightarrow{\alpha/\cdot}...\xrightarrow{\cdot/\ter}w$ for some $w\in \Wd$, then 
    $\pfDLp\vdash (\cdot\Rightarrow \sigma \termi \alpha)$. 
    \fi

    \item \label{item:Simple Conditions}\textit{Simple Conditions}. For any valid sequent $\Gamma\Rightarrow (\alpha, \sigma)\trans(\alpha',\sigma')$, there is a context $\Gamma'$ in which all formulas are non-dynamic ones, such that $\models(\Gamma'\Rightarrow (\alpha, \sigma)\trans(\alpha',\sigma'))$ and $\models (\Gamma\Rightarrow \Gamma')$. 
    \end{enumerate}
    
    \ifx
    System $\pfDLp$ is assumed to satisfy the following properties:
    \begin{enumerate}
        \item \textit{Soundness of $\pfPT$ and $\pfPTer$} . All rules of $\pfPT$ and $\pfPTer$ are sound. 
        
        \item \label{item:assump 4} \textit{Completeness w.r.t. $\Kr$}. For any $\sigma\in \Conf$, $\Gamma$ and $\lm\in \LM$ with $\lm\models \Gamma$, if $\lm(\sigma)\xrightarrow{\alpha/\alpha'}w$ is a relation on $\Kr$ for some $\alpha, \alpha'\in \Prog$ and $w\in \Wd$, then there exists a label $\sigma'\in \Conf$ such that $\lm(\sigma') = w$ and $\pfDLp\vdash (\Gamma\Rightarrow{}\sigma\xrightarrow{\alpha/\alpha'} \sigma')$. 
    \end{enumerate}
    \fi
\end{definition}

Intuitively, 
the coincidence with $\Kr(\Conf, \Fmla)$ says that the relations on $\Kr$ between worlds (, once their starting nodes can be captured by labels,) can be expressed by the program transitions as labeled formulas of $\DLp$.  
It is essential for proving the soundness of the rules in $\pfLDLp$ (see the proof of Theorem~\ref{theo:soundness of rules for LDL} in Appendix~\ref{section:Other Propositions and Proofs}). 
The soundness and completeness w.r.t. $\PT$ and $\PTer$ are required for proving Theorem~\ref{theo:soundness of rules for LDL}, \ref{theo:Soundness of A Cyclic Preproof} and \ref{theo:Completeness of DLp}. 
The soundness condition says that the proof system $\pfOper$ derives no more than the program transitions that are valid on $\Kr$, while the completeness condition says that $\pfOper$ is enough for deriving all those behaviours. 
The assumption ``simple conditions'' is needed in our proof of the conditional completeness of $\DLp$ (Theorem~\ref{theo:Completeness of DLp}). 
It means that the conditions for program transitions do not depend on the behaviours of other programs. 
This is usually the case for the program languages in practice. 
However, this assumption makes $\DLp$ unable to instantiate the dynamic logics with ``rich tests'' (cf.~\cite{Harel00}), where a test can be a dynamic formula itself, for example a test $[x := e]\phi?$. 



\ifx
\begin{definition}[Completeness of $\pfDLp$ w.r.t. $\PT$ and $\PTer$]
\label{def:completeness w.r.t. program transitions and program terminations}
    System $\pfDLp$ is `complete' w.r.t. program transitions $\PT$ (resp. program terminations $\PTer$), if 
for any valid sequent $\Gamma\Rightarrow \sigma \xrightarrow{\alpha/\alpha'}\sigma'$ (resp. valid sequent $\Gamma\Rightarrow \sigma\termi \alpha$), 
$\pfDLp\vdash (\Gamma\Rightarrow \sigma \xrightarrow{\alpha/\alpha'}\sigma')$ (resp. $\pfDLp\vdash (\Gamma\Rightarrow \sigma\termi \alpha)$) holds. 
\end{definition}
\fi

\ifx
The proof system $\pfDLp$ of $\DLp$ 
relies on a pre-defined proof system $P_{\Oper,\Terminate}$ 
for deriving program behaviours $\Oper$ and terminations $\Terminate$. 
$P_{\Oper,\Terminate}$ is \emph{sound} and \emph{complete} w.r.t. $\Oper$ and $\Terminate$ in the sense that 
for any transition $(\alpha_1, \sigma_1)\trans (\alpha_2, \sigma_2)\in \Prop$ and termination $(\alpha, \sigma)\termi\ \in \Prop$, $(\alpha_1, \sigma_1)\trans (\alpha_2, \sigma_2)\in \Oper$ iff 
$\vdash_{P_{\Oper, \Terminate}} \cdot \Rightarrow (\alpha_1, \sigma_1)\trans (\alpha_2, \sigma_2)$, and $(\alpha, \sigma)\termi\ \in \Terminate$ iff $\vdash_{P_{\Oper, \Terminate}}\cdot \Rightarrow (\alpha, \sigma)\termi$. 

Note that in $\pfDLp$, we usually assume that a formula does not contain a program transition or termination. 
\fi

\renewcommand{\arraystretch}{1.5}
\begin{table}[tb]
         \begin{center}
         \noindent\makebox[\textwidth]{%
         \scalebox{1}{
         \begin{tabular}{c}
         \toprule
        $
         \infer[^{(x:=e)}]
         {\Gamma\Rightarrow (x:=e, \sigma)\trans(\ter, \sigma^x_{e}), \Delta}
         {
         }
         $
         \ \
         $
         \infer[^{(;)}]
         {\Gamma\Rightarrow (\alpha_1 ; \alpha_2, \sigma)\trans(\alpha'_1 ; \alpha_2, \sigma'), \Delta}
         {
            \Gamma\Rightarrow (\alpha_1, \sigma)\trans (\alpha'_1, \sigma'), \Delta
         }
         $
         \\
         $
         \infer[^{(;\ter)}]
         {\Gamma\Rightarrow(\alpha_1; \alpha_2, \sigma)\trans(\alpha_2, \sigma'), \Delta}
         {
            \Gamma\Rightarrow (\alpha_1, \sigma)\trans (\ter, \sigma'), \Delta
         }
         $
         \ \ 
         $
         \infer[^{(\textit{ite1})}]
         {\Gamma\Rightarrow (\Wif\ \phi\ \Wthen\ \alpha_1\ \Welse\ \alpha_2\ \Wend, \sigma)\trans (\alpha'_1, \sigma'), \Delta}
         {
         \Gamma\Rightarrow (\alpha_1, \sigma)\trans (\alpha'_1, \sigma'), \Delta
         &
         \Gamma\Rightarrow \sigma : \phi, \Delta
         }
         $
         \\
         $
         \infer[^{(\textit{ite2})}]
         {\Gamma\Rightarrow (\Wif\ \phi\ \Wthen\ \alpha_1\ \Welse\ \alpha_2\ \Wend, \sigma)\trans (\alpha'_2, \sigma'), \Delta}
         {
         \Gamma\Rightarrow (\alpha_2, \sigma)\trans (\alpha'_2, \sigma'), \Delta
         &
         \Gamma\Rightarrow \sigma : \neg \phi, \Delta
         }
         $
         \\
         $
         \infer[^{(\textit{wh1})}]
         {\Gamma\Rightarrow (\textit{while}\ \phi\ \textit{do}\ \alpha\ \textit{end}, \sigma)\trans(\alpha';\ \textit{while}\ \phi\ \textit{do}\ \alpha\ \textit{end}, \sigma'), \Delta}
         {
         \Gamma, \sigma : \phi\Rightarrow (\alpha, \sigma)\trans(\alpha', \sigma'), \Delta
         &
                  \Gamma\Rightarrow \phi : \sigma, \Delta
         }
         $
         \\
         $
         \infer[^{(\textit{wh1}\ter)}]
         {\Gamma\Rightarrow (\textit{while}\ \phi\ \textit{do}\ \alpha\ \textit{end}, \sigma)\trans(\textit{while}\ \phi\ \textit{do}\ \alpha\ \textit{end}, \sigma'), \Delta}
         {
          \Gamma, \sigma : \phi\Rightarrow (\alpha, \sigma)\trans(\ter, \sigma'), \Delta
          &
                  \Gamma\Rightarrow \sigma : \phi, \Delta
         }
         $
         \\
         $
         \infer[^{(\textit{wh2})}]
         {\Gamma\Rightarrow (\textit{while}\ \phi\ \textit{do}\ \alpha\ \textit{end}, \sigma)\trans(\ter, \sigma), \Delta}
         {
         \Gamma\Rightarrow \sigma : \neg\phi, \Delta
         }
         $
         \\
         \bottomrule
         \end{tabular}
              }
              }
          \end{center}
          \caption{Partial Rules of $(\pfOper)_W$ for Program Transitions of While Programs}
          \label{table:An Example of Inference Rules for Program Behaviours}
    \end{table}

\begin{example}[An Instantiation of $\pfOper$]
\label{example:Proof system for program behaviours}
Table~\ref{table:An Example of Inference Rules for Program Behaviours} displays a set of partial rules of $(\pfOper)_W$ for describing the operational semantics of while programs, where we omit the rules for deriving program terminations.  
In Table~\ref{table:An Example of Inference Rules for Program Behaviours}, $\sigma^x_e$ is defined in Example~\ref{example:Labels and Formulas of labeled Sequents}.

\end{example}

In practice, we usually think that the rules $\pfOper$ faithfully commit the operational semantics of the programs. 
So in this manner we simply trust the assumptions \ref{item:coincidence}, \ref{item:soundness}, \ref{item:assump 3} of Definition~\ref{def:Matching Proof System} without proving. 

\ifx REPLACED BY AN ALTERNATIVE STATEMENT AS ABOVE
In practice, we can simply ``trust'' $\pfOper$ (thinking that they faithfully commit the operational semantics of the programs) and directly define the programs' transitional behaviours (i.e. the relation $\trans$ of $\Kr(\Prog, \Fmla)$) according to $\pfOper$. 
In this manner, the coincidence, the soundness and complteness conditions of Definition~\ref{def:Matching Proof System} are automatically met.

\begin{example}
    ddd
\end{example}
\fi

Through the rules in $\pfLDLp$, a labeled $\DLp$ formula can be transformed into proof obligations as non-dynamic formulas, which can then be encoded and verified accordingly through, for example, an SAT/SMT checking procedure. 
The rules for other operators like $\vee$, $\to$ can be derived accordingly using the rules in Table~\ref{table:General Rules for LDL}.   

\ifx
Table~\ref{table:General Rules for LDL} lists the primitive rules of $\pfDLp$.
Through $\pfDLp$, 
a $\DLp$ formula can be transformed into proof obligations as non-dynamic formulas, which can then be encoded and verified accordingly through, for example, an SAT/SMT checking procedure. 
For easy understanding, we give rule $([\alpha]L)$ instead of the version of rule $([\alpha]R)$ for the left-side derivations, which can be derived using $([\alpha]L)$ and rules $(\neg L)$ and $(\neg R)$.  
The rules for other operators like $\vee$, $\to$ can be derived accordingly using the rules in Table~\ref{table:General Rules for LDL}.   
\fi

\renewcommand{\arraystretch}{1.5}
    \begin{table}[tb]
         \begin{center}
         \noindent\makebox[\textwidth]{%
         \scalebox{1.0}{
         \begingroup
         \begin{tabular}{c}
         \toprule
         $
         \begin{aligned}
         \infer[^{1\ ([\alpha]R)}]
         {\Gamma\Rightarrow \sigma : [\alpha]\phi, \Delta}
         {
         \{
         \Gamma\Rightarrow \sigma' : [\alpha']\phi, \Delta
         \}_{(\alpha', \sigma')\in \Phi}
         }
         \end{aligned}
         $,
         \ \ 
         where
         $
         \Phi\dddef \{(\alpha', \sigma')\ |\ \mbox{$\pfDLp\vdash (\Gamma\Rightarrow (\alpha, \sigma) \xrightarrow{}(\alpha', \sigma'), \Delta)$}\}
         $
        \\
        \midrule
        $
         \begin{aligned}
         \infer[^{1\ ([\alpha]L)}]
         {\Gamma, \sigma: [\alpha]\phi\Rightarrow  \Delta}
         {
         \Gamma, \sigma': [\alpha']\phi\Rightarrow  \Delta
         }
         \end{aligned}
         $,
         \ \ 
         if $\pfDLp\vdash (\Gamma\Rightarrow (\alpha, \sigma)\xrightarrow{}(\alpha', \sigma'), \Delta)$
         \ifx
         \\
         \midrule
         $
         \begin{aligned}
         \infer[^{1\ (\la\alpha\ra t)}]
         {\Gamma \Rightarrow \sigma: \la\alpha\ra\phi, \Delta}
         {
         \Gamma\Rightarrow \sigma': \la\alpha'\ra\phi, \Delta
         &
         \Gamma\Rightarrow \sigma\xrightarrow{\alpha/\alpha'}\sigma', \Delta
         &
         \Gamma\Rightarrow \sigma\termi \alpha, \Delta
         }
         \end{aligned}
         $
         \fi
        \\
        \midrule
        $
          \begin{aligned}
        \infer=[^{([\ter])}]
         {\sigma : [\ter]\phi}
         {\sigma : \phi}
          \end{aligned}
         $
         \ \vline\ 
           $
         \begin{aligned}
         \infer[^{2\ (\textit{Ter})}]
         {\Gamma \Rightarrow \Delta}
         {}
         \end{aligned}
         $
        \ \vline\ 
        $
        \begin{aligned}
         \infer[^{3\ (\textit{Sub})}]
         {\Sub(\Gamma)\Rightarrow \Sub(\Delta)}
         {\Gamma\Rightarrow \Delta}
         \end{aligned}
         $
         \ifx
        \ \vline\ 
        $
        \begin{aligned}
         \infer=[^{(\sigma\neg)}]
         {\sigma : (\neg \phi)}
         {\neg(\sigma : \phi)}
        \end{aligned}
         $
         \ \vline\ 
        $
        \begin{aligned}
         \infer=[^{(\sigma\wedge)}]
         { \sigma : (\phi\wedge\psi)}
         {(\sigma : \phi)\wedge (\sigma : \psi)}
         \end{aligned}
         $
         \fi
         \ifx
         \ \vline\ 
         $
        \begin{aligned}
         \infer=[^{(\sigma\wedge 2)}]
         {(\sigma : \phi)\wedge (\sigma : \psi)}
         {\sigma : (\phi\wedge\psi)}
         \end{aligned}
         $
         \fi
        \ \vline\ 
         $
         \begin{aligned}
         \infer[^{(\textit{ax})}]
         {\Gamma, \sigma : \phi\Rightarrow \sigma : \phi, \Delta}
         {}
         \end{aligned}
         $
         \\
         \midrule
         $
         \begin{aligned}
         \infer[^{(\textit{Cut})}]
         {\Gamma\Rightarrow \Delta}
         {\Gamma\Rightarrow \sigma : \phi & 
         \Gamma, \sigma : \phi\Rightarrow \Delta}
         \end{aligned}
         $
        \ \vline\ 
                 $
         \begin{aligned}
         \infer[^{(\textit{WkR})}]
         {\Gamma\Rightarrow \sigma : \phi, \Delta}
         {\Gamma\Rightarrow \Delta}
         \end{aligned}
         $
         \ \vline\ 
        $
         \begin{aligned}
         \infer[^{(\textit{WkL})}]
         {\Gamma, \sigma : \phi\Rightarrow  \Delta}
         {\Gamma \Rightarrow \Delta}
         \end{aligned}
         $
         \ \vline\ 
         $
         \begin{aligned}
         \infer=[^{(\textit{Con})}]
         {\sigma : \phi}
         {\sigma : \phi, \sigma : \phi}
         \end{aligned}
         $
         \\
         \midrule
                   $
        \begin{aligned}
         \infer[^{(\neg R)}]
         {\Gamma\Rightarrow \sigma : \neg \phi, \Delta}
         {\Gamma, \sigma : \phi \Rightarrow \Delta}
        \end{aligned}
         $
         \ \vline\ 
                  $
        \begin{aligned}
         \infer[^{(\neg L)}]
         {\Gamma, \sigma : \neg \phi\Rightarrow \Delta}
         {\Gamma\Rightarrow \sigma : \phi, \Delta}
        \end{aligned}
         $
         \ \vline\ 
                 $
        \begin{aligned}
         \infer[^{(\wedge R)}]
         {\Gamma\Rightarrow \sigma : \phi\wedge \psi, \Delta}
         {\Gamma\Rightarrow \sigma : \phi, \Delta
         &
         \Gamma\Rightarrow \sigma : \psi, \Delta}
         \end{aligned}
         $
         \ \vline\ 
        $
        \begin{aligned}
         \infer[^{(\wedge L)}]
         {\Gamma, \sigma : \phi\wedge\psi\Rightarrow  \Delta}
         {\Gamma, \sigma : \phi, \sigma : \psi\Rightarrow \Delta}
         \end{aligned}
         $
         \\
         \midrule
         \multicolumn{1}{l}{
            \begin{tabular}{l}
            $^{1}$ $\alpha\notin \{\ter\}$.
            $^{2}$ for each $\sigma : \phi\in \Gamma\cup\Delta$, $\phi\in \Fmla$; 
            Sequent $\Gamma\Rightarrow \Delta$ is valid.
            $^{3}$ $\Sub$ is given by Definition~\ref{def:Substitution of Labels}. 
            \end{tabular}
         }
         \\
         \bottomrule
         \end{tabular}
         \endgroup
              }
              }
          \end{center}
          \caption{Rules $\pfLDLp$ for the Proof System of $\DLp$}
          \label{table:General Rules for LDL}
    \end{table}
    
The illustration of each rule in Table~\ref{table:General Rules for LDL} is as follows. 

Rules $([\alpha]R)$ and $([\alpha]L)$ reason about dynamic parts of labeled $\DLp$ formulas. 
Both rules rely on side deductions: ``$\pfDLp\vdash (\Gamma\Rightarrow{} (\alpha, \sigma)\xrightarrow{}(\alpha', \sigma'), \Delta)$'' as sub-proof procedures of program transitions. 
In rule $([\alpha]R)$, 
$\{...\}_{(\alpha', \sigma')\in \Phi}$ represents the collection of premises for all program states $(\alpha', \sigma')\in \Phi$. 
By the finiteness of system $\pfDLp$, set $\Phi$ must be finite (because only a finite number of forms $(\alpha', \sigma')$ can be derived). 
So rule $([\alpha]R)$ only has a finite number of premises. 
When $\Phi$ is empty, the conclusion terminates.  
Compared to rule $([\alpha]R)$, rule $([\alpha]L)$ has only one premise for some program state $(\alpha', \sigma')$. 

\ifx
we prove that there exists a program state $(\alpha', \sigma')$ such that 
$(\alpha, \sigma)\trans (\alpha', \sigma')$, and 
$\la\alpha'\ra\phi$ holds under configuration $\sigma'$.  
At the same time, we need to guarantee that the termination factor $t'$ cannot grow larger than $t$ w.r.t. relation $\prec$. 
Note that in both rules, we assume that there always exists a transition $(\alpha, \sigma)\trans (\alpha', \sigma')$ from program state $(\alpha, \sigma)$ since $\alpha$ is neither $\ter$ nor $\abort$. This is guaranteed by the well-definedness of $\alpha$'s operational semantics, as stated in~\ref{item:Well-definedness} of Definition~\ref{def:Program Properties}. 
\fi

Rule $([\ter])$ deals with the situation when the program is a terminal one $\ter$. 
Its soundness is straightforward by the semantics of $\ter$ in Definition~\ref{def:Program-labeled Kripke Structure}.  

\ifx
In rule $(\textit{Ter})$, 
when sequent $\Gamma\Rightarrow \Delta$ of labeled non-dynamic formulas is valid, it is assumed to be checked through SAT/SMT procedure. 
\fi

Rule $(\textit{Ter})$ indicates that one proof branch terminates when a sequent $\Gamma\Rightarrow\Delta$ is valid in which all labeled formulas are non-dynamic ones. 

\ifx
Rules $([\ter])$, $(\la\ter\ra)$ and $(\sigma[\abort])$ deal with the situations when program $\alpha$ is either a termination $\ter$ or an abortion $\abort$.
Note that there is no rules for formula $\sigma: \la\abort\ra \phi$. 
Intuitively, 
it is not hard to see that $\la \abort\ra\phi$ is false under any configuration because $\abort$ never terminates. 
\fi

Rule $(\Sub)$ describes a specialization process for labeled dynamic formulas. 
For a set $A$ of labeled formulas, $\Sub(A)\dddef \{\Sub(\tau)\ |\ \tau\in A\}$, with $\Sub$ a substitution (Definition~\ref{def:Substitution of Labels}). 
Intuitively, if sequent $\Gamma\Rightarrow \Delta$ is valid, then its one of special cases $\Sub(\Gamma)\Rightarrow \Sub(\Delta)$ is also valid. 
Rule $(\Sub)$ plays an important role in constructing a bud in a cyclic proof structure (Section~\ref{section:Construction of A Cyclic Preproof Structure}). 
See Section~\ref{section:Case Study} for more details.

\ifx DEFINED EARLIER
\begin{definition}[Substitution of Labels]
\label{def:Substitution of Labels}
A `substitution' $\eta : \Conf\to \Conf$ is a function on $\Conf$ satisfying that 
    for any label mapping $\lm\in \LM$, there exists a label mapping $\lm'(\lm, \eta)$ (determined only by $\lm$ and $\eta$) such that 
    $\lm'(\sigma) = \lm(\eta(\sigma))$ for all labels $\sigma\in \Conf$. 
\end{definition}

Definition~\ref{def:Substitution of Labels} will be used in the proof of soundness of rules $\pfLDLp$ and the cyclic proof system of $\DLp$. 
\fi


Rules from $(\textit{ax})$ to $(\wedge L)$ are the ``labeled verions'' of the corresponding rules inherited from traditional first-order logic. 
Their meanings are classical and we omit their discussions here. 

\ifx
Rules from $(\textit{ax})$ to $(\vee)$ are direct from the traditional first-order logic. 
Note that both rules $(\neg L)$ and $(\neg R)$ are required to derive the rules for the left-side target formulas of sequents. 
Rules $(\textit{Sub} L)$ and $(\textit{Sub} R)$ are the versions of substitutions without quantifiers.  
The meanings of these rules are classical and we omit them here. 
\fi

\ifx
Rules $(\textit{Ter})$ and $(\textit{ax})$ declare a termination of a proof branch. 
In rule $(\textit{Ter})$, 
when each formula in $\Gamma$ and $\Delta$ is a formula in $\Fmla$.  
We can conclude the proof branch if 
the proof obligation $\mfr{P}(\Gamma\Rightarrow \Delta)$ is true. 
Rule $(\textit{ax})$ is the `labeled' version of the corresponding rule in traditional propositional logic. 
Rule $(\textit{Cut})$, the labeled version of the corresponding traditional ``cut rule'', provides a mechanism to derive by providing additional lemmas. 
\fi

\ifx
The introduction rule $(\textit{Int})$ for labels is applied when $\phi$ is a formula in $\Fmla$ without any dynamic parts. 
Through this rule we eliminate a configuration $\sigma$ and obtain a formula 
$\app(\sigma, \phi)$ in $\Fmla$. 
\ifx
Rule $(\sigma\textit{Lif})$ lifts a rule in specific domains to a rule labeled with a free configuration $\sigma$ w.r.t. set $\free(\Conf,\Sigma)$ (Definition~\ref{def:free configurations}), where $\Sigma \dddef \Gamma\cup \Delta\cup \{\phi, \psi\}$.  
We define $\sigma : A\dddef \{\sigma : \phi\ |\ \phi\in A\}$ for a multi-set $A$ of formulas. 
Rule $(\sigma\textit{Lif})$ is useful, when a structure-based rule exists, we can simply lift it and make derivations in its labeled forms. 
\fi
Rule $(\sigma \cfeq)$ replaces a configuration $\sigma$ with one that is equivalent as $\sigma$ under evaluations. 
It is critical for the proof system because by looking for a suitable $\sigma'$, 
one can successfully find a back-link, thus closes one branch in a cyclic proof structure (these concepts will be introduced in Section~\ref{section:Construction of A Cyclic Preproof Structure}). 
In Section~\ref{section:Case Studies}, our examples will illustrate how to construct a cyclic proof structure making use of this rule. 
\fi

\ifx
Rules $(\neg\sigma L)$, $(\neg\sigma R)$, $(\wedge\sigma)$ and $(\vee\sigma)$ deal with logical connectives $\neg$, $\wedge$ and $\vee$ in a labeled $\DLp$ formula.
They correspond to the rules in traditional propositional logic (without labels) for $\neg$, $\wedge$ and $\vee$ respectively.  
\fi

\ifx
The lifting rule $(\sigma)$ lifts a rule $(\textit{sp})$ in specific domains to a rule labeled with an irrelevant configuration $\sigma$. 
$\sigma : A$ means the multi-set $\{\sigma : \phi\ |\ \phi \in A\}$, where 
$\phi$ is a $\GDL$ formula in $A$. 
Rule $(\sigma)$ demonstrates the compatibility of $\DLp$ to other existed dynamic-logic proof systems. The soundness of rule $(\sigma)$ is guaranteed by the following proposition, whose proof is directly by the irrelevance of $\sigma$ (Section~\ref{section:Dynamic Logic LDL}). 
\fi

\ifx
\begin{proposition}
    If a configuration $\sigma$ is irrelevant to a $\GDL$ formula $\phi$, 
    then $\models \sigma : \phi$ iff $\models \phi$.  
\end{proposition}
\fi


\begin{theorem}
\label{theo:soundness of rules for LDL}
    Each rule from $\pfLDLp$ in Table~\ref{table:General Rules for LDL} is sound. 
\end{theorem}

Following the above explanations, Theorem~\ref{theo:soundness of rules for LDL} can be proved according to the semantics of labeled $\DLp$ formulas under the assumption of Definition~\ref{def:Matching Proof System}. 
See Appendix~\ref{section:Other Propositions and Proofs} for more details.

\subsection{Construction of a Cyclic Proof Structure for $\DLp$}
\label{section:Construction of A Cyclic Preproof Structure}


We build a cyclic labeled proof system for $\DLp$, in order to recognize and admit potential infinite derivations as the example shown in Section~\ref{section:An Overview}. 
Based on the notion of preproofs (Section~\ref{section:Proof Preproof Cyclic Preproof}), 
we build a cyclic proof structure for system $\pfDLp$, where the key part is to introduce the notion of progressive derivation traces in $\DLp$ (Definition~\ref{def:Progressive Step/Progressive Derivation Trace}). 

\ifx
Cyclic proof (cf.~\cite{Brotherston07}) is a technique to ensure a valid conclusion of a preproof structure (Section~\ref{section:Proof Preproof Cyclic Preproof}), under some certain conditions called ``soundness conditions''. 
We build a so-called ``cyclic preproof'' for system $\pfDLp$, which satisfies certain soundness conditions special for dynamic $\DLp$ formulas. 
Section~\ref{section:Proof of Theorem - theo:Soundness of A Cyclic Preproof} will explain why these cyclic preproof structures are sound.  

Next we firstly introduce the necessary notions of preproofs, derivation paths and progressive derivation traces. 
Then we propose a cyclic preproof structure for system $\pfDLp$. 
\fi

\ifx
In this paper, we apply the cyclic proof approach (cf.~\cite{Brotherston07}) to deal with a potential infinite proof tree in $\DLp$. 
Cyclic proof is a technique to ensure a valid conclusion of a special type of 
infinite proof trees, called `preproofs', under some certain conditions called ``soundness conditions''. 
Below we firstly introduce a preproof structure, 
then based on it we propose a `cyclic' preproof --- a structure that satisfies certain soundness conditions.  
Section~\ref{section:Proof of Theorem - theo:Soundness of A Cyclic Preproof} will explain why cyclic preproof can be sound.  
\fi


Next we first introduce the notion of progressive derivation traces for $\DLp$, then 
we define the cyclic proof structure for $\DLp$ as a special case of the notion already given in Section~\ref{section:Proof Preproof Cyclic Preproof}.

\ifx old def of derivation traces
A \emph{derivation trace} over a derivation path $\mu_1\mu_2...\mu_k\nu_1\nu_2...\nu_m...$ ($k \ge 0, m\ge 1$) is an infinite sequence $\tau_1\tau_2...\tau_m...$ of formulas with each formula $\tau_i$ ($1\le i\le m$) in node $\nu_i$. 
Each CP pair $(\tau_i, \tau_{i+1})$ ($i\ge 1$) of derivation $(\nu_i, \nu_{i+1})$ is either 
a target pair of a rule in $\pfLDLp$, or a CP pair satisfying $\tau_i = \tau_{i+1}$. 
\fi

\begin{definition}[Derivation Traces]
\label{def:Derivation Traces}
    A ``derivation trace'' over a derivation path $\mu_1\mu_2...\mu_k\nu_1\nu_2...\nu_m...$ ($k \ge 0, m\ge 1$) is an infinite sequence $\tau_1\tau_2...\tau_m...$ of formulas with each formula $\tau_i$ ($1\le i\le m$) in node $\nu_i$. 
    Each CP pair $(\tau_i, \tau_{i+1})$ ($i\ge 1$) of derivation $(\nu_i, \nu_{i+1})$ satisfies special conditions as follows according to $(\nu_i, \nu_{i+1})$ being the different instances of rules from $\pfLDLp$:
    \begin{enumerate}
        \item If $(\nu_i, \nu_{i+1})$ is an instance of rule $([\alpha]R)$, $([\alpha]L)$, $([\ter])$, $(\neg R)$, $(\neg L)$, $(\wedge R)$ or $(\wedge L)$, 
        then either 
        $\tau_i$ is the target formula and $\tau_{i+1}$ is its replacement by application of the rule, or $\tau_i = \tau_{i+1}$; 

        \item If $(\nu_i, \nu_{i+1})$ is an instance of rule $(\Sub)$, 
        then $\tau_i = \Sub(\sigma) : \phi$ and $\tau_{i+1} = \sigma : \phi$ for some $\sigma\in \Conf$ and $\phi\in \DLF$;
        
        \item If $(\nu_i, \nu_{i+1})$ is an instance of other rules, 
        then $\tau_i = \tau_{i+1}$. 
    \end{enumerate}
\end{definition}


Below an expression $n :: O$ means that we use name $n$ to denote the object $O$. 

\begin{definition}[Progressive Derivation Traces]
\label{def:Progressive Step/Progressive Derivation Trace}
In a preproof of system $\pfDLp$, given a derivation trace $\tau_1\tau_2...\tau_m...$ over a derivation path $...\nu_1\nu_2...\nu_m...$ ($m\ge 1$) starting from $\tau_1$ in node $\nu_1$, 
a CP pair $(\tau_i, \tau_{i+1})$ ($1\le i\le m$) of derivation $(\nu_i, \nu_{i+1})$ is called a ``progressive step'', if $(\tau_i, \tau_{i+1})$ is the following CP pair 
of an instance of rule $([\alpha]R)$: 
$$
\begin{gathered}
    \infer[^{([\alpha]R)}]
    {\nu_i:: (\Gamma\Rightarrow \tau_i :: (\sigma : [\alpha]\phi), \Delta),}
    {
    ...
    &
    \nu_{i+1}:: (\Gamma\Rightarrow \tau_{i+1} :: (\sigma' : [\alpha']\phi), \Delta)
    &
    ...
    }
\end{gathered}
;
$$
or the following CP pair of an instance of rule $([\alpha]L)$:
$$
\begin{gathered}
    \infer[^{([\alpha]L)}]
    {\nu_i:: (\Gamma, \tau_i :: (\sigma : [\alpha]\phi)\Rightarrow \Delta)}
    {
    \nu_{i+1}:: (\Gamma, \tau_{i+1} :: (\sigma' : [\alpha']\phi)\Rightarrow  \Delta)
    }
\end{gathered}
, 
$$
provided with an additional side deduction $\pfDLp\vdash (\Gamma\Rightarrow \sigma\termi \alpha, \Delta)$.

If a derivation trace has an infinite number of progressive steps, we say that the trace is ``progressive''.
\end{definition}

The additional side condition of the instance of rule $([\alpha]L)$ is the key to prove the corresponding case in 
Lemma~\ref{lemma:infinite descent sequence} (see Appendix~\ref{section:Other Propositions and Proofs}). 


\ifx
\begin{definition}[Cyclic Proof in $\DLp$]
\label{def:Cyclic Preproof}
In system $\pfDLp$, a proof tree is called `cyclic', if it is a preproof in which there exists a progressive trace over every derivation path.
\end{definition}
\fi

Theorem~\ref{theo:soundness of rules for LDL} shows that each rule of $\pfDLp$ is sound. But that does not mean that the proof system $\pfDLp$ is sound, because we need to make sure that each cyclic proof also leads to a valid conclusion. 
The soundness of the system $\pfDLp$ is fully discussed in Section~\ref{section:Proof of Theorem - theo:Soundness of A Cyclic Preproof}. 

\ifx
Theorem~\ref{theo:soundness of rules for LDL} shows that each rule of $\pfDLp$ is sound. But that does not mean that the proof system $\pfDLp$ is sound, because we need to make sure that each cyclic proof also leads to a valid conclusion. 
In Section~\ref{section:Proof of Theorem - theo:Soundness of A Cyclic Preproof}, we analyze and prove the soundness of the system $\pfDLp$ under a restriction on the program behaviours of $\Prog$ (stated as Theorem~\ref{theo:Soundness of A Cyclic Preproof} of Section~\ref{section:Proof of Theorem - theo:Soundness of A Cyclic Preproof}). 
\fi

\ifx
Since $\DLp$ is not a specific logic,
it is impossible to discuss about its decidability, completeness or whether it is cut-free without any restrictions on parameters $\Prog, \Fmla$ and $\Conf$. 
In Section~\ref{section:Proof of Theorem - theo:Soundness of A Cyclic Preproof}, we show that $\DLp$ is complete under a certain condition for the \emph{loop programs} of $\Prog$. 
Our future work will discuss more about these attributes under restrictions. 
\fi

\subsection{Lifting Rules From Dynamic Logic Theories}
\label{section:Lifting Process From Program Domains}
We introduce a technique of lifting the rules from particular dynamic-logic theories, e.g. FODL~\cite{Pratt76}, to the labeled ones in $\DLp$. 
It makes possible for embedding existing  dynamic-logic theories into $\DLp$ without losing their abilities of deriving based on programs' syntactic structures. 
This in turn facilitates deriving $\DLp$ formulas in particular program domains by making use of special inference rules.  
\ifx
This is useful in two aspects: 
firstly, it allows making use of existing rules to facilitate deriving $\DLp$ formulas in particular program domains;
secondly, it makes possible for embedding (at least partial) existing  dynamic-logic theories into $\DLp$, 
without losing their abilities of deriving based on programs' syntactic structures. 
\fi
Below, we propose a lifting process for general inference rules under a certain condition of labels (Proposition~\ref{prop:lifting process 2}). 
One example of the applications of this technique is given in Section~\ref{section:Lifting Process in While Programs}. 

\ifx
In a particular program domains like $\TA_\WP$ (Section~\ref{section:Examples of Term Structures}), we are interested in adapting a special rule 
to a labeled one as a part of extra theory $\mcl{E}$ of $\DLp$.
As pointed out in Section~\ref{section:Summery}, it allows $\DLp$ to make use of existing rules of a particular theory to facilitate derivations of $\DLp$ formulas, thus providing a flexible verification framework in which both symbolic-execution-based reasoning and structural reasoning are possible. 
Below we introduce a process of lifting logical consequences $\phi\to \psi$ and general inference rules under certain conditions. 
Logical consequences are commonly found as axioms in dynamic logics and their variations (cf.~\cite{???}). 
One example of applications of the lifting process will be given later in Section~\ref{section:Lifting Process in While Programs}. 
\fi

We first introduce the concept of \emph{free labels}, as a sufficient condition for the labels to carry out the lifting. 
Then we introduce the lifting process as Proposition~\ref{prop:lifting process 2}. 


\begin{definition}[Effect Equivalence]
\label{def:Same Effects}
    Two worlds $w,w'\in \Wd$ ``have the same effect'' w.r.t a set of unlabeled formulas $A\subseteq \DLF$, denoted by $w=_A w'$, if 
    for any $\phi\in A$, $w\models \phi$ iff $w'\models \phi$. 
    \ifx
    Given a label mapping $\sigma\in \LM$, a label $\sigma\in \Conf$, a world $w\in \Wd$ and a set $A$ of unlabeled formulas,  
the pair $(\lm, \sigma)$ ``has the same effect'' on $A$ as $w$, denoted by $(\lm, \sigma) \se_A w$, if for any $\phi\in A$, $\lm\models \sigma : \phi$ iff $w\models \phi$. 
    \fi
\end{definition}
 

\ifx
\begin{definition}[Standard Configuration]
\label{def:Simple Configurations}
    We say a configuration $\sigma\in\Conf$ is `standard' w.r.t. a set $A$ of formulas, if it satisfies that for any $\rho\in \Eval$, there exists a $\rho'\in \Eval$ such that 
    $(\rho, \sigma)\se_A \rho'$. 

    Simply call $\sigma$ `standard' if $A = \DLF$. We denote the set of all standard configurations w.r.t. $A$ as $\Conf_\std(A)$. 
\end{definition}

The effect of a standard configuration on formulas is equivalent to the effect of an evaluation. 
\fi


\begin{example}
\label{example:Effect Equivalence}
In $\DLpWP$, consider two worlds $w, w'\in \Wd_W$ (which are two mappings from $\Var_\fodl$ to $\mbb{Z}$) satisfying that $w(x) = 1, w(y) = 1$ and $w'(x) = 2, w'(y) = 1$ for some $x, y\in \Var_\fodl$. 
Let fomula $\phi = (x + y > 1)\in \Fmla_\afo$. 
Then $w =_{\{\phi\}} w'$, although $w\neq w'$. 

\end{example}

\ifx
\begin{example}
\label{example:normal configurations}
    In program domain $\Domain{\WP}$, let $\rho, \rho'\in \Eval_\WP$ be two evaluations satisfying $\rho(x) = 5, \rho(y) = 1$ and $\rho'(x) = 6, \rho'(y) = 1$. 
    Let $\sigma = \{x\mapsto x + 1\}\in \Conf_\WP$, formula $\phi = (x + y + 1 > 7)\in \Fmla_\WP$. 
    Then $(\rho, \sigma)\se_{\{\phi\}}\rho'$, since 
    $\rho(\app_\WP(\sigma, \phi))\equiv \rho((x + 1) + y + 1 > 7)\equiv (8 > 7)\equiv \rho'(\phi)$. 

    \ifx
    $\sigma$ is standard, because for any evaluation $\xi\in \Eval_\WP$, let $\xi'\dddef \xi^x_{\xi(x) + 1}$, 
    then we always have $(\xi, \sigma)\se \xi'$. 
    It is not hard to see that 
    the configurations in $\Conf_\WP$ and $\Conf_\E$ introduced in Section~\ref{section:Examples of Term Structures} are actually standard. 
    See a proof for $\Conf_\WP$ in Lemma~\ref{???} (Appendix~\ref{section:Formal Definitions of While Programs}). 
    \fi
\end{example}
\fi


\begin{definition}[Free Labels]
\label{def:free configurations}
A label $\sigma\in \Conf$ is called ``free'' w.r.t. a set $A$ of formulas if 
for any world $w\in \Wd$, there exists a label mapping $\lm\in \LM$ such that $$w =_A \lm(\sigma).$$ 
We denote the set of all free labels w.r.t. $A$ as $\free(\Conf,A)$. 
\end{definition}

Intuitively, the freedom of a label $\sigma$ w.r.t. a set $A$ of formulas means that $\sigma$ is general enough so that it does not have an impact on the validity of the formulas in $A$.

\begin{example}
    In $\DLpWP$, 
    let $\sigma = \{x\mapsto t + 1\}\in \Conf_W$ (with $t$ a variable). 
    $\sigma$ is free w.r.t. $\{\phi\}$ for $\phi = (x + y > 1)$ (Example~\ref{example:Effect Equivalence}). 
    Because for any world $w\in \Wd_W$, let $w'\dddef w[t\mapsto x - 1]$, then we have $w =_{\{\phi\}}m_{w'}(\sigma)$, since $w(x) = \lm_{w'}(\sigma)(x)$ and $w(y) = \lm_{w'}(\sigma)(y)$. 
    On the other hand, let $\sigma' = \{x\mapsto 0, y\mapsto 0\}$, 
    then $\sigma'$ is not free w.r.t. $\{\phi\}$. 
    Because for any $w\in \Wd_W$, $\lm_w(\sigma')(x) = \lm_w(\sigma')(y) = 0$, so $\lm_w(\sigma')\not\models \phi$. 
    
    We see that compared to $\sigma$, $\sigma'$ is too ``explicit'' so that it affects the validity of formula $\phi$. 
\end{example}

\ifx
\begin{definition}[Free Configurations]
\label{def:free configurations}

A configuration $\sigma$ is called `free' w.r.t. a set $A$ of formulas if it satisfies that 
\begin{enumerate}
    \item\label{item:free configurations cond 1} for any $\rho\in \Eval$, there is a $\rho'\in \Eval$ such that $(\rho, \sigma)\se_A \rho'$; 
    \item\label{item:free configurations cond 2} for any $\rho\in \Eval$, there is a $\rho'\in \Eval$ such that 
    $(\rho', \sigma)\se_A \rho$.
\end{enumerate}

We denote the set of all free configurations w.r.t. $A$ as $\free(\Conf,A)$. 
\end{definition}

Intuitively, a free configuration in some sense neither strengthen nor weaken a formula w.r.t. its boolean semantics after affections. 

For a configuration $\sigma\in \Conf_\WP$, 
recall that $e^x_e$ is explained in Example~\ref{example:Proof system for program behaviours}.  

\begin{example}
    In Example~\ref{example:normal configurations}, 
    the configuration $\sigma$ is free w.r.t. $\{\phi\}$.
    On one hand, for any evaluation $\xi\in \Eval_\WP$, let $\xi'\dddef \xi^x_{\xi(x) + 1}$, 
    then we have $(\xi, \sigma)\se_{\{\phi\}} \xi'$ (Definition~\ref{def:free configurations}-\ref{item:free configurations cond 1});
    On the other hand, 
    for any evaluation $\xi$, let $\xi'\dddef \xi^x_{\xi(x) - 1}$, 
    then $(\xi', \sigma)\se_{\{\phi\}} \xi$ holds (which is Definition~\ref{def:free configurations}-\ref{item:free configurations cond 2}). 
    Let $\sigma' = \{x \mapsto 0, y\mapsto 0\}$, then 
    $\sigma'$ is not free w.r.t. $\{\phi\}$, since it makes $\phi$ too strong so that there does not exist an evaluation $\xi$ such that 
    $(\xi, \sigma')\se_{\{\phi\}} \rho'$, violating Definition~\ref{def:free configurations}-\ref{item:free configurations cond 2}. 
    
\end{example}

\fi

For a set $A$ of unlabeled formulas, 
we write $\sigma : A$ to mean the set of labeled formulas $\{\sigma : \phi\ |\ \phi\in A\}$. 

\begin{proposition}[Lifting Process]
\label{prop:lifting process 2}
Given a sound rule of the form
$$
\begin{aligned}
    \infer[]
    {\Gamma\Rightarrow \Delta}
    {\Gamma_1\Rightarrow \Delta_1
    &...&
    \Gamma_n\Rightarrow \Delta_n}
\end{aligned}, \mbox{ $n\ge 1$}, 
$$
in which all formulas are unlabeled, then the rule 
$$
\begin{aligned}
    \infer[]
    {\sigma : \Gamma\Rightarrow \sigma : \Delta}
    {\sigma : \Gamma_1\Rightarrow \sigma : \Delta_1
    &...&
    \sigma : \Gamma_n\Rightarrow \sigma : \Delta_n}
\end{aligned}
$$
is sound for any label $\sigma\in \free(\Conf, \Gamma\cup \Delta\cup \Gamma_1\cup \Delta_1\cup...\cup\Gamma_n\cup\Delta_n)$. 
\end{proposition}

Proposition~\ref{prop:lifting process 2} is proved in Appendix~\ref{section:Other Propositions and Proofs} based on the notion of free labels defined above. 

\ifx
\begin{example}
    ???
\end{example}
\fi

\ifx
From the proof of Proposition~\ref{prop:lifting process 2}, we can see that when the rule is an axiom, we actually only need Definition~\ref{def:free configurations}-\ref{item:free configurations cond 1}, as stated in the following proposition, where we call a configuration $\sigma$ \emph{standard} w.r.t. $A$ if it satisfies  Definition~\ref{def:free configurations}-\ref{item:free configurations cond 1}, and write $\Conf_\std(A)$ as the set of all standard configurations w.r.t. $A$. 

\begin{proposition}[Simple Version of Proposition~\ref{prop:lifting process 2}]
\label{prop:lifting process simple}
Given a sound axiom
$
\begin{aligned}
    \infer[]
    {\Gamma\Rightarrow \Delta}
    {}
\end{aligned}, 
$
rule 
$
\begin{aligned}
    \infer[]
    {\sigma : \Gamma\Rightarrow \sigma : \Delta}
    {}
\end{aligned}
$
is sound for any standard configuration $\sigma\in \Conf_\std(\Gamma\cup \Delta)$. 
\end{proposition}
\fi


\section{Case Studies}
\label{section:Case Studies}
In this section, we illustrate the potential usage of $\DLp$ by several instances. 

We firstly show how labeled dynamic formulas in $\DLp$ can be derived according to the cyclic proof system proposed in Table~\ref{table:General Rules for LDL}. 
We give an example of deduction for a while program (Section~\ref{section:Case Study}). In Appendix~\ref{section:Example Two: A Synchronous Loop Program}, 
we briefly introduce another instantiation of $\DLp$ for the synchronous language Esterel~\cite{Berry92} and show a cyclic derivation of an Esterel program. 
The second example better highlights the advantages of $\DLp$ since the loop structures of some Esterel programs are implicit. 

Secondly, we take the rules in FODL as examples to illustrate how to carry out rule lifting in $\DLp$ (Section~\ref{section:Lifting Process in While Programs}). 
It demonstrates the compatibility of $\DLp$ to the existing dynamic-logic theories, allowing them to be reused in $\DLp$. As we can see, this also helps increasing the efficiency of the derivations in $\DLp$ by adopting the compositional rules in special domains.

In Section~\ref{section:Dynamic Logic LDL} and~\ref{section:A Cyclic Proof System for LDL}, we have seen the instantiation theory $\DLpWP$ of $\DLp$. 
Section~\ref{section:Instantiation of FODL in DLp} and~\ref{section:Encoding of Complex Configurations} further introduces more complex instantiations of $\DLp$. 

In Section~\ref{section:Instantiation of FODL in DLp}, We embed FODL theory into $\DLp$. 
This example shows the potential usefulness of $\DLp$ for different program models in practice, because FODL is the basic theory underlying many dynamic-logic variations (such as~\cite{Beckert2016,Platzer07b,Zhang22,Pardo22}). 
By an example of reasoning about both while programs and regular programs at the same time, 
we also show the heterogeneity of $\DLp$, that different program models can be easily compared with different operational semantics. 

In Section~\ref{section:Encoding of Complex Configurations}, more complex encoding of labels and formulas are further displayed. 
We propose to encode a first-ordered version of process logic~\cite{Harel82} into $\DLp$.
Process logic provides a logical framework for not just reasoning about program properties after the terminations, but also properties during the executions. 
From this example, we see that in $\DLp$ the labels and formulas can be more expressive than in traditional Hoare-style logics where only before-after properties can be reasoned about. 
In Appendix~\ref{section:An Encoding of Separation Logic in DLp}, we give another example to show this by encoding separation logic~\cite{Reynolds02} in $\DLp$. 
It provides a novel way of reasoning about separation-logic formulas directly through symbolic executions. 

\ifx
In Section~\ref{section:Encoding of Complex Configurations} more complex encoding of labels and formulas are further displayed, where 
we propose to encode (partial) separation logic~\cite{Reynolds02} in $\DLp$. 
From this example, we see that in $\DLp$ the labels and formulas can be more expressive than in traditional Hoare-style logics where only before-after properties can be reasoned about. 
The example also provides a novel way of reasoning about separation-logic formulas directly through symbolic executions. 
\fi

\ifx
Note that our following instantiations (i.e. $\DLpFODL, \DLpPL$ and $\DLpSP$) might not be complete, which means that currently we are not sure whether they are (relative) complete or have the same expressiveness as their original counterparts. 
Making a full analysis of these theories is definitely out of the scope of this paper. 
However, our analysis provide strong proofs of the potential usage of $\DLp$ in different domains. 
\fi

\subsection{A Cyclic Deduction of A While Program}
\label{section:Case Study}

We prove the property in Example~\ref{example:DLp specifications} according to the rules in Table~\ref{table:General Rules for LDL}. 
This property can be captured by the following equivalent labeled sequent
$$
\nu_1 \dddef \sigma_1: n \ge 0 \Rightarrow  \sigma_1 : [\textit{WP}] (s = ((N+1)N)/2), 
$$
where $\sigma_1 \dddef \{n\mapsto N, s\mapsto 0\}$,  
describing the initial configuration of $\WP$. 

\renewcommand{\arraystretch}{1}
\begin{table}[tb]
        \noindent\makebox[\textwidth]{%
        \scalebox{1}{
        \begin{tabular}{l|l}
        \toprule
        \begin{tabular}{c}
        \begin{tikzpicture}[->,>=stealth', node distance=3cm]
        \node[draw=none] (txt2) {
            $
            \infer[^{(\Sub)}]
                {\mbox{$\nu_1$: 1}}
                {
                    \infer[^{(\textit{Cut})}]
                    {2}
                    {
                        \infer[^{(\textit{Wk} R)}]
                        {3}
                        {
                            \infer[^{(\textit{Ter})}]
                            {17}
                            {}
                        }
                        &    
                        \infer[^{(\vee L)}]
                        {4}
                        {\infer[^{([\alpha]R)}]
                            {5}
                            {
                            \infer[^{([\alpha]R)}]
                                {9}
                                {
                                    \infer[^{(\textit{Cut})}]
                                    {10}
                                    {
                                        \infer[^{(\textit{Wk} L)}]
                                        {11}
                                        {
                                            \infer[^{(\Sub)}]
                                            {14}
                                            {
                                                \infer[^{(\textit{Wk} L)}]
                                                {15}
                                                {16}
                                            }
                                        }
                                        &
                                        \infer[^{(\textit{Wk} R)}]
                                        {12}
                                        {
                                            \infer[^{(\textit{Ter})}]
                                            {13}
                                            {}
                                        }
                                    }
                                }
                            }
                        &
                        \infer[^{([\alpha]R)}]
                            {6}
                            {
                            \infer[^{([\ter])}]
                                {7}
                                {\infer[^{(\textit{Ter})}]
                                    {8}
                                    {}
                                }
                            }
                        }
                    }
                }
            $
        };


        \path
        ;


        \draw[dotted,thick,red] ([xshift=-1.35cm, yshift=1.6cm]txt2.center) -- 
        ([xshift=-1.6cm, yshift=1.6cm]txt2.center) --
        ([xshift=-1.6cm, yshift=-1.25cm]txt2.center) --
        ([xshift=-1.1cm, yshift=-1.25cm]txt2.center);

        \end{tikzpicture}
        \end{tabular}
        &
        \begin{tabular}{l}
            Definitions of other symbols:
            \\
            $
            \textit{WP}
             \dddef
            \{
            \textit{while}\
            (n > 0)\
            \textit{do}\
            s := s + n\ ;\
            n := n - 1\
            \textit{end}\
           \}$
            \\
            $\alpha_1\dddef s := s + n\ ;\ n := n - 1$
            \\
            $\phi_1 \dddef (s = ((N + 1)N)/2)$
            \\
            $\sigma_1 \dddef \{n\mapsto N, s\mapsto 0\}$
            \\
            $\sigma_2\dddef \{n\mapsto N - m,
            s\mapsto (2N - m + 1)m/2\}$
            \\
            $\sigma_3 \dddef \{n\mapsto N - m,
            s\mapsto (2N - (m + 1) + 1)(m + 1)/2\}$
            \\
            $\sigma_4\dddef \{n\mapsto N - (m + 1),
            s\mapsto (2N - (m + 1) + 1)(m + 1)/2\}$
            \\
            \end{tabular}
        \\
        \midrule
        \multicolumn{2}{c}{
            \begin{tabular}{l c c c}
            1: & $\sigma_1 : n \ge 0$ & $\Rightarrow$ & $\sigma_1 : [\textit{while}\
    (n > 0)\
    \textit{do}\
    \alpha_1\
    \textit{end}\ ]\phi_1$
            \\
            2: & $\sigma_2 : n\ge 0$ & $\Rightarrow$ & $\mbox{\ul{$\sigma_2 : [\textit{while}\
    (n > 0)\
    \textit{do}\
    \alpha_1\
    \textit{end}\ ]\phi_1$}}$
            \\
            3: & $\sigma_2 : n\ge 0$ & $\Rightarrow$ & 
            $
                 \sigma_2 : [\textit{while}\
    (n > 0)\
    \textit{do}\
    \alpha_1\
    \textit{end}\ ]\phi_1
            , \sigma_2 : (n > 0 \vee n\le 0)
            $
            \\
            17: & $\sigma_2 : n\ge 0$ & $\Rightarrow$ & $\sigma_2 : (n > 0 \vee n\le 0)$
            \\
            \midrule
            4: & $\sigma_2: n\ge 0, \sigma_2 : (n > 0 \vee n\le 0)$ & $\Rightarrow$ & $\mbox{\ul{$\sigma_2 : [\textit{while}\
    (n > 0)\
    \textit{do}\
    \alpha_1\
    \textit{end}\ ]\phi_1$}}$
            \\
            5: & $\sigma_2: n\ge 0, \sigma_2 : n > 0$ & $\Rightarrow$ & $\mbox{\ul{$\sigma_2 : [\textit{while}\
    (n > 0)\
    \textit{do}\
    \alpha_1\
    \textit{end}\ ]\phi_1$}}$
            \\
            9: & $\sigma_2: n\ge 0, \sigma_2 : n > 0$ & $\Rightarrow$ & $\mbox{\ul{$\sigma_3 : [n := n - 1;\ \textit{while}\
    (n > 0)\
    \textit{do}\
    \alpha_1\
    \textit{end}\ ]\phi_1$}}$
            \\
            10: & $\sigma_2: n\ge 0, \sigma_2 : n > 0$ & $\Rightarrow$ & $\mbox{\ul{$\sigma_4 : [\textit{while}\
    (n > 0)\
    \textit{do}\
    \alpha_1\
    \textit{end}\ ]\phi_1$}}$
            \\
            11: & $
                \sigma_2: n\ge 0, \sigma_2 : n > 0,
                \sigma_4: n \ge -1, \sigma_4: n \ge 0$
            & $\Rightarrow$ & $\mbox{\ul{$\sigma_4 : [\textit{while}\
    (n > 0)\
    \textit{do}\
    \alpha_1\
    \textit{end}\ ]\phi_1$}}$
            \\
            14: & $\sigma_4 : n \ge -1, \sigma_4: n \ge 0$ & $\Rightarrow$ & $\mbox{\ul{$\sigma_4 : [\textit{while}\
    (n > 0)\
    \textit{do}\
    \alpha_1\
    \textit{end}\ ]\phi_1$}}$
            \\
            15: & $\sigma_2: n \ge -1, \sigma_2: n \ge 0$ & $\Rightarrow$ & $\mbox{\ul{$\sigma_2 : [\textit{while}\
    (n > 0)\
    \textit{do}\
    \alpha_1\
    \textit{end}\ ]\phi_1$}}$
            \\
            16: & $\sigma_2: n \ge 0$ & $\Rightarrow$ & $\mbox{\ul{$\sigma_2 : [\textit{while}\
    (n > 0)\
    \textit{do}\
    \alpha_1\
    \textit{end}\ ]\phi_1$}}$
            \\
            \midrule
            12: & $\sigma_2 : n\ge 0, \sigma_2 : n > 0$ & $\Rightarrow$ & 
            $
                 \sigma_4 : [\textit{while}\
    (n > 0)\
    \textit{do}\
    \alpha_1\
    \textit{end}\ ]\phi_1
                 , \sigma_4: n \ge -1, \sigma_4: n \ge 0$
            \\
            13: & $\sigma_2 : n\ge 0, \sigma_2 : n > 0$ & $\Rightarrow$ & 
            $
                 \sigma_4 : n \ge -1, \sigma_4 : n \ge 0
            $
            \\
            \midrule
            6: & $\sigma_2 : n\ge 0, \sigma_2 : n \le 0$ & $\Rightarrow$ & $\sigma_2 : [\textit{while}\
    (n > 0)\
    \textit{do}\
    \alpha_1\
    \textit{end}\ ]\phi_1$
            \\
            7: & $\sigma_2 : n\ge 0, \sigma_2 : n \le 0$ & $\Rightarrow$ &  $\sigma_2 : [\ter]\phi_1$
            \\
            8: & $\sigma_2 : n\ge 0, \sigma_2 : n \le 0$ & $\Rightarrow$ & $\sigma_2 : (s=((N+1)N)/2)$
            \\
            \end{tabular}
        }
        \\
        \bottomrule
        \end{tabular}
        }
         }
        \caption{A Derivation of Property $\nu_1$}
        \label{figure:The derivation of Example 1}
\end{table}

Table~\ref{figure:The derivation of Example 1} shows its derivations.
We omit all side deductions as sub-proof procedures in instances of rule $([\alpha]R)$ derived using the inference rules in 
Table~\ref{table:An Example of Inference Rules for Program Behaviours}. 
Non-primitive rule $(\vee L)$ can be derived by the rules for $\neg$ and $\wedge$ as follows: 
$$
\begin{aligned}
    \infer[^{(\neg L)}]
    {\Gamma, \phi\vee \psi\Rightarrow \Delta}
    {
        \infer[^{(\wedge R)}]
        {\Gamma\Rightarrow (\neg\phi) \wedge (\neg\psi), \Delta}
        {
            \infer[^{(\neg R)}]
            {\Gamma\Rightarrow \neg\phi, \Delta}
            {\Gamma, \phi\Rightarrow \Delta}
        &
            \infer[^{(\neg R)}]
            {\Gamma\Rightarrow \neg\psi, \Delta}
            {\Gamma, \psi\Rightarrow \Delta}
        }
    }
\end{aligned}
.
$$
\ifx
For example, rule $(\vee L)$ can be derived as follows:
$$
\begin{aligned}
    \infer[^{(\neg L)}]
    {\Gamma, \phi\vee \psi\Rightarrow \Delta}
    {
        \infer[^{(\wedge R)}]
        {\Gamma\Rightarrow (\neg\phi) \wedge (\neg\psi), \Delta}
        {
            \infer[^{(\neg R)}]
            {\Gamma\Rightarrow \neg\phi, \Delta}
            {\Gamma, \phi\Rightarrow \Delta}
        &
            \infer[^{(\neg R)}]
            {\Gamma\Rightarrow \neg\psi, \Delta}
            {\Gamma, \psi\Rightarrow \Delta}
        }
    }
\end{aligned}
.
$$
\fi

The derivation from sequent 1 to 2 (also the derivation from 14 to 15) is according to the rule $(\Sub)$:
$$
\begin{aligned}
    \infer[^{(\Sub)}]
    {\Gamma[e/x]\Rightarrow \Delta[e/x]}
    {\Gamma\Rightarrow\Delta}, 
\end{aligned}
$$
where the function $(\cdot)[e/x]$ is an instantiation of the abstract subsitution defined in Definition~\ref{def:Substitution of Labels}. 
For any label $\sigma$, $\sigma[e/x]$ returns the label by substituting each free variable $x$ of $\sigma$ with term $e$.  
We observe that $\sigma_1 = \sigma_2[0/m]$, so sequent 1 is a special case of sequent 2 by substitution $(\cdot)[0/m]$. 
\ifx
can be writen as: 
$$\sigma_2[0/m] : n\ge 0\Rightarrow \sigma_2[0/m] : [\textit{while}\
    (n > 0)\
    \textit{do}\
    \alpha_1\
    \textit{end}\ ]\phi_1, $$
with $\sigma_1 = \sigma_2[0/m]$, 
as a special case obtained from sequent 2 by substitution $(\cdot)[0/m]$.  
\fi
Intuitively, label $\sigma_2$ captures the program configuration after the $m$th loop ($m\ge 0$) of program $\WP$. 
This step is crucial as starting from sequent 2, we can find a bud node --- 16 --- that is identical to node 2. 

The derivation from sequent 2 to \{3, 4\} provides a lemma: 
$\sigma_2 : (n > 0\vee n\le 0)$, which is trivially valid. 
Sequent 16 indicates the end of the $(m+1)$th loop of program $\WP$. 
From node 10 to 16, we transform the formulas on the left side into a trivial logical equivalent form in order to apply rule $(\Sub)$ from sequent 14 to 15. 
Sequent 14 is a special case of sequent 15 since $\sigma_4 = \sigma_2[m + 1/m]$. 
\ifx
can be written as:
$$
\sigma_2[m + 1/m] : n \ge -1, \sigma_2[m + 1/m] : n \ge 0\Rightarrow \sigma_2[m + 1/m] : [\textit{while}\
    (n > 0)\
    \textit{do}\
    \alpha_1\
    \textit{end}\ ]\phi_1, 
$$
with $\sigma_4 = \sigma_2[m + 1/m]$. 
\fi


The whole proof tree is cyclic because the only derivation path: $2,4,5,9,10,11,14,15,16,2,...$
has a progressive derivation trace whose elements are underlined in Table~\ref{figure:The derivation of Example 1}.


One feature of the above deduction process is that the loop structure of the while program $\WP$ (i.e. $\textit{while}...\textit{do}...\textit{end}$) is reflected in the cyclic derivation tree itself. 
To reason about $\WP$ one does not need the inference rule for decomposing the loop structure. 
This is useful especially in program models in which loop structures are usually implicit, such as CCS-like process algebras~\cite{Milner82,Milner92} and imperative synchronous languages~\cite{Berry92,Schneider17}.
$\DLp$ provides an incremental reasoning in which we can avoid prior program transformations as done in work like~\cite{DLForCCS,Schneider17}.

\ifx
Compared to the deduction processes in traditional dynamic logics and Hoare logics, a notable feature of the above deduction process is that the search for a loop invariant is reflected in looking for a suitable configuration (i.e. $\sigma_2$). 
One advantage brought by this cyclic derivation approach is that it does not rely on the inference rule for decomposing an explicit loop structure (here $\textit{while}...\textit{do}...\textit{end}$), which also makes it easily amendable for reasoning about programs with implicit loop structures, such as CCS-like process algebras~\cite{Milner82,Milner92} and some synchronous languages~\cite{Berry92,Schneider17}.  
\fi



\ifx
As a demonstration of its powerfulness,
in Appendix~\ref{section:Example Two: A Synchronous Loop Program}, 
we briefly introduce another instantiation of $\DLp$ for the synchronous language Esterel~\cite{Berry92} and show a cyclic derivation of an Esterel program. 
That example can better highlight the advantages of $\DLp$ since the loop structures of some Esterel programs are implicit. 
In Section~\ref{section:Instantiation of FODL in DLp}, we also briefly show that FODL~\cite{Pratt76} can be instantiated in $\DLp$. 
\fi

\subsection{Lifting Rules From FODL}
\label{section:Lifting Process in While Programs}
Two examples in $\DLpWP$ are given to illustrate how the existing inference rules from the theory of FODL (cf.~\cite{Harel00}) can be applied for deriving the compositional while programs through the lifting processes as defined in Section~\ref{section:Lifting Process From Program Domains}.  


In FODL, consider the rule
$$
\infer[^{([\seq])}]
{\Gamma \Rightarrow [\alpha\seq \beta]\phi, \Delta}
{\Gamma \Rightarrow [\alpha][\beta]\phi, \Delta},
$$
which means that to prove formula $[\alpha\seq \beta]\phi$, we only need to prove formula $[\alpha][\beta]\phi$ in which program $\alpha$ is firstly proved separated from program $\beta$.  
It comes from the valid formula $[\alpha][\beta]\phi\rightarrow [\alpha\seq \beta]\phi$, acting
as a compositional rule appearing in many dynamic logic calculi that are based on FODL (e.g.~\cite{Beckert2016}). 
By Proposition~\ref{prop:lifting process 2}, in $\DLpWP$, we can lift $([;])$ 
as a rule
$$
    \infer[^{(\sigma[\seq])}]
    {\sigma : \Gamma\Rightarrow \sigma : [\alpha\seq\beta]\phi, \sigma: \Delta}
    {\sigma : \Gamma\Rightarrow \sigma : [\alpha][\beta]\phi, \sigma : \Delta}, 
$$
where $\sigma$ is a free configuration in $\free(\Conf, \Gamma\cup \Delta\cup \{[\alpha][\beta]\phi, [\alpha\seq\beta]\phi\})$. 
As an additional rule, in system $\pfDLp$, $(\sigma[\seq])$ provides a compositional reasoning for sequential programs. 
    It is useful when verifying a property like $\Gamma, \sigma' : [\beta]\phi\Rightarrow \sigma : [\alpha\seq \beta]\phi, \Delta$, in which we might finish the proof by only symbolic executing program $\alpha$ as: 
    $$
    \infer[^{(\sigma[\seq])}]
    {\Gamma, \sigma' : [\beta]\phi\Rightarrow \sigma : [\alpha\seq\beta]\phi, \Delta}
    {
        \infer[^{([\alpha]R)}]
        {
        \Gamma, \sigma' : [\beta]\phi\Rightarrow \sigma : [\alpha][\beta]\phi, \Delta
        }
        {
            \infer*[]
            {...}
            {
                \infer[^{([\alpha]R)}]
                {...}
                {
                    \infer[^{([\ter])}]
                    {\Gamma, \sigma' : [\beta]\phi\Rightarrow \sigma' : [\ter][\beta]\phi, \Delta}
                    {
                        \infer[^{(\textit{ax})}]
                        {\Gamma, \sigma' : [\beta]\phi\Rightarrow \sigma' : [\beta]\phi, \Delta}
                        {}
                    }
                }
            }
        }
    }, 
    $$
    especially when verifying the program $\beta$ can be very costly. 

Another example is the rule
$$
\infer[^{([\textit{Gen}])}]
{[\alpha]\phi\Rightarrow [\alpha]\psi}
{\phi\Rightarrow \psi}
$$
for generating modality $[\cdot]$, 
which were used for deriving the structural rule of star regular programs in FODL (cf.~\cite{Harel00}). 
By Proposition~\ref{prop:lifting process 2}, we lift $([\textit{Gen}])$ as the following rule:
$$
\infer[^{(\sigma [\textit{Gen}])}]
{\sigma : [\alpha]\phi\Rightarrow \sigma : [\alpha]\psi}
{\sigma : \phi\Rightarrow \sigma : \psi}, 
$$
where $\sigma\in \free(\Conf, \{[\alpha]\phi, [\alpha]\psi, \phi, \psi\})$. 
It is useful, for example, when deriving a property
$\sigma : [\alpha]\la\beta\ra \phi\Rightarrow \sigma : [\alpha]\la\beta'\ra\psi$, where 
we can skip the derivation of program $\alpha$ as follows: 
$$
\infer[^{(\sigma [\textit{Gen}])}]
{\sigma : [\alpha]\la\beta\ra \phi\Rightarrow \sigma : [\alpha]\la\beta'\ra\psi}
{\sigma : \la\beta\ra\phi\Rightarrow \sigma : \la\beta'\ra\psi},  
$$
and directly focus on deriving the programs $\beta$ and $\beta'$. 

From these two examples it can be seen that in practical derivations, lifting process can be used to reduce the burden of certain verifications.

\ifx
\subsection{Instantiation of $\DLp$ as PDL}
\label{section:Instantiation of PDL in DLp}

\cite{Docherty19} proposed a labeled cyclic proof system for PDL, where.... 
Below we show that an alternative version of cyclic proof system for PDL can be realized in $\DLp$ under a certain instantiation. 
As PDL/FODL is the basic theory underlying many dynamic-logic variations (such as~\cite{Beckert2016,Platzer07b,Zhang22,Pardo22}), this example shows that how $\DLp$ can be potentially useful in specifying and reasoning about different programming languages and models. 

(introduce about how they have done in the previous work. ???)

In $\DLp$, we can realize a similar theory to~\cite{Docherty19} by specifying a suitable set $\Conf_\pdl$ of labels.  
We define a label $\sigma\in \Conf_\pdl$ as the form: 
$$
\sigma \dddef (x, \rel), 
$$
where $x$ is a world variable, $\rel$ is a set of action relations. 
With this, we thus instantiate the set $\pfOper$ for regular programs as shown in Table~\ref{???}, namely $(\pfOper)_\pdl$. 

\begin{table}[tb]
         \begin{center}
         \noindent\makebox[\textwidth]{%
         \scalebox{0.9}{
         \begin{tabular}{c}
         \toprule
        $
         \infer[^{(a)}]
         {\Gamma\Rightarrow (a, x)\trans(\ter, y), \Delta}
         {
         \Gamma\Rightarrow x : (x\xrightarrow{a} y), \Delta
         }
         $
         \ \ 
         $
         \infer[^{(\phi?)}]
         {\Gamma\Rightarrow (\phi?, \sigma)\trans(\ter, \sigma), \Delta}
         {
         \Gamma\Rightarrow \sigma : \phi, \Delta
         }
         $
         \\
         $
         \infer[^{1\ (;)}]
         {\Gamma \Rightarrow (\alpha\seq\beta, \sigma)\trans (\alpha'\seq \beta, \sigma'), \Delta}
         {
            \Gamma\Rightarrow (\alpha, \sigma)\trans (\alpha', \sigma'), \Delta
         }
         $
         \ \ 
         $
         \infer[^{(;\ter)}]
         {\Gamma\Rightarrow (\alpha\seq\beta, \sigma)\trans(\beta, \sigma), \Delta}
         {
         \Gamma\Rightarrow (\alpha, \sigma)\trans(\ter, \sigma'), \Delta)
         }
         $
         \\
         $
         \infer[^{(\cho 1)}]
         {\Gamma\Rightarrow (\alpha\cho\beta, \sigma)\trans (\alpha, \sigma), \Delta}
         {
         }
         $
         \ \ 
        $
         \infer[^{(\cho 2)}]
         {\Gamma\Rightarrow (\alpha\cho\beta, \sigma)\trans (\beta, \sigma), \Delta}
         {
         }
         $
         \\
         $
         \infer[^{2\ (*)}]
         {\Gamma\Rightarrow (\alpha^\lup, \sigma)\trans (\alpha\seq\alpha^\lup\cup \true?, \sigma), \Delta}
         {
         }
         $
         \ifx
         \ \ 
         $
         \infer[^{2\ (* 2)}]
         {\Gamma\Rightarrow \upd\xrightarrow{\alpha^*/\alpha^*}\upd', \Delta}
         {
         \Gamma\Rightarrow \upd\xrightarrow{\alpha/\ter}\upd' \Delta
         }
         $
          \ \ 
         $
         \infer[^{(* \ter)}]
         {\Gamma\Rightarrow \upd\xrightarrow{\alpha^*/\ter}\upd, \Delta}
         {
         }
         $
         \fi
         \\
         \bottomrule
         \end{tabular}
              }
              }
          \end{center}
          \caption{Inference Rules $(\pfOper)_\pdl$ for Regular Programs of PDL}
          \label{table:Inference Rules for Regular Programs of PDL}
    \end{table}
\fi

\subsection{Instantiation of $\DLp$ in FODL Theory}
\label{section:Instantiation of FODL in DLp}

We instantiate $\DLp$ with the theory of FODL. 
The resulted theory, namely $\DLpFODL$, provides an alternative way of reasoning about 
FODL formulas through symbolically executing regular programs.  

The instantiation process mainly follows that for while programs as we have seen in Example~\ref{example:While program}, \ref{example:Labels and Formulas of labeled Sequents}, \ref{example:Label Mappings} and~\ref{example:Proof system for program behaviours}, where the only differences are: (1) the parameter $\Prog$ is instantiated as the set of regular programs (Section~\ref{section:PDL FODL}), denoted by $\Prog_\fodl$; (2) the program behaviours are captured by a set of rules for regular programs, denoted by $(\pfOper)_\fodl$, whose rules for the part of the program transitions of regular programs are shown in Table~\ref{table:Inference Rules for Program Transitions of Regular Programs of FODL}. And we omit the part of the rules for program terminations of regular programs. 

\renewcommand{\arraystretch}{1.5}
\begin{table}[tb]
         \begin{center}
         \noindent\makebox[\textwidth]{%
         \scalebox{1}{
         \begin{tabular}{c}
         \toprule
        $
         \infer[^{(x:=e)}]
         {\Gamma\Rightarrow (x:=e, \sigma)\trans(\ter, \sigma^x_e), \Delta}
         {
         }
         $
         \ \ 
         $
         \infer[^{(\phi?)}]
         {\Gamma\Rightarrow (\phi?, \sigma)\trans(\ter, \sigma), \Delta}
         {
         \Gamma\Rightarrow \sigma : \phi, \Delta
         }
         $
         \\
         $
         \infer[^{1\ (;)}]
         {\Gamma\Rightarrow (\alpha\seq\beta, \sigma)\trans(\alpha'\seq\beta, \sigma'), \Delta}
         {
            \Gamma\Rightarrow (\alpha, \sigma)\trans(\alpha', \sigma'), \Delta
         }
         $
         \ \  
         $
         \infer[^{(;\ter)}]
         {\Gamma\Rightarrow (\alpha\seq\beta, \sigma)\trans(\beta, \sigma'), \Delta}
         {
            \Gamma\Rightarrow (\alpha, \sigma)\trans(\ter, \sigma'), \Delta
         }
         $
         \\
         $
         \infer[^{(\cho 1)}]
         {\Gamma\Rightarrow (\alpha\cho\beta, \sigma)\trans(\alpha, \sigma), \Delta}
         {
         }
         $
         \ \ 
        $
         \infer[^{(\cho 2)}]
         {\Gamma\Rightarrow (\alpha\cho\beta, \sigma)\trans(\beta, \sigma), \Delta}
         {
         }
         $
         \\
         $
         \infer[^{2\ (*)}]
         {\Gamma\Rightarrow (\alpha^\lup, \sigma)\trans(\alpha\seq\alpha^\lup\cho \true?, \sigma), \Delta}
         {
         }
         $
         \\
         \bottomrule
         \end{tabular}
              }
              }
          \end{center}
          \caption{Partial Rules of $(\pfOper)_\fodl$ for Program Transitions of Regular Programs}
          \label{table:Inference Rules for Program Transitions of Regular Programs of FODL}
    \end{table}

It is interesting to compare our proof system: $(\pfDLp)_\fodl \dddef \pfLDLp\cup (\pfOper)_\fodl$ for FODL with the traditional proof system of FODL (cf.~\cite{Harel00}). 
By simple observations, we can see that for non-star regular programs, our proof system can do what the traditional proof system can. 
For example, an FODL formula $$[(\alpha\seq \beta)\cup \alpha]\phi, $$
where let $\alpha\dddef (x := x + 1)$ and $\beta\dddef (y := 0)$, 
can be derived in the following process in the traditional proof system of FODL by using certain rules:
$$
\infer[^{(\cup)}]
{x \ge 0\Rightarrow [(\alpha\seq \beta)\cup \alpha]x > 0}
{
    \infer[^{(\seq)}]
    {x \ge 0\Rightarrow [\alpha\seq \beta]x > 0}
    {
        \infer[^{(x:=e)}]
        {x \ge 0\Rightarrow [\alpha][\beta]x > 0}
        {x \ge 0\Rightarrow [\beta]x + 1 > 0}
    }
    &
    x\ge 0\Rightarrow [\alpha]x > 0
}. 
$$
In $(\pfDLp)_\fodl$, correspondingly, we can find a logical equivalent labeled version: 
$$\{x\mapsto t\} : x\ge 0\Rightarrow \{x\mapsto t\} : [(\alpha\seq \beta)\cup \alpha]x > 0$$ and have the following derivations by applying the rule $([\alpha]R)$ and using the corresponding operational rules in $(\pfOper)_\fodl$ (which are not shown below):
$$
\infer[^{([\alpha]R)}]
{\{x\mapsto t\} : x\ge 0\Rightarrow \{x\mapsto t\} : [(\alpha\seq \beta)\cup \alpha]x > 0}
{
    \infer[^{([\alpha]R)}]
    {\{x\mapsto t\} : x\ge 0\Rightarrow \{x\mapsto t\} : [\alpha\seq \beta]x > 0}
    {
        \{x\mapsto t\} : x\ge 0\Rightarrow \{x\mapsto t + 1\} : [\beta]x > 0
    }
&
\{x\mapsto t\} : x\ge 0\Rightarrow \{x\mapsto t\} : [\alpha]x > 0
}. 
$$

Especially, we are interested in whether $(\pfDLp)_\fodl$ (w.r.t. a suitable $(\pfOper)_\fodl$ for which we only give a part of the rules here), like the traditional proof system of FODL, is complete  related to the arithmetical theory of integers. 
One way to prove this, as we can see now, is by applying our conditional completeness result for $\DLp$ proposed in Section~\ref{section:Conditional Completeness of DLp}, in which the crucial step is showing that regular programs in the proof system $(\pfDLp)_\fodl$ is well-behaved (Definition~\ref{def:Well-behaved Loop Programs}). 

More of these aspects will be discussed in detail in our future work. 

The heterogeneity of the verification framework of $\DLp$ can be reflected from this example. 
Although while programs are a subset of regular programs in the context of our discussion (see Section~\ref{section:PDL FODL}), 
the while programs have its own set of the inference rules for their program transitions (Table~\ref{table:An Example of Inference Rules for Program Behaviours}), which is different from those for regular programs (Table~\ref{table:Inference Rules for Program Transitions of Regular Programs of FODL}). 
Our $\DLp$ formulas provide a convenient way to compare the behaviours of different models.
For example, given a regular program $\WP_r$ which is syntactically equivalent to $\WP$ (Example~\ref{example:While program}): 
$$
\WP_r \dddef ((n > 0)?\seq s:=s+n\seq n:= n - 1)^*\seq \neg (n>0)?, 
$$
we want to verify that whether their behaviours lead to the same result according to their own operational semantics. This property can be described as a $\DLp$ formula as follows: 
$$
\cdot \Rightarrow \sigma : [\WP][\WP'_r](s = s'\wedge n = n'), 
$$
where $\sigma = \{s\mapsto t, n\mapsto N, s'\mapsto t, n'\mapsto N\}$; 
$\WP'_r$ is obtained from $\WP_r$ by replacing all appearances of the variables $s, n$ with their fresh counterparts $s', n'$ in order to avoid variable collisions;  
$t, N$ are fresh variables other than $s, n, s', n'$. 
Intuitively, the formula says that if the inputs of $\WP$ and $\WP_r$ are the same, after running them separately (without interactions), their outputs are the same.

\ifx
\begin{table}[tb]
         \begin{center}
         \noindent\makebox[\textwidth]{%
         \scalebox{0.9}{
         \begin{tabular}{c}
         \toprule
        $
         \infer[^{(x:=e)}]
         {\Gamma\Rightarrow \upd\xrightarrow{x:=e/\ter}\{x := e\}\upd, \Delta}
         {
         }
         $
         \ \
         $
         \infer[^{1\ (;)}]
         {\Gamma\Rightarrow \upd\xrightarrow{\alpha\seq \beta/\alpha'\seq \beta}\upd', \Delta}
         {
            \Gamma\Rightarrow \upd\xrightarrow{\alpha/\alpha'}\upd', \Delta
         }
         $
         \ \ 
         $
         \infer[^{(;\ter)}]
         {\Gamma\Rightarrow \upd\xrightarrow{\alpha\seq \beta/\beta}\upd', \Delta}
         {
            \Gamma\Rightarrow \upd\xrightarrow{\alpha/\ter}\upd', \Delta
         }
         $
         \ \ 
         $
         \infer[^{(\cho 1)}]
         {\Gamma\Rightarrow \upd\xrightarrow{\alpha\cho \beta/\alpha'}\upd', \Delta}
         {
            \Gamma\Rightarrow \upd\xrightarrow{\alpha/\alpha'}\upd', \Delta
         }
         $
         \\
        $
         \infer[^{(\cho 2)}]
         {\Gamma\Rightarrow \upd\xrightarrow{\alpha\cho \beta/\beta'}\upd', \Delta}
         {
            \Gamma\Rightarrow \upd\xrightarrow{\beta/\beta'}\upd', \Delta
         }
         $
         \ \ 
         $
         \infer[^{2\ (* 1)}]
         {\Gamma\Rightarrow \upd\xrightarrow{\alpha^*/\alpha'\seq\alpha^*}\upd', \Delta}
         {
         \Gamma\Rightarrow \upd\xrightarrow{\alpha/\alpha'}\upd', \Delta
         }
         $
         \ \ 
         $
         \infer[^{2\ (* 2)}]
         {\Gamma\Rightarrow \upd\xrightarrow{\alpha^*/\alpha^*}\upd', \Delta}
         {
         \Gamma\Rightarrow \upd\xrightarrow{\alpha/\ter}\upd' \Delta
         }
         $
          \ \ 
         $
         \infer[^{(* \ter)}]
         {\Gamma\Rightarrow \upd\xrightarrow{\alpha^*/\ter}\upd, \Delta}
         {
         }
         $
         \\
         \bottomrule
         \end{tabular}
              }
              }
          \end{center}
          \caption{Inference Rules $(\pfPT)_\fodl$ for Program Transitions of Regular Programs}
          \label{table:Inference Rules for Program Transitions of Regular Programs of FODL}
    \end{table}
\fi

\ifx
Let $\Prog_{\fodl}$ be the set of regular programs of FODL and $\Fmla_\fodl$ be the set of non-dynamic FODL formulas. 
$\Krp_\fodl=(\Wd_\fodl, \trans, \mcl{I}_\fodl)$ is the \PLK\ structure of FODL, where 
each world $w\in \Wd_\fodl$ is a mapping $w: \Var_\fodl\to \mbb{Z}$ from the set $\Var_\fodl$ of variables to integer domain. 

An `update' is of the form: $\{x := e\}$ (cf.~\cite{Beckert13}) where $x\in \Var_\fodl$ and $e$ is an expression. 
For a formula $\phi$, an update $\{x := e\}$ on $\phi$, denoted by $\{x := e\}(\phi)$, returns a formula by substituting each free variable $x$ of $\phi$ with $e$. 
A label $\upd\in \Conf_\fodl$ of FODL is a sequence of `updates', 
which applies to a formula $\phi$ by applying each update of $\sigma$ on $\phi$ from right to left, denoted by $\upd(\phi)$. 
For example, an application $\{x := y\}\{x := x + 1\}(x = 5)$ results in $y + 1 = 5$. 
When a label $\upd$ contains no free variables, it is a world in $\Wd$. 

For a world $w\in \Krp_\fodl$, let $w(\upd)$ return an update by substituting each free variable $x$ of $\upd$ with value $w(x)$. 
So $w(\upd)$ (without free variables) is a world itself. 
Therefore, each world $w$ itself is a label mapping, and we have $\LM_\fodl = \Wd$ in $\Krp_\fodl$.  



Transitional behaviours of $\Krp_\fodl$ are captured by the inference rules namely $(\pfPT)_\fodl$ in 
Table~\ref{table:Inference Rules for Program Transitions of Regular Programs of FODL} following Definition~\ref{def:Matching Proof System}. 
One can also define a set $(\pfPTer)_\fodl$ of inference rules for the terminations of regular programs, which we omit in this paper. 
\fi

\ifx
\subsection{Instantiation of FODL in $\DLp$}

We show that FODL can be instantiated in the theory of $\DLp$. 
As FODL is the basic theory underlying many dynamic-logic variations (such as~\cite{Beckert2016,Platzer07b,Zhang22,Pardo22}), it demonstrates how $\DLp$ can be potentially useful in specifying and reasoning about different languages.

\begin{table}[tb]
         \begin{center}
         \noindent\makebox[\textwidth]{%
         \scalebox{0.9}{
         \begin{tabular}{c}
         \toprule
        $
         \infer[^{(x:=e)}]
         {\Gamma\Rightarrow \upd\xrightarrow{x:=e/\ter}\{x := e\}\upd, \Delta}
         {
         }
         $
         \ \
         $
         \infer[^{1\ (;)}]
         {\Gamma\Rightarrow \upd\xrightarrow{\alpha\seq \beta/\alpha'\seq \beta}\upd', \Delta}
         {
            \Gamma\Rightarrow \upd\xrightarrow{\alpha/\alpha'}\upd', \Delta
         }
         $
         \ \ 
         $
         \infer[^{(;\ter)}]
         {\Gamma\Rightarrow \upd\xrightarrow{\alpha\seq \beta/\beta}\upd', \Delta}
         {
            \Gamma\Rightarrow \upd\xrightarrow{\alpha/\ter}\upd', \Delta
         }
         $
         \ \ 
         $
         \infer[^{(\cho 1)}]
         {\Gamma\Rightarrow \upd\xrightarrow{\alpha\cho \beta/\alpha'}\upd', \Delta}
         {
            \Gamma\Rightarrow \upd\xrightarrow{\alpha/\alpha'}\upd', \Delta
         }
         $
         \\
        $
         \infer[^{(\cho 2)}]
         {\Gamma\Rightarrow \upd\xrightarrow{\alpha\cho \beta/\beta'}\upd', \Delta}
         {
            \Gamma\Rightarrow \upd\xrightarrow{\beta/\beta'}\upd', \Delta
         }
         $
         \ \ 
         $
         \infer[^{2\ (* 1)}]
         {\Gamma\Rightarrow \upd\xrightarrow{\alpha^*/\alpha'\seq\alpha^*}\upd', \Delta}
         {
         \Gamma\Rightarrow \upd\xrightarrow{\alpha/\alpha'}\upd', \Delta
         }
         $
         \ \ 
         $
         \infer[^{2\ (* 2)}]
         {\Gamma\Rightarrow \upd\xrightarrow{\alpha^*/\alpha^*}\upd', \Delta}
         {
         \Gamma\Rightarrow \upd\xrightarrow{\alpha/\ter}\upd' \Delta
         }
         $
          \ \ 
         $
         \infer[^{(* \ter)}]
         {\Gamma\Rightarrow \upd\xrightarrow{\alpha^*/\ter}\upd, \Delta}
         {
         }
         $
         \\
         \bottomrule
         \end{tabular}
              }
              }
          \end{center}
          \caption{Inference Rules $(\pfPT)_\fodl$ for Program Transitions of Regular Programs}
          \label{table:Inference Rules for Program Transitions of Regular Programs of FODL}
    \end{table}

Let $\Prog_{\fodl}$ be the set of regular programs of FODL and $\Fmla_\fodl$ be the set of non-dynamic FODL formulas. 
$\Krp_\fodl=(\Wd_\fodl, \trans, \mcl{I}_\fodl)$ is the \PLK\ structure of FODL, where 
each world $w\in \Wd_\fodl$ is a mapping $w: \Var_\fodl\to \mbb{Z}$ from the set $\Var_\fodl$ of variables to integer domain. 
An `update' is of the form: $\{x := e\}$ (cf.~\cite{Beckert13}) where $x\in \Var_\fodl$ and $e$ is an expression. 
For a formula $\phi$, an update $\{x := e\}$ on $\phi$ returns a formula by substituting variable $x$ of $\phi$ with $e$. 
A label $\upd\in \Conf_\fodl$ of FODL is a sequence of `updates', 
which applies to a formula $\phi$ by applying each update of $\sigma$ on $\phi$ from right to left. 
For example, a label $\{x := y\}\{x := x + 1\}$ applied to a formula $x = 5$ results in $y + 1 = 5$. 
A label mapping $\rho\in \LM_\fodl$ of FODL is an evaluation on labels by mapping each free variable to a certain value, whose details is further illustrated in Appendix~\ref{Zhang25}. 

Program behaviours of $\Krp_\fodl$ are captured by the inference rules namely $(\pfPT)_\fodl$ in 
Table~\ref{table:Inference Rules for Program Transitions of Regular Programs of FODL} following Definition~\ref{def:Matching Proof System}. 
In Appendix~\ref{???}, a set $(\pfPTer)_\fodl$ of inference rules for the terminations of regular programs is given (see Table~\ref{???}). 
\fi

\subsection{An Encoding of Process Logic in $\DLp$}
\label{section:Encoding of Complex Configurations}

$\DLp$ allows even more complex labels and formulas: the labels can be more than simple program states, while the formulas can be more than simple static ones (which is either true or false at a world). 
Below, we give a sketch of a possible encoding of a first-ordered version of the theory of process logic~\cite{Harel82} (PL) in $\DLp$, namely $\DLpPL$. 
$\DLpPL$ is able to specify and reason about progressive behaviours of programs using temporal formulas. 
Compared to Hoare logic, it is more suitable for reactive systems~\cite{Halbwachs93b}, in which a program may never terminate and we care more about if a property holds on an intermediate state. 


Process logic can be seen as a type of dynamic logics in which the semantics of formulas is defined in terms of paths rather than worlds. 
The form $[\alpha]\phi$ of a PL dynamic formula inherits from that of PDL, where $\alpha$ is a regular program just as in PDL. Formula $\phi$ is a temporal formula defined such that (1) any atomic proposition $p$ is a temporal formula; (2) $\f \phi$ and $\phi \Suf \psi$ are temporal formulas, provided that $\phi$ and $\psi$ are temporal formulas; (3) $\neg\phi$, $\phi\wedge \psi$ are temporal formulas, if $\phi$ and $\psi$ are temporal formulas. 
The semantics of a regular program and a temporal formula is thus given by paths of worlds. 
For a regular program, its semantics corresponds to the set of its execution paths. 
For a temporal formula, given a path $tr$, its semantics is defined as follows:
\begin{enumerate}
    \item $tr\models p$, if $tr_b\models p$;
    \item $tr\models \f \phi$, if $tr_b\models \phi$;
    \item $tr\models \phi\Suf \psi$, if there is a path $tr'$ such that (i) $tr' \psuf tr$ and $tr'\models \psi$, and (ii) for all paths $tr''$ with $tr' \psuf tr'' \psuf tr$, $tr''\models \phi$;
    \item $tr\models \neg\phi$, if $tr\not\models \phi$, and $tr\models \phi\wedge \psi$, if $tr\models \phi$ and $tr\models \psi$. 
\end{enumerate}
Note that the operators $\f$ and $\Suf$ are sufficient to express the meaning of the usual temporal operators (cf.~\cite{Harel82}), for example the ``next'' operator: $\n \phi\dddef \false\Suf \phi$, the ``future'' operator: $\Diamond\phi\dddef \phi \vee (\true\Suf \phi)$, etc.  
Based on the semantics of temporal formulas, the semantics of a dynamic PL formula $[\alpha]\phi$ is defined w.r.t. a path $tr$ as: 
$$
tr\models [\alpha]\phi, \mbox{ if for all execution paths $tr'$ of $\alpha$, $tr'\cdot tr\models \phi$}. 
$$

Our instantiation process is almost the same as $\DLpFODL$, except that we choose $\Fmla$ to be the set of temporal formulas introduced above, denoted by $\Fmla_\pl$. 
And we choose a different set of labels named $\Conf_\pl$ in which each label is defined as a form that captures the meaning of a sequence of the configurations in $\Conf_W$: 
\begin{center}
    $\sigma_1\sigma_2...\sigma_n$ ($n\ge 1$), $\sigma_i\in \Conf_W$ ($1\le i\le n$). 
\end{center}
Besides, similar to the choosing of $(\pfOper)_\fodl$, we propose a set $(\pfOper)_\pl$ of rules for the regular programs. 
In $(\pfOper)_\pl$, the forms of the rules for the program transitions are the same as those in $(\pfOper)_\fodl$ (see Table~\ref{table:Inference Rules for Program Transitions of Regular Programs of FODL}), except that (1) we replace each configuration with the configuration in $\Conf_\pl$, and (2) we have the following rule for assignments: 
$$
\begin{aligned}
\infer[^{(x := e)}]
{\Gamma\Rightarrow (x:=e, l)\trans(\ter, l(\sigma_n)^x_e), \Delta}
{},
\end{aligned}
\mbox{ where $l = \sigma_1\sigma_2...\sigma_n$ ($n\ge 1$)}
$$
where instead we append the current result $(\sigma_n)^x_e$ to the tail of the current sequence $l$ of the configurations in $\Conf_W$.

\ifx
whose partial rules are shown in Table~\ref{table:Inference Rules for Program Transitions of Regular Programs of PL}.
We can see that the forms of the rules in $(\pfOper)_\pl$ for program transitions are the same as those in $(\pfOper)_\fodl$ except the rule $(x:=e)$. 

\begin{table}[tb]
         \begin{center}
         \noindent\makebox[\textwidth]{%
         \scalebox{1}{
         \begin{tabular}{c}
         \toprule
        $
         \infer[^{(x:=e)}]
         {\Gamma\Rightarrow (x:=e, l)\trans(\ter, \sigma^x_e), \Delta}
         {
         }
         $
         \ \ 
         $
         \infer[^{(\phi?)}]
         {\Gamma\Rightarrow (\phi?, \sigma)\trans(\ter, \sigma), \Delta}
         {
         \Gamma\Rightarrow \sigma : \phi, \Delta
         }
         $
         \\
         $
         \infer[^{1\ (;)}]
         {\Gamma\Rightarrow (\alpha\seq\beta, \sigma)\trans(\alpha'\seq\beta, \sigma'), \Delta}
         {
            \Gamma\Rightarrow (\alpha, \sigma)\trans(\alpha', \sigma'), \Delta
         }
         $
         \ \  
         $
         \infer[^{(;\ter)}]
         {\Gamma\Rightarrow (\alpha\seq\beta, \sigma)\trans(\beta, \sigma'), \Delta}
         {
            \Gamma\Rightarrow (\alpha, \sigma)\trans(\ter, \sigma'), \Delta
         }
         $
         \\
         $
         \infer[^{(\cho 1)}]
         {\Gamma\Rightarrow (\alpha\cho\beta, \sigma)\trans(\alpha, \sigma), \Delta}
         {
         }
         $
         \ \ 
        $
         \infer[^{(\cho 2)}]
         {\Gamma\Rightarrow (\alpha\cho\beta, \sigma)\trans(\beta, \sigma), \Delta}
         {
         }
         $
         \\
         $
         \infer[^{2\ (*)}]
         {\Gamma\Rightarrow (\alpha^\lup, \sigma)\trans(\alpha\seq\alpha^\lup\cho \true?, \sigma), \Delta}
         {
         }
         $
         \\
         \bottomrule
         \end{tabular}
              }
              }
          \end{center}
          \caption{Partial Rules of $(\pfOper)_\pl$ for Program Transitions of Regular Programs}
          \label{table:Inference Rules for Program Transitions of Regular Programs of PL}
    \end{table}
\fi


Consider an example of programs' temporal properties:
$$
\sigma : [\alpha\seq \alpha]\Diamond x > 0, 
$$
with $\sigma = \{x\mapsto -1\}$, $\alpha = (x := x + 1)$. 
It says that along the execution path of $\alpha\seq \alpha$, $x > 0$ eventually holds. 
By the rules in $(\pfOper)_\pl$, we have the following derivation: 
$$
\begin{aligned}
    \infer[^{([\alpha]R)}]
    {\cdot \Rightarrow \sigma : [\alpha\seq \alpha]\Diamond x > 0}
    {
        \infer[^{([\alpha]R)}]
        {\cdot \Rightarrow \sigma\sigma' : [\alpha]\Diamond x > 0}
        {
            \infer[^{([\ter])}]
            {\cdot \Rightarrow \sigma\sigma'\sigma'' : [\ter]\Diamond x > 0}
            {
                \cdot \Rightarrow \sigma\sigma'\sigma'' : \Diamond x > 0
            }
        }
    }
\end{aligned}, 
$$
where $\sigma' = \{x\mapsto 0\}$, $\sigma'' = \{x\mapsto 1\}$. 

\renewcommand{\arraystretch}{1.5}
\begin{table}[tb]
     \begin{center}
     \noindent\makebox[\textwidth]{%
     \scalebox{1}{
     \begin{tabular}{c}
     \toprule
     $
    \infer=[^{(\f)}]
    {l : \f \phi}
    {l_b : \phi}
    $
    \ \ 
    $
    \infer[^{(\Suf R 1)}]
    {\Gamma\Rightarrow \sigma l : \phi\Suf \psi, \Delta}
    {\Gamma\Rightarrow l : \phi, \Delta
    &
    \Gamma\Rightarrow l : \phi\Suf\psi, \Delta
    }
    $
    \ \ 
    $
    \infer[^{(\Suf R 2)}]
    {\Gamma\Rightarrow \sigma l : \phi\Suf \psi, \Delta}
    {
    \Gamma\Rightarrow l : \psi, \Delta
    }
    $
    \\
    $
    \infer[^{(\Suf L)}]
    {\Gamma, \sigma l : \phi\Suf \psi\Rightarrow \Delta}
    {\Gamma, l : \phi, l : \phi\Suf\psi\Rightarrow \Delta
    &
    \Gamma, l : \psi\Rightarrow \Delta
    }
    $
    \ifx
    \ \ 
    $
    \infer[^{(\Suf L 2)}]
    {\Gamma, \sigma l : \phi\Suf \psi\Rightarrow \Delta}
    {\Gamma, l : \psi\Rightarrow \Delta
    }
    $
    \fi
     \\
     \midrule
 \multicolumn{1}{l}{
        in the above rules, $l\in \Conf_W\Conf^*_W$, $\sigma \in \Conf_W$. 
     }
        \\
     \bottomrule
     \end{tabular}
          }
          }
      \end{center}
      \caption{Rules for Labeled Temporal Formulas in $\DLpPL$}
      \label{table:Additional Rules for Labeled Temporal Formulas in DLpPL}
\end{table}

Unlike the formulas of $\Fmla_\afo$ in $\DLpWP$ and $\DLpFODL$, in $\DLpPL$ we can further derive labeled temporal formulas (e.g. $\sigma\sigma'\sigma'' : \Diamond x > 0$) using the following additional rules shown in Table~\ref{table:Additional Rules for Labeled Temporal Formulas in DLpPL}. 
These rules are directly according to the semantics of the operators $\f$ and $\Suf$. 
\ifx
In Table~\ref{table:Additional Rules for Labeled Temporal Formulas in DLpPL}, the rule $(\f)$ is according to the semantics of the operator $\f$. The rules $(\Suf R)$ and $(\Suf L)$ are directly from the logical equations $\phi\Suf \psi \leftrightarrow ()$
\fi

A cyclic derivation for iterative programs (like $\alpha^*$) in $\DLpPL$, however, requires higher-ordered label structures together with a suitable instantiation of the abstract substitutions as defined in Definition~\ref{def:Substitution of Labels}. 
Our future work will discuss more about it, as well as the analysis of the (relative) completeness of $\DLpPL$ and its comparison to the traditional theory of PL. 


\section{Analysis of Soundness and Completeness of $\DLp$}
\label{section:Proof of Theorem - theo:Soundness of A Cyclic Preproof}


\ifx
In this section, we study the soundness and completeness of the proof system $\pfDLp$. 
Section~\ref{section:Conditional Soundness of DLp} discusses about the soundness of $\pfDLp$, and Section~\ref{section:Conditional Completeness of DLp} discusses about the completeness of $\pfDLp$. 
\fi

In this section, we analyze the soundness and completeness of the proof system $\pfDLp$. 
Currently, we consider the soundness under a restriction on the program behaviours of $\Prog$ (Definition~\ref{def:Program Properties}). 
However, as analyzed below in detail, the set of programs under the restriction is still a rich one. 
For the completeness, since $\DLp$ is not a specific logic, generally, it is impossible to discuss about its completeness without any restrictions on the parameters of $\DLp$. 
Instead, we study under which conditions (Definition~\ref{def:Expression Finiteness Property} and \ref{def:Well-behaved Loop Programs}) can we obtain a completeness result relative to the labeled non-dynamic formulas. 

Section~\ref{section:Conditional Soundness of DLp} discusses about the soundness of $\pfDLp$, while Section~\ref{section:Conditional Completeness of DLp} discusses about its completeness.

\subsection{Conditional Soundness of $\DLp$}
\label{section:Conditional Soundness of DLp}

We first introduce the concept of \emph{minimum execution paths}. 

\begin{definition}
    An execution path (Definition~\ref{def:Execution Path}) $w_1...w_n$ ($n\ge 1$) is called ``minimum'', if there are no two relations $w_i\xrightarrow{\alpha_i/\cdot}\cdot$ and $w_j\xrightarrow{\alpha_j/\cdot}\cdot$ for some $1\le i < j < n$ such that 
$w_i = w_j$ and $\alpha_i = \alpha_j$. 
\end{definition}

Intuitively, in a minimum execution path, there are no two relations starting from the same world and program. 

The restriction condition is stated in the following definition. 
\begin{definition}[Termination Finiteness]
\label{def:Program Properties}
Starting from a world $w\in \Wd$ and a program $\alpha\in \Prog$, there is only a finite number of minimum execution paths (of the form:  $w\xrightarrow{\alpha/\cdot}...$). 
\ifx
In \PLK\ structure $\Kr$,  a program $\alpha\in \Prog$ satisfies the ``termination finiteness'' property, if for a world $w\in \Wd$, 
        there is only a finite number of minimum execution paths starting from a relation of the form $w\xrightarrow{\alpha/\cdot}\cdot$. 
        \fi
\end{definition}

The programs satisfying termination finiteness are in fact a rich set, including, for example, all the programs whose behaviour is deterministic, such as while programs discussed in this paper, programming languages like Esterel, C, Java, etc.  
There exist non-deterministic programs that obviously fall into this category. 
For example, automata that have non-deterministic transitions but have a finite number of states. 
More on this restriction will be discussed in our future work.



\begin{theorem}[Conditional Soundness of $\DLp$]
\label{theo:Soundness of A Cyclic Preproof}
If the programs in $\Prog$ satisfy the termination finiteness property, then
for any labeled formula $\sigma : \phi\in \DLpF$, $\pfDLp\vdash (\cdot \Rightarrow\sigma : \phi)$ implies $\models \sigma : \phi$. 
\end{theorem}

\textbf{Main Idea for Proving the Soundness}. 
We follow the main idea behind~\cite{Brotherston08} to prove Theorem~\ref{theo:Soundness of A Cyclic Preproof} by contradiction. 
The key point is that, 
if the conclusion of a cyclic proof is invalid, then by the soundness of all the rules in $\pfLDLp$ (Theorem~\ref{theo:soundness of rules for LDL}), 
there must exist an \emph{invalid derivation path} in which each node is invalid, and one of its progressive traces leads to an infinite descent sequence of some well-founded set (introduced below), which violates the definition of the well-foundedness (cf.~\cite{Dershowitz79}) itself. 

\ifx
We need to show that
if a preproof is cyclic (Definition~\ref{def:Cyclic Preproof}), that is, if every derivation path of the preproof is followed by a  progressive trace, then the conclusion is valid. 
We follow the main idea behind~\cite{Brotherston08} to carry out a proof by contradiction: 
suppose the conclusion is invalid, then 
we will show that there exists an \emph{invalid derivation path} in which each node is invalid, 
and 
one of its progressive traces starting from a formula $\tau$ in a node $\nu$ leads to an infinite descent sequence of elements along this trace ordered by a well-founded relation $\mult$ (introduced below), which violates the definition of well-foundedness itself. 
\fi


\ifx
In this subsection, we analyze and prove Theorem~\ref{theo:Soundness of A Cyclic Preproof}. 
We only focus on the case when the conclusion is of the form $\Gamma \Rightarrow \sigma : [\alpha]\phi$ or $\Gamma\Rightarrow \sigma : \la\alpha\ra$, where dynamic $\DLp$ formulas $ \sigma : [\alpha]\phi$ and $\sigma : \la\alpha\ra\phi$ are the only formula on the right side of the sequent.
Other cases are trivial.

To prove Theorem~\ref{theo:Soundness of A Cyclic Preproof}, we need to show that
if a preproof is cyclic (Definition~\ref{def:Cyclic Preproof}), that is, if any derivation path is followed by a progressive trace, then the conclusion is sound. 
We carry out the proof by contradiction following the main idea behind~\cite{Brotherston08}: Suppose the conclusion is not sound, that is, proposition $\mfr{P}(\Gamma\Rightarrow \sigma: [\alpha]\phi)$ (resp. $\mfr{P}(\Gamma\Rightarrow\sigma: \la\alpha\ra\phi)$) is not true,
then we will show that it induces an infinite descent sequence of elements ordered by 
the well-founded relation $\mult$ introduced as follows, which violates Definition~\ref{def:Well-foundedness}. 
\fi
Below we firstly introduce the well-founded relation $\mult$ we rely on, then we focus on the main skeleton of proving Theorem~\ref{theo:Soundness of A Cyclic Preproof}. 
Other proof details are given in Appendix~\ref{section:Other Propositions and Proofs}. 


\textbf{Well-foundedness \& Relation $\mult$}. 
\ifx
A relation $\preceq$ on a set $S$ is \emph{partially ordered}, if 
it satisfies the following properties: 
(1) Reflexivity. $t\preceq t$ for each $t\in S$. 
(2) Anti-symmetry. For any $t_1, t_2\in S$, if $t_1\preceq t_2$ and $t_2\preceq t_1$, then $t_1$ and $t_2$ are the same element in $S$, we denote by $t_1 = t_2$.
(3) Transitivity. For any $t_1, t_2, t_3\in S$, if $t_1\preceq t_2$ and $t_2\preceq t_3$, then $t_1\preceq t_3$. 
$t_1 \prec t_2$ is defined as $t_1\preceq t_2$ and $t_1\neq t_2$. 
\fi
    Given a set $S$ and a partial-order relation $\preceq$ on $S$, 
    $\preceq$ is called a \emph{well-founded relation} over $S$, if for any element $a$ in $S$, there is no infinite descent sequence: $a \succ a_1 \succ a_2 \succ ...$ in $S$. 
   Set $S$ is called a \emph{well-founded} set w.r.t. $\preceq$. 

\ifx
Given two paths $tr_1$ and $tr_2$, 
relation $tr_1 \suf tr_2$ is defined if $tr_2$ is a suffix of $tr_1$. 
Obviously $\suf$ is partially ordered. 
Relation $tr_1\psuf tr_2$ expresses that $tr_1$ is a proper suffix of $tr_2$. 
In a set of finite paths, relation $\suf$ is well-founded,  because every finite path has only a finite number of suffixes. 
\fi

\begin{definition}[Relation $\mult$]
\label{def:Relation pmult}
    Given two finite sets $\cnt_1$ and $\cnt_2$ of finite execution paths,  
    $\cnt_1\mult \cnt_2$ is defined if either (1) $\cnt_1 = \cnt_2$; or (2) set $\cnt_1$ can be obtained from $\cnt_2$ by replacing one or more elements of $\cnt_2$ each with a finite number of elements, such that
    for each replaced element $tr$, its replacements $tr_1,...,tr_n$ ($n\ge 1$) in $\cnt_1$ are proper suffixes of $tr$.  

\end{definition}

\ifx
\begin{proposition}
\label{prop:relation mult is partially ordered}
$\mult$ is a partial-order relation. 
\end{proposition}
\fi

In Definition~\ref{def:Relation pmult}, note that we can replace an element of $\cnt_2$ with an \emph{empty execution path} whose length is $0$. And if we do so, it is equivalent to that we remove an element from $\cnt_2$. 

\begin{proposition}
\label{prop:relation mult is partially ordered}
    $\mult$ is a partial-order relation. 
\end{proposition}
The proof of Proposition~\ref{prop:relation mult is partially ordered} is given in Appendix~\ref{section:Other Propositions and Proofs}. 


\begin{example}
    Let $C_1=\{tr_1, tr_2, tr_3\}$, where $tr_1 \dddef w w_1 w_2 w_3 w_4, tr_2 \dddef w w_1w_5w_6w_7$ and $tr_3 \dddef ww_8$; 
$C_2 = \{tr'_1, tr'_2\}$, where 
$tr'_1 \dddef w_1 w_2 w_3 w_4, tr'_2 \dddef w_1 w_5w_6w_7$. 
We see that $tr'_1$ is a proper suffix of $tr_1$ and $tr'_2$ is a proper suffix of $tr_2$. 
$C_2$ can be obtained from $C_1$ by replacing $tr_1$ and $tr_2$ with $tr'_1$ and $tr'_2$  respectively, and removing $tr_3$. 
Hence $C_2\mult C_1$. 
Since $C_1\neq C_2$, $C_2\pmult C_1$. 
\end{example}

\begin{proposition}
\label{prop:well-foundedness of relation pmult}
    Relation $\mult$ is a  well-founded relation. 
\end{proposition}
We omit the proof of Proposition~\ref{prop:well-foundedness of relation pmult}. 
Relation $\mult$ is just a special case of the ``multi-set ordering'' introduced in~\cite{Dershowitz79}, where it has been proved to be well-founded. 
Intuitively, we observe that for two sets $\cnt_1$ and $\cnt_2$ such that $\cnt_1\pmult \cnt_2$, 
for each set $D_{tr}$ of the paths in $\cnt_1$ that replaces an element $tr$ in $\cnt_2$, 
the maximum length of the elements of $D_{tr}$ is strictly smaller than that of $tr$. 
By that $\cnt_2$ is finite, we can see that such a replacement decreases the number of the paths that have the maximum length of the elements in $\cnt_2$.

\ifx
The proof of Proposition~\ref{prop:well-foundedness of relation pmult} is omitted since 
relation $\mult$ is just the ``multi-set ordering'' introduced in~\cite{Dershowitz79}, where it has been proved that multi-set ordering is well-founded. 
\fi

\textbf{Proof Skeleton of Theorem~\ref{theo:Soundness of A Cyclic Preproof}}. 
Below we give the main skeleton of the proof by skipping the details of the proof of Lemma~\ref{lemma:infinite descent sequence}, which can be found in Appendix~\ref{section:Other Propositions and Proofs}. 

Following the main idea above, 
we first introduce the concept of the ``execution paths of a dynamic $\DLp$ formula''.  
They are the elements of a well-founded relation $\mult$. 
Next, we propose Lemma~\ref{lemma:infinite descent sequence}, which plays a key role in the proof of Theorem~\ref{theo:Soundness of A Cyclic Preproof} that follows. 

\ifx
Following the main idea at the beginning of Section~\ref{section:Proof of Theorem - theo:Soundness of A Cyclic Preproof}, 
next we first introduce the notion of ``counter-example paths'' that makes a formula $\sigma : [\alpha]\phi$ or $\sigma : \la\alpha\ra \phi$ invalid.   
\fi

\ifx
\begin{definition}[Minimum Paths]
An execution path (Definition~\ref{def:Execution Path}) $s_1...s_n$ ($n\ge 1$) is `minimum' if there are no two transitions $s_i\xrightarrow{\alpha_i/\cdot}\cdot$ and $s_j\xrightarrow{\alpha_j/\cdot}\cdot$ for some $1\le i < j < n$ such that 
$s_i = s_j$ and $\alpha_i = \alpha_j$. 
\end{definition}
Intuitively, in a minimum execution path, there are no two transitions starting from the same world and program. 
\fi

\begin{definition}[Execution Paths of Dynamic Formulas]
\label{def:counter-example set}
    Given a world $w\in \Wd$ and a dynamic formula $\phi$, 
    the execution paths $\EX(w, \phi)$ of $\phi$ w.r.t. $w$ is inductively defined according to the structure of $\phi$ as follows:
    \begin{enumerate}
    \item $\EX(w, [\alpha]F)\dddef \mex(w, \alpha)$, where $F\in \Fmla$;
    \item $\EX(w, [\alpha]\phi_1)\dddef \mex(w, \alpha)\cup \{tr_1\cdot tr_2\ |\ tr_1\in \mex(w, \alpha), tr_2\in \EX((tr_1)_e, \phi_1)\}$;
    \item $\EX(w, \neg\phi_1)\dddef \EX(w, \phi_1)$;
    \item $\EX(w, \phi_1\wedge \phi_2)\dddef \EX(w, \phi_1)\cup \EX(w, \phi_2)$.
    \end{enumerate}
    Where $\mex(w, \alpha)\dddef \{w...w'\ |\ \mbox{$w\xrightarrow{\alpha/\cdot}...\xrightarrow{\cdot/\ter}w'$ is a min. exec. path for some $w'\in \Wd$}\}$ is the set of all minimum paths of $\alpha$ starting from world $w$. 
\end{definition}


In Definition~\ref{def:counter-example set}, an execution path of a dynamic formula may be concatenated by several execution paths that belong to different programs in a sequence of modalities. 
As seen in the proof of Lemma~\ref{lemma:infinite descent sequence} (Appendix~\ref{section:Other Propositions and Proofs}), this consideration is necessary because a dynamic formula may contain more than one modality (e.g. $[\alpha][\beta]\phi$). 
It is also one of the main differences between our proof and the proof given in~\cite{zhang2025parameterizeddynamiclogic}.

In the following, we call $\lm\in \LM$ a \emph{counter-example mapping} of a node $\nu$, if it makes $\nu$ invalid. 

\begin{lemma}
    \label{lemma:infinite descent sequence}
    In a cyclic proof (where there is at least one derivation path), 
    let $(\sigma:\phi, \sigma':\phi')$ be a step of a derivation trace over a derivation $(\nu, \nu')$ of an invalid derivation path, where $\phi, \phi'\in \DLF$.  
    For any set $\EX(\lm(\sigma), \phi)$ of $\sigma : \phi$ w.r.t. a counter-example mapping $\lm$ of $\nu$, 
    there exists a counter-example mapping $\lm'$ of $\nu'$ and a set $\EX(\lm'(\sigma'), \phi')$ of $\sigma':\phi'$ such that $\EX(\lm'(\sigma'), \phi')\mult \EX(\lm(\sigma), \phi)$. Moreover, if $(\sigma:\phi, \sigma':\phi')$ is a progressive step, then 
    $\EX(\lm'(\sigma'), \phi')\pmult \EX(\lm(\sigma), \phi)$. 
\end{lemma}


Intuitively, Lemma~\ref{lemma:infinite descent sequence} helps us discover suitable execution-path sets imposed by a well-founded relation $\mult$ between them in an invalid derivation path. 

Based on Proposition~\ref{prop:well-foundedness of relation pmult} and Lemma~\ref{lemma:infinite descent sequence}, we give the proof of Theorem~\ref{theo:Soundness of A Cyclic Preproof} as follows. 

\begin{proof}[Proof of Theorem~\ref{theo:Soundness of A Cyclic Preproof}]
    Let $\nu = (\cdot \Rightarrow \sigma : \phi)$. 
    By contradiction, suppose $\not\models \sigma : \phi$, that is, $\nu$ is invalid. 
    Then by the soundness of each rule in $\pfDLp$ (Theorem~\ref{theo:soundness of rules for LDL}), there exists an invalid derivation path $p$ from $\nu$ (where every sequent is invalid). 
    Since $\pfDLp\vdash \nu$ (i.e., a cyclic proof tree is formed to prove the conclusion $\nu$), 
    let $\tau_1\tau_2...\tau_k...$ be a progressive trace over $p$ of the form: 
    $\nu...\nu_1\nu_2...\nu_k...$ ($k\ge 1$), where each formula $\tau_i$ is in $\nu_i$ ($i\ge 1$).  
    Let $\tau_i\dddef \sigma_i : \phi_i$.  

    Since $\nu_1$ is invalid, let $\lm_1$ be one of its counter-example mappings. 
    By Lemma~\ref{lemma:infinite descent sequence}, 
    from $\EX(\lm_1(\sigma_1), \phi_1)$, there exists an infinite sequence of sets $\EX_1,...,\EX_k,...$ $(k\ge 1)$, 
    where each $\EX_i\dddef \EX(\lm_i(\sigma_i), \phi_i)$ $(i\ge 1)$ with $\lm_i$ a counter-example mapping of node $\nu_i$, 
    and which satisfies that 
    $\EX_1\multr...\multr \EX_k\multr...$.
    Moreover, since trace $\tau_1\tau_2...\tau_k...$ is progressive (Definition~\ref{def:Progressive Step/Progressive Derivation Trace}), 
    there must be an infinite number of $j\ge 1$ 
    such that $\EX_j\pmultr \EX_{j+1}$. 
    This thus forms an infinite descent sequence w.r.t. $\pmult$, violating 
    the well-foundedness of relation $\mult$ (Proposition~\ref{prop:well-foundedness of relation pmult}). 
    
    \ifx
    From that $\Gamma\Rightarrow \tau_1, \Delta$ is invalid, 
    there is an evaluation $\rho$ satisfying that $\rho\models \Gamma$ but $\rho\not\models \tau_1$, so $\CT_1(\rho, \tau_1, \nu_1)$ is well defined. 
    By Lemma~\ref{lemma:infinite descent sequence}, 
    from $\CT_1$ we can obtain an infinite sequence of counter-examples: 
    $\CT_1,\CT_2,...,\CT_m,...$, with each $\CT_i$ ($i\ge 1$) being the counter-example of formula $\tau_i$, and which satisfies that
    $\CT_1\multr \CT_2\multr...\multr \CT_m\multr...$. 
    Moreover, since $\tau_1\tau_2...\tau_m...$ is progressive (Definition~\ref{def:Cyclic Preproof}), 
     there must be an infinite number of strict relation $\pmult$ among these $\mult$s. This thus forms an infinite descent sequence w.r.t. $\pmult$, violating 
    the well-foundedness of relation $\mult$ (Proposition~\ref{prop:well-foundedness of relation pmult}). 
    \fi
\end{proof}


\ifx
Therefore, it is impossible to discuss the completeness of its proof system, unless the structures of formulas, programs and configurations are given specifically.
Whether the proof system of a ``specific'' $\DLp$, which is given by specific formulas, programs and configurations, is complete or not, depends on whether we can build a cyclic preproof for any formula in this logic, by constructing suitable configurations so that every infinite derivation path has a progressive derivation trace.
\fi

\subsection{Conditional Completeness of $\DLp$}
\label{section:Conditional Completeness of DLp}

We propose two sufficient conditions for the relative completeness of $\DLp$: 
1) that the program models of $\DLp$ always have finite expressions (Definition~\ref{def:Expression Finiteness Property}); 
and 2) that their \emph{loop programs} are always \emph{well-behaved} (Definition~\ref{def:Well-behaved Loop Programs}) during the reasoning process.  
A loop program is a program that eventually reaches itself during a sequence of symbolic-execution reasoning under a label. 

Below we first introduce these conditions, under them we then prove the relative completeness of $\DLp$. 
We only give an outline and put the technical details of the proof in Appendix~\ref{section:Other Propositions and Proofs}.

\ifx
We propose a sufficient condition for the relative completeness of $\DLp$ --- a so-called \emph{well-behaved} property for \emph{loop programs}. 
A loop program eventually reaches itself during a sequence of symbolic-execution reasoning under a label. 
We show that $\DLp$ is relatively complete if the program models of $\DLp$ always have finite expressions (Definition~\ref{def:Expression Finiteness Property}) and their loop programs are always well-behaved (Definition~\ref{def:Well-behaved Loop Programs}) during the reasoning process. 

Below we first introduce the key concepts of the completeness condition (Definition~\ref{def:Expression Finiteness Property} and~\ref{def:Well-behaved Loop Programs}), then we prove the relative completeness of $\DLp$ under the condition. 
We put the technical details of the proof (the proof of Lemma~\ref{def:Well-behaved Loop Programs}) in Appendix~\ref{section:Other Propositions and Proofs} but only give an outline here. 
\fi


\begin{definition}[Program Sequences]
    Given a context $\Gamma$, a program $\alpha\in \Prog$ and a label $\sigma\in \Conf$, a (potentially infinite) sequence: $\Gamma : (\alpha_1, \sigma_1, \Gamma_1)(\alpha_2,\sigma_2, \Gamma_2)...(\alpha_n, \sigma_n, \Gamma_{n})...$ ($1\le n < \infty$, $\alpha_1 = \alpha$, $\sigma_1 = \sigma$), called an ``$\alpha$ sequence'',  is defined if there is
    a sequence of derivations in system $\pfDLp$ as follows: 
    $\pfDLp \vdash (\Gamma_1\Rightarrow (\alpha_1,\sigma_1)\trans(\sigma_2, \sigma_2))$, 
    $\pfDLp \vdash (\Gamma_2\Rightarrow (\alpha_2,\sigma_2)\trans(\sigma_3, \sigma_3))$, 
    ...,
    $\pfDLp \vdash (\Gamma_{n-1} \Rightarrow(\alpha_{n-1}, \sigma_{n-1})\trans(\alpha_n, \sigma_n))$, ..., which satisfies that 
    $\models (\Gamma_1\Rightarrow \Gamma)$ and $\models (\Gamma_n\Rightarrow \Gamma_{n-1})$ for all $n\ge 2$. 

    We call a sequence $\Gamma : (\alpha_1, \sigma_1)(\alpha_2, \sigma_2)...(\alpha_n, \sigma_n)...$ a ``core $\alpha$ sequence'' if there exist $\Gamma_1,\Gamma_2,...,\Gamma_n,...$ such that  $\Gamma : (\alpha_1, \sigma_1, \Gamma_1)(\alpha_2,\sigma_2, \Gamma_2)...(\alpha_n, \sigma_n, \Gamma_{n})...$ is an $\alpha$ sequence. 
    
    \ifx
    a derivation branch in system $\pfDLp$ as follows: 
    $$
    \infer[]
    {\Gamma_1\Rightarrow (\alpha_1, \sigma_1)\trans(\alpha_2,\sigma_2)}
    {
        ...
        &
        \infer[]
        {\Gamma_2\Rightarrow (\alpha_2, \sigma_2)\trans (\alpha_3,\sigma_3)}
        {
            \infer*[]
            {...}
            {
                \infer[]
                {...}
                {
                    ...
                    &
                    \Gamma_{n-1}\Rightarrow(\alpha_{n-1}, \sigma_{n-1})\trans(\alpha_n, \sigma_n)
                    &
                    ...
                }
            }
        }
        &
        ...
    }, 
    $$
    where each derivation step is performed by using some rules in $\pfDLp$;
    And it satisfies that $\Gamma_k\Rightarrow \Gamma_{k-1}$ for all $2\le k\le n-1$. 
    \fi
\end{definition}

Intuitively, starting from $(\alpha, \sigma)$ under context $\Gamma$, an $\alpha$ sequence is a sequence of derivations where in each step, the context can only be strengthen from $\Gamma$. 
By the soundness of $\pfDLp$ w.r.t. $\PT$ (see Definition~\ref{def:Matching Proof System}), each derivation of an $\alpha$ sequence is actually a symbolic execution step of the program. 

\begin{definition}[Program Loop Sequences]
    An ``$\alpha$-loop sequence'' of a program $\alpha\in \Prog$ is an $\alpha$ sequence: 
    $\Gamma : (\alpha_1, \sigma_1, \Gamma_1)...(\alpha_n,\sigma_n, \Gamma_n)$ ($n\ge 1$) such that $\alpha_1 = \alpha_n = \alpha$, and $\alpha_i \neq \alpha_j$ for any other $\alpha_i$ and $\alpha_j$, with $1 \le i < j \le n$. 
    
\end{definition}

\begin{example}
    In the instantiation theory $\DLpFODL$ as defined in Section~\ref{section:Instantiation of FODL in DLp}, let $\alpha\dddef (x := x + 1)$, $\beta\dddef (y := 0)$, $\sigma =\{x\mapsto t\}$, then for the program $\alpha^*\seq \beta$, we have an $\alpha$ sequence: $\emptyset : (\alpha^*\seq \beta, \sigma, \emptyset)((\alpha\seq\alpha^*\cup \true?)\seq \beta, \sigma, \emptyset)(\alpha\seq \alpha^*\seq \beta, \sigma, \emptyset)(\alpha^*\seq \beta, \sigma', \emptyset)$, 
    where $\sigma' = \{x\mapsto t + 1\}$. 
    It corresponds to the following derivations: 
    $$
    \begin{aligned}
        &(\pfDLp)_\fodl\vdash (\cdot \Rightarrow (\alpha^*\seq \beta, \sigma)\trans ((\alpha\seq\alpha^*\cup \true?), \sigma),\\
        &(\pfDLp)_\fodl\vdash (\cdot \Rightarrow ((\alpha\seq\alpha^*\cup \true?)\seq \beta, \sigma)\trans (\alpha\seq \alpha^*\seq \beta, \sigma)),\\
        &(\pfDLp)_\fodl\vdash (\cdot \Rightarrow (\alpha\seq \alpha^*\seq \beta, \sigma)\trans(\alpha^*\seq \beta, \sigma'))\\
    \end{aligned}
    $$
    according to the corresponding rules in Table~\ref{table:Inference Rules for Program Transitions of Regular Programs of FODL}. 
    This $\alpha$ sequence is also a loop one. 
    
\end{example}

\begin{definition}[Expression Finiteness Property]
\label{def:Expression Finiteness Property}
    For a program $\alpha\in \Prog$, there is a natural number $N_{\alpha}$ such that
    in each $\alpha$ sequence under a label $\sigma\in \Conf$ and a context $\Gamma$: $\Gamma : (\alpha_1, \sigma_1, \Gamma_1)...(\alpha_n,\sigma_n, \Gamma_n)...$ ($1\le n < \infty$, $\alpha_1 = \alpha$, $\sigma_1 = \sigma$), the number of the different programs among $\alpha_1,...,\alpha_n,...$ is no greater than $N_\alpha$. 
    \ifx
    starting from a world $w\in \Wd$ and a program $\alpha\in \Prog$, for any (potentially infinite) path: $w\xrightarrow{\alpha/\alpha_1}w_1\xrightarrow{\alpha_1/\alpha_2}...\xrightarrow{\alpha_{n-1}/\alpha_n}w_n\xrightarrow{\alpha_n/\alpha_{n+1}}...$, the set $\{\alpha, \alpha_1,\alpha_2,...,\alpha_{n-1}, \alpha_n, \alpha_{n+1},...\}$ of programs appearing along the path is finite. 
    \fi
\end{definition}

Definition~\ref{def:Expression Finiteness Property} means that as a program proceeds, it eventually reaches to the form of itself in a limit number of steps. 
Most program models in practice satisfy this property. 
But there are exceptions, for example, 
the programs in $\pi-$calculus with the replication operator (cf.~\cite{Milner92}).

\ifx
\begin{definition}[Loop Programs]
    A program $\alpha\in \Prog$ is called a ``loop program'', if there exists an $\alpha$-loop sequence under some context. 
\end{definition}
\fi

\begin{definition}[Well-behaved Loop Programs]
\label{def:Well-behaved Loop Programs}
    A program $\alpha\in \Prog$ is called a ``loop program'', if there exists an ``$\alpha$-loop sequence'' for some label and context. 

    A loop program $\alpha$ is ``well-behaved'', if for any label $\sigma\in \Conf$ and context $\Gamma$, 
    there exist a label $\sigma'\in \Conf$, a context $\Gamma'$ and a substitution $\eta : \Conf\to \Conf$ satisfying the following conditions: 
    \begin{enumerate}
        \item $\sigma = \eta(\sigma')$ and $\Gamma = \eta(\Gamma')$;
        \item $\models (\Gamma\Rightarrow \sigma\termi \alpha)$ implies $\models (\Gamma'\Rightarrow \sigma'\termi\alpha)$;
        \item 
        For each $\alpha$-loop sequence: $\Gamma' : (\alpha, \sigma', \Gamma_1)...(\alpha, \sigma'', \Gamma_n)$ ($n\ge 1$), 
        there exist a context $\Gamma''$ and a substitution $\xi : \Conf\to \Conf$ such that $\Gamma'' = \xi(\Gamma')$, $\sigma'' = \xi(\sigma')$ and $\models (\Gamma_n\Rightarrow \Gamma'')$. 
    \end{enumerate}
\end{definition}

The well-behaved property describes the ``ability'' of a loop program $\alpha$ to form back-links along its symbolic executions. 
To be more specific, for a label $\sigma\in \Conf$, it is possible to find a suitable label $\sigma'$ from which $\sigma$ can be obtained through substitutions, and for every symbolic execution starting from $(\alpha, \sigma')$, we can go back to $(\alpha, \sigma')$ through substitutions. 
The $\sigma'$ here plays the same role of the loop invariants in the normal deduction approaches of program logics. 

\begin{theorem}[Conditional Completeness of $\DLp$]
\label{theo:Completeness of DLp}
    If the programs in $\Prog$ satisfy the expression finiteness property and among them all loop programs are well-behaved, then for any labeled formula $\sigma : \phi\in \DLpF$, $\models \sigma : \phi$ implies $\pfDLp\vdash (\cdot \Rightarrow\sigma : \phi)$.  
\end{theorem}

Note that the relativeness of the completeness result to labeled non-dynamic formulas is reflected by the rule $(\textit{Ter})$ (Table~\ref{table:General Rules for LDL}). 

\textbf{Main Idea for Proving the Completeness}. 
To prove Theorem~\ref{theo:Completeness of DLp}, we firstly reduce it to the special case of deriving a sequent of the form $\Gamma\Rightarrow [\la\alpha\ra]\phi$ as shown in Lemma~\ref{lemma:completeness of DLp a more general case}. For this step we take a similar approach from~\cite{Harel79}. 
The main technical part is deriving the sequent $\Gamma\Rightarrow \sigma : [\la\alpha\ra]\phi$. 
We proceed by a simultaneous induction on the number of the modalities and maximum number of the forms of the programs that can appear during the derivation of the sequent. 
The critical observation is that by the finiteness when executing $\alpha$ (Lemma~\ref{lemma:Finite Loop Sequences}) and the well-behaved property (Definition~\ref{def:Well-behaved Loop Programs}) it satisfies, each non-terminal derivation branch from $\Gamma\Rightarrow \sigma : [\la\alpha\ra]\phi$ is able to form a back-link, which in-turn shows that the whole derivation of $\Gamma\Rightarrow \sigma : [\la\alpha\ra]\phi$ can form a cyclic proof. 

We put the proof of Lemma~\ref{lemma:completeness of DLp a more general case} in Appendix~\ref{section:Other Propositions and Proofs} in details.



\begin{lemma}
\label{lemma:completeness of DLp a more general case}
    Under the same conditions as in Theorem~\ref{theo:Completeness of DLp}, for any valid sequent of the form: $\Gamma\Rightarrow \sigma : [\la\alpha \ra]\phi$, $\pfDLp \vdash (\Gamma\Rightarrow \sigma : [\la\alpha \ra]\phi)$. 
\end{lemma}

In Lemma~\ref{lemma:completeness of DLp a more general case}, ``$[\la\alpha\ra]$'' just means either $[\alpha]$ or $\la\alpha\ra$.

\begin{lemma}
\label{lemma:Finite Loop Sequences}
    Starting from a program $\alpha\in \Prog$ and a label $\sigma\in \Conf$ under a context $\Gamma$, there is only a finite number of core $\alpha$-loop sequences.     
\end{lemma}

\begin{proof}
    By the finiteness of set $\pfOper$, fixing a program $\beta$, there is a maximum number of transitions starting from $(\beta, l)$ in the form of: $\Gamma \Rightarrow (\beta, l)\trans ...$ for all labels $l\in \Conf$ and contexts $\Gamma$. 
    By the expression finiteness property and the characteristic of program loop sequences, it is not hard to see the result. 
\end{proof}

To close this section, we give the proof of Theorem~\ref{theo:Completeness of DLp}. 

\begin{proof}[Proof of Theorem~\ref{theo:Completeness of DLp}]
    For a labelled formula $\sigma : \phi$, 
    $\phi$ is semantically equivalent to a conjunctive normal form: $C_1\wedge ...\wedge C_n$ $(n\ge 1)$.
    Each clause $C_i$ ($1\le i\le n$) is a disjunction of literals: 
    $C_i = l_{i,1}\vee ...\vee l_{i,m_i}$, where $l_{i,j}$ ($1\le i\le n, 1\le j\le m_i$) is an atomic $\DLp$ formula or its negation. 
    By the rules for labeled proposition logical formulas in Table~\ref{table:General Rules for LDL}, 
    to prove formula $\sigma : \phi$, 
    it is enough to show that for each clause $C_i$, $\models \sigma : C_i$ implies $\pfDLp\vdash \sigma : C_i$.  
    Without loss of generality, let $C_i = \psi\vee [\la\alpha\ra]\phi$. 
    Then it is sufficient to prove $\pfDLp\vdash (\sigma : \neg \psi\Rightarrow \sigma : [\la\alpha\ra]\phi)$. 
    But it is just a special case of Lemma~\ref{lemma:completeness of DLp a more general case}. 
\end{proof}

\section{Related Work}
\label{section:Related Work} 

    \textbf{Matching Logic and Its Variations}. 
    The idea of reasoning about programs based on their operational semantics is not new. 
    Previous work such as~\cite{Rosu12,Rosu13,Stefanescu14,X.Chen19} in the last decade has addressed this issue using theories based on rewriting logic~\cite{Meseguer12}. 
    Matching logic~\cite{Rosu12} is based on patterns and pattern matching. 
    Its basic form, a reachability rule $\varphi\Rightarrow \varphi'$ (where $\Rightarrow$ has another meaning from its use in this paper), captures whether pattern $\varphi'$ is reachable from pattern $\varphi$ in a given pattern reachability system. 
    Based on matching logic, one-path and all-paths reachability logics~\cite{Rosu13,Stefanescu14} were developed by enhancing the expressive power of the reachability rule. 
    A more powerful matching $\mu$-logic~\cite{X.Chen19} was proposed by adding a least fixpoint $\mu$-binder to matching logic. 

    In these theories, 
    ``patterns'' are more general structures. So to encode the dynamic forms $[\alpha]\phi$ of $\DLp$ requires additional work and program transformations. 
    On the other hand, dynamic logics like $\DLp$ provide a more direct way to express and reason about complex before-after and temporal program properties with their modalities $[\cdot]$ and $\la\cdot\ra$. 
    In terms of expressiveness, matching logic and one-path reachability logic cannot capture the semantics of modality $[\cdot]$ in dynamic logic when the programs are non-deterministic (which means that there are more than one execution path starting from a world and a program). 
    We conjecture that matching $\mu$-logic can encode $\DLp$, as it has been claimed that it can encode traditional dynamic logics (cf.~\cite{X.Chen19}).  
    
    \ifx
    In terms of expressiveness, the semantics of modality $[\cdot]$ in dynamic logic cannot be fully captured by matching logic and one-path reachability logic when the programs are non-deterministic. Because in their theories the reachability rule $\varphi\Rightarrow{}\varphi'$ only considers the existence of one execution path.   
    Matching $\mu$-logic is more general and 
    we conjecture that it can also encode $\DLp$, as 
    it has been claimed that it can encode traditional dynamic logics (cf.~\cite{X.Chen19}).  
    Compared to these theories, 
    the specialty of $\DLp$ is that it can naturally build up nested modalities to support 
    efficiently capturing and reasoning about certain program specifications without the need of further encoding or transformation. 
    \fi
    
    \textbf{General Frameworks based on Set Theories}. 
    \cite{Moore18}~proposed a general program verification framework based on coinduction. 
    Using the terminology in this paper, 
    a program specification $\sigma : [\alpha]\phi$ can be expressed as a pair $((\alpha, \sigma), P(\phi))$ in~\cite{Moore18}, with $P(\phi)$ a set of program states capturing the semantics of formula $\phi$.    
    A method was designed to derive a program specification in a coinductive way 
    according to the operational semantics of $(\alpha, \sigma)$. 
    Following~\cite{Moore18}, \cite{X.Li21}~also proposed a general framework for program reasoning,  but via big-step operational semantics. 
    Unlike the frameworks in~\cite{Moore18} and~\cite{X.Li21} which are directly built up on mathematical set theory, 
    $\DLp$ is in logical forms, and is based on a cyclic deduction approach rather than coinduction.   
    In terms of expressiveness, the meaning of modality $\la\cdot\ra$ in $\DLp$ cannot be expressed in the framework of~\cite{Moore18}.

    \textbf{Updates}. 
    The structure ``updates'' adopted in work~\cite{Platzer07b,Beckert13,Beckert2016}
    are ``delay substitutions'' of variables and terms. 
    They in fact can be defined as a special case of the more general structure labels in $\DLp$ by choosing suitable label mappings accordingly. 

    \ifx
    in dynamic logics, used in Java dynamic logic~\cite{Beckert2016}, differential dynamic logic~\cite{Platzer07b}, dynamic trace logic~\cite{Beckert13}, etc., works as a special case of the configurations proposed in this paper, if we take the configuration as a meta variable without explicit structures. 
    As illustrated in Section~\ref{section:Example One: A While Program}, a configuration in $\DLp$ can be more than just a ``delay substitution'' (cf.~\cite{Beckert13}) of variables and terms. 
    \fi

    \textbf{Logics based on Cyclic Proof Approach}.
    The proof system of $\DLp$ relies on the cyclic proof theory which firstly arose in~\cite{Stirling91} and was later developed in 
    different logics such as~\cite{Brotherston07,Brotherston08}, and more recent work like~\cite{Gadi20,Jones22,Afshari22}. 
    \cite{Jungteerapanich09}
    proposed a complete cyclic proof system for $\mu$-calculus, which subsumes PDL~\cite{Fischer79} in its expressiveness. In~\cite{Docherty19}, the authors proposed a complete labeled cyclic proof system for PDL. 
    Both logics in \cite{Jungteerapanich09,Docherty19} are propositional and cannot be used to prove many valid formulas in particular domains, for example, the arithmetic first-order formulas in number theory as shown in our example. 
    The labeled form of $\DLp$ formula $\sigma : [\alpha]\phi$ is inspired from~\cite{Docherty19}, where a label is just a variable of worlds in a traditional Kripke structure. 
    On the other hand, the labels in $\DLp$ allow arbitrary terms from actual program configurations. 

    \textbf{Generalizations of Dynamic Logic}. 
    There are some recent work for generalizing the theories of dynamic logics~\cite{Hennicker19,Acclavio2024-uv,Teuber2025-et}. 
    \cite{Hennicker19}~proposes a general dynamic logic by allowing the program models of PDL to be any forms than regular programs. The semantics of a program is given by a set of so-called ``interaction-based'' behaviours, very similar to the program transitions here. 
    However, there, it only focuses on the building of the logic theory.
    No associated proof systems were proposed.
    In~\cite{Acclavio2024-uv}, an operational version of PDL (namely OPDL) was studied. 
    There, the proof of a dynamic formula $[\alpha]\phi$ can be reduced to the proof of formula $[a][\beta]\phi$ if $\alpha\xrightarrow{a}\beta$ is a transition by doing an action $a$. 
    Similar to~\cite{Docherty19}, a complete non-well-founded proof system was built for OPDL. 
    Although in~\cite{Acclavio2024-uv} it was claimed that OPDL can be adapted to arbitrary program models, its theory was analyzed only on the propositional level and only for regular programs. 
    \cite{Teuber2025-et}~develops heterogeneous dynamic logic (HDL), a theoretical framework in which different dynamic-logic theories can be compared and jointly used.  
    Unlike~\cite{Teuber2025-et} which makes a systematic analysis of the integration of different theories, 
    the start point of our work (as well as~\cite{zhang2025parameterizeddynamiclogic}) is to facilitate the operationally-based reasoning of different programs. 
    This leads to the introduction of labels and the development of the cyclic reasoning in $\DLp$ as critical contributions, 
    while the lifting process acts as a ``side technique'' to compensate for the core proof system. 
    The result of~\cite{Teuber2025-et} offers a thorough guide for the analysis of the completeness and other properties of the lifted theories in $\DLp$ in our future work. 



    
    
    \ifx
    However, generally speaking, they neither consider structures as general as allowed by programs and configurations in $\DLp$, nor adopt a similar approach for reasoning about programs. 
    \cite{MossakowskiEA09} proposed a complete generic dynamic logic for monads of programming languages. 
    In the dynamic logic proposed in~\cite{Hennicker19}, more general programs can be reasoned about than regular expressions of PDL, based on so-called ``interaction-based'' behaviours. 
    But there program transitions are captured by abstract actions in a form e.g. $\alpha\xrightarrow{a}\beta$, where no explicit structures of program configurations are allowed. 
    And yet no proof systems have been built for that logic.

    \textbf{Generalizations of Dynamic Logic}. 
    There has been some other work for generalizing the theories of dynamic logics, such as~\cite{MossakowskiEA09,Hennicker19}. 
    However, generally speaking, they neither consider structures as general as allowed by programs and configurations in $\DLp$, nor adopt a similar approach for reasoning about programs. 
    \cite{MossakowskiEA09} proposed a complete generic dynamic logic for monads of programming languages. 
    In the dynamic logic proposed in~\cite{Hennicker19}, more general programs can be reasoned about than regular expressions of PDL, based on so-called ``interaction-based'' behaviours. 
    But there program transitions are captured by abstract actions in a form e.g. $\alpha\xrightarrow{a}\beta$, where no explicit structures of program configurations are allowed. 
    And yet no proof systems have been built for that logic.  
    \fi

    \section{Conclusion \&\ Future Work}
\label{section:Discussions and Future Work}

In this paper, we propose a novel verification framework based on dynamic logic for reasoning about programs based on their operational semantics. 
We mainly build the theory of $\DLp$ and analyze its soundness and completeness under certain conditions. 
Through the examples and case studies, we have shown the potential usage of this formalism in different aspects of program reasoning. 

For future work we focus on two aspects. 
On the theoretical aspect, we are interested in whether we can further relax our conditions for proving the soundness and completeness of $\DLp$. This is important to know how our framework can be also adapted to more complex models, such as hybrid or probabilistic systems. 
We will further study the instantiated theory $\DLpPL$, as a promising first-ordered version of process logic ever built, and also $\DLpSP$. 
On the practical aspect, we are carrying out a full mechanization of $\DLp$ in Rocq~\cite{Bertot04}. 
Currently, we have managed to deeply embed the whole theory of $\DLp$ (cf.~\cite{Zhang25-code}). 
To explore the potential applications of $\DLp$ in practice, we are working on applying $\DLp$ for specifying and reasoning about different program models, like Esterel, Rust, etc.  

\ifx
In this paper, we propose a novel dynamic logic $\DLp$ that supports reasoning about general forms of programs and formulas based on programs' operational semantics. We mainly build the theory of $\DLp$ and propose a sound cyclic proof system of $\DLp$ that is proved to be useful and applicable. As the main theoretical result we prove the soundness of $\DLp$. 

On theoretical aspects, we are interested in obtaining a complete proof system by fixing the parameters of $\DLp$ in some algebraic domains, e.g., a multi-sorted signature for imperative programs as in~\cite{Joseph96}. 
We also want to analyze whether our framework can be adapted to a wider range of program behaviours, i.e., to relax the property Definition~\ref{def:Program Properties}. 
From an applied perspective, 
we are carrying out a full mechanization of $\DLp$ in Coq~\cite{Bertot04}. 
To see the full potential of $\DLp$, 
 we will instantiate more types of programs or system models in $\DLp$, and to specify and verify their properties using $\DLp$ formulas.  
\fi

\ifx
In this paper, we propose a dynamic-logic-like parameterized formalism $\DLp$ allowing instantiations of a rich set of programs and formulas in interested domains. $\DLp$ supports a direct symbolic-execution-based reasoning of programs, which not only helps fast-prototyping program derivation rules based on transitional behaviours, but also provides a convenient framework for verifying non-compositional models, such as for verifying transitional models~\cite{Hennicker19} or neural networks~\cite{XiyueZhang22}. 

One future work is a full mechanization of $\DLp$ using Coq~\cite{Bertot04}, which on one hand can help verifying the correctness of the theory itself, and on the other hand provides a cyclic verification framework built upon Coq to support verifying general programs. To see the full potential of $\DLp$, we are also trying to instantiate more types of programs or system models in $\DLp$. 
\fi

\ifx
$\DLp$ provides a flexible verification framework to encompass existing dynamic-logic theories, by supporting a direct symbolic-execution-based reasoning according to programs behaviours, while still preserving structure-based reasoning through lifting processes. 
In practical verification, trade-off can be made between these two types of reasoning.  
Through examples displayed in this paper, one can see the potential of $\DLp$ to be used for other types of programs, especially those do not support structural rules, such as transitional models~\cite{Hennicker19} or neural networks~\cite{XiyueZhang22}. 
On the other hand, due to its generality, $\DLp$ loses the completeness that a dynamic logic usually possesses, and must rely on additional lifted rules to support a compositional reasoning.  
\fi

\ifx

About future work, 
currently we plan a full mechanization of $\DLp$ using Coq~\cite{Bertot04}, which on one hand can help verifying the correctness of the theory itself, and on the other hand provides a verification framework built upon Coq to support verifying general programs. 
Our long-term goal is to develop an independent program verifier based on the theory of $\DLp$, which could fully make use of the advantages of symbolic-execution-based reasoning and cyclic proofs to maximize automation of verifications. 
One possible way is based on \textit{Cyclist}~\cite{Brotherston12}, a proof engine supported by an efficient cyclic-proof-search algorithm. 
To see the full potential of $\DLp$, we are also trying to use $\DLp$ to describe and verify more types of programs or system models. 
\fi



\ifx
\textbf{Acknowledgment}.
    This work is partially supported by the Youth Project of National Science Foundation of China (No. 62102329), the Project
of National Science Foundation of China (No. 62272397), and the New Faculty Start-up Foundation of NUAA (No. 90YAT24003). 
\fi

\paragraph{Acknowledgment}
    This work is partially supported by the New Faculty Start-up Foundation of NUAA (No. 90YAT24003) and the General Program of the National Natural Science Foundation of China (No. 62272397).

\bibliographystyle{ACM-Reference-Format}
\bibliography{main}

\appendix

\section{Other Propositions and Proofs}
\label{section:Other Propositions and Proofs}


\begin{lemma}
\label{lemma:from logical consequence to sequent derivation}
Given labeled formulas $\tau_1,...,\tau_n, \tau$, 
    for any label mapping $\lm\in \LM$ with $\lm\models \Gamma$, if $\lm\models \tau_1$ and ... and $\lm\models \tau_n$ implies $\lm\models \tau$, 
then rule 
    $$
\begin{aligned}
\infer[]
{
    \Gamma\Rightarrow \tau, \Delta
}
{
    \Gamma\Rightarrow \tau_1, \Delta
    &
    ...
    &
    \Gamma\Rightarrow \tau_n, \Delta
}\end{aligned}
$$ is sound for any contexts $\Gamma, \Delta$. 
\end{lemma}

\begin{proof}[Proof of Lemma~\ref{lemma:from logical consequence to sequent derivation}]
    Assume $\Gamma\Rightarrow \tau_1, \Delta$,...,$\Gamma\Rightarrow \tau_n, \Delta$ are valid, 
    for any $\lm\in \LM$ with $\lm\models \Gamma$,  let $\Kr, \lm\not\models \tau'$ for all $\tau'\in \Delta$, 
    we need to prove $\lm\models \tau$. 
    From the assumption we have $\lm\models \tau_1$, ..., $\lm\models \tau_n$. 
    Then $\lm\models \tau$ is an immediate result. 
\end{proof}

\begin{lemma}
\label{lemma:from logical consequence to sequent derivation 2}
Given labeled formulas $\tau_1,...,\tau_n, \tau$, 
    for any label mapping $\lm\in \LM$ with $\Kr,\lm\models \Gamma$, if $\lm\models \tau$ implies either $\lm\models \tau_1$ or ... or $\lm\models \tau_n$, 
then rule 
    $$
\begin{aligned}
\infer[]
{
    \Gamma, \tau\Rightarrow \Delta
}
{
    \Gamma, \tau_1\Rightarrow \Delta
    &
    ...
    &
    \Gamma, \tau_n\Rightarrow \Delta
}\end{aligned}
$$ is sound for any contexts $\Gamma, \Delta$. 
\end{lemma}

\begin{proof}[Proof of Lemma~\ref{lemma:from logical consequence to sequent derivation 2}]
    Assume $\Gamma, \tau_1\Rightarrow  \Delta$,...,$\Gamma, \tau_n\Rightarrow \Delta$ are valid.
    For any $\lm\in \LM$ with $\lm\models \Gamma$, if $\lm\models \tau$, 
    by the assumption, $\lm\models \tau_i$ for some $1\le i\le n$. 
    By the validity of $\Gamma, \tau_i\Rightarrow \Delta$, $\lm\models \Delta$. 
    By the arbitrariness of $\lm$ we know that $\Gamma, \tau\Rightarrow \Delta$ is valid. 
\end{proof}

\ifx
\begin{lemma}
\label{lemma:from logical consequence to sequent derivation}
For any labeled formulas $\phi_1,...,\phi_n, \phi$, 
    if formula $\phi_1\wedge ...\wedge\phi_n\to \phi$ is valid, 
    then sequent 
    $$
\begin{aligned}
\infer[]
{
    \Gamma\Rightarrow \phi, \Delta
}
{
    \Gamma\Rightarrow \phi_1, \Delta
    &
    ...
    &
    \Gamma\Rightarrow \phi_n, \Delta
}\end{aligned}$$ is sound for any contexts $\Gamma, \Delta$. 
\end{lemma}

\begin{proof}
    Assume $\mfr{P}(\Gamma\Rightarrow \phi_1, \Delta)$,...,$\mfr{P}(\Gamma\Rightarrow \phi_n, \Delta)$ are valid, 
    for any evaluation $\rho\in \Eval$, let $\rho\models \Gamma$ and $\rho\not\models \psi$ for all $\psi\in \Delta$, 
    we need to prove $\rho\models \phi$. 
    From the assumption we have $\rho\models \phi_1$, ..., $\rho\models \phi_n$. 
    $\rho\models \phi$ is an immediate result since
    $\phi_1\wedge ...\wedge \phi_n\to \phi$. 
\end{proof}
\fi

{\parindent 0pt
\textbf{Content of Theorem~\ref{theo:soundness of rules for LDL}: }
Each rule from $\pfLDLp$ in Table~\ref{table:General Rules for LDL} is sound. 
}

Below we only prove the soundness of the rules $([\alpha]R)$, $([\alpha]L)$ and $(\Sub)$. 
Other rules can be proved similarly based on the semantics of $\DLp$ formulas (Definition~\ref{def:Semantics of dynamic logical formulas}).  

\begin{proof}[Proof of Theorem~\ref{theo:soundness of rules for LDL}]   

    For rule $([\alpha]R)$, by Lemma~\ref{lemma:from logical consequence to sequent derivation}, it is sufficient to prove that 
    for any $\lm\in \LM$ with $\lm\models \Gamma$, if $\lm\models \sigma' : [\alpha']\phi$ for all $(\alpha', \sigma')\in \Phi$, 
    then $\lm\models \sigma : [\alpha]\phi$. 
    For any relation $\lm(\sigma)\xrightarrow{\alpha/\alpha_0}w_0$ with some $\alpha_0\in \Prog$ and $w_0\in \Wd$, by Item~\ref{item:coincidence} of Definition~\ref{def:Matching Proof System}, there is a label $\sigma_0\in \Conf$ such that $\lm(\sigma_0) = w_0$ and $\models (\Gamma\Rightarrow{} (\alpha, \sigma)\trans(\alpha_0, \sigma_0))$. 
    So $\pfDLp\vdash (\Gamma\Rightarrow (\alpha, \sigma)\trans(\alpha_0, \sigma_0))$ (by Item~\ref{item:assump 3} of Definition~\ref{def:Matching Proof System}). 
    Hence $(\alpha_0, \sigma_0)\in \Phi$, 
    and by assumption, $\Kr,\lm\models \sigma_0 : [\alpha_0]\phi$. 
    By the arbitrariness of $\alpha_0$ and $\sigma_0$, we have $\lm\models \sigma : [\alpha]\phi$ according to the semantics of formula $[\alpha]\phi$ (Definition~\ref{def:Semantics of dynamic logical formulas}). 

    \ifx new proof for the case ([alpha]R)
    From $\vdash (\Gamma\Rightarrow (\alpha, \sigma)\trans(\alpha', \sigma'))$, by the assumption made about the soundness of $\pfDLp$ w.r.t. $\PT$ (Item~\ref{item:assump 2} of Definition~\ref{def:Matching Proof System}), 
    $\models (\Gamma\Rightarrow (\alpha, \sigma)\trans(\alpha', \sigma'))$. 
    Since $\lm\models \Gamma$, $\lm\models (\alpha, \sigma)\trans(\alpha', \sigma')$, which means $\lm(\sigma)\xrightarrow{\alpha/\alpha'}\lm(\sigma')$ is a relation on $\Kr$. 
    From $\Kr,\lm\models \sigma' : [\alpha']\phi$ and the arbitrary of $(\alpha', \sigma')$, by the semantics of dynamic formulas (Definition~\ref{def:Semantics of dynamic logical formulas}), we have $\lm\models \sigma : [\alpha]\phi$. 
    \fi
    
    \ifx old proof for the case ([alpha]R)
    For any relation $\lm(\sigma)\xrightarrow{\alpha/\alpha_0}w_0$ with some $\alpha_0\in \Prog$ and $w_0\in \Wd$, by the completeness w.r.t. $\PT$ of system $\pfDLp$ (Item~\ref{item:assump 3} of Definition~\ref{def:Matching Proof System}), there is a label $\sigma_0\in \Conf$ such that $\lm(\sigma_0) = w_0$ and $\pfDLp\vdash (\cdot\Rightarrow{} (\alpha, \sigma)\trans(\alpha_0, \sigma_0))$. 
    So $(\alpha_0, \sigma_0)\in \Phi$. 
    Hence by assumption, $\Kr,\lm\models \sigma_0 : [\alpha_0]\phi$. 
    By the arbitrariness of $\alpha_0$ and $\sigma_0$, we have $\Kr, \lm\models \sigma : [\alpha]\phi$ according to the semantics of formula $[\alpha]\phi$ (Definition~\ref{def:Semantics of dynamic logical formulas}). 
    \fi


    For rule $([\alpha]L)$,  
    by Lemma~\ref{lemma:from logical consequence to sequent derivation 2}, it is sufficient to prove that 
    for any $\lm\in \LM$ with $\lm\models \Gamma$, if $\lm\models \sigma : [\alpha]\phi$, 
    then $\lm\models \sigma' : [\alpha']\phi$. 
    For any execution path $\lm(\sigma')\xrightarrow{\alpha'/\cdot}...\xrightarrow{\cdot/\ter}w$ of $\alpha'$ for some $w\in \Wd$, by the soundness of the side deduction $\pfDLp\vdash (\Gamma\Rightarrow{} (\alpha, \sigma)\trans(\alpha', \sigma'), \Delta)$, $\lm(\sigma)\xrightarrow{\alpha/\alpha'}\lm(\sigma')$ is a relation on $\Kr$, so 
    $\lm(\sigma)\xrightarrow{\alpha/\alpha'}\lm(\sigma')\xrightarrow{\alpha'/\cdot}...\xrightarrow{\cdot/\ter}w$ is an execution path of $\alpha$. 
    By the semantics of the dynamic formulas (Definition~\ref{def:Semantics of dynamic logical formulas}), we obtain the result.


    For rule $(\Sub)$, 
    we need to prove that the validity of $\Gamma\Rightarrow\Delta$ implies the validity of $\Sub(\Gamma)\Rightarrow \Sub(\Delta)$. 
    For any $\lm\in \LM$ satisfying $\lm\models \Sub(\Gamma)$, 
    by Definition~\ref{def:Substitution of Labels}, there exists a $\lm'$ such that for any formula $\sigma : \phi\in \Gamma\cup \Delta$, $\lm'(\sigma) = \lm(\Sub(\sigma))$. 
    So $\Kr, \lm'\models \Gamma$. 
    Since $\Gamma\Rightarrow \Delta$ is valid, 
    $\Kr, \lm'\models \sigma_0 : \phi_0$ for some $\sigma_0 : \phi_0\in \Delta$. 
    By that $\lm'(\sigma_0) = \lm(\Sub(\sigma_0))$, $\lm\models \Sub(\sigma_0) : \phi_0$ with $\Sub(\sigma_0) : \phi_0\in \Sub(\Delta)$. 
    By the arbitrariness of $\lm$, $\Sub(\Gamma)\Rightarrow \Sub(\Delta)$ is valid. 
    
    \ifx
    it is sufficient to prove that for any label mapping $\lm\in \LM$, there exists a $\lm'\in \LM$ such that for each formula $\tau\in \Gamma\cup \Delta$, 
    we have $\lm(\Sub(\tau)) = \lm'(\tau)$. 
    But this just matches the definition of $\Sub$ in Definition~\ref{def:Abstract Substitution}. 
    \fi

    \ifx
    For rule $(\sigma \textit{Sub})$, 
    for any evaluation $\rho\in \Eval$, let $\rho'\dddef \rho[x\mapsto \rho(t)]$. 
    By the definition of substitutions in Section~\ref{section:Terms and Substitutions}, we can have that 
    for any term $u\in \TA$, $\rho(u[t/x]) \equiv \rho|_{\FV(u)\setminus\FV_x(u)}(u)[\rho(t)/x] \equiv \rho'^*|_{\FV(u)\setminus\FV_x(u)}(u)[\rho(t)/x] \equiv \rho'^*(u)$. 
    Therefore, it is easy to prove that $\mfr{P}(\Gamma\Rightarrow \Delta)$ implies $\mfr{P}(\Gamma[t/x]\Rightarrow \Delta[t/x])$. 
    \fi
\end{proof}

{\parindent 0pt
\textbf{Content of Proposition~\ref{prop:relation mult is partially ordered}:} 
$\mult$ is a partial-order relation. 

Before proving Proposition~\ref{prop:relation mult is partially ordered}, we firstly review the notion of \emph{partial-order relation}. 

A relation $\preceq$ on a set $S$ is \emph{partially ordered}, if 
it satisfies the following properties: 
(1) Reflexivity. $t\preceq t$ for each $t\in S$. 
(2) Anti-symmetry. For any $t_1, t_2\in S$, if $t_1\preceq t_2$ and $t_2\preceq t_1$, then $t_1$ and $t_2$ are the same element in $S$, we denote by $t_1 = t_2$.
(3) Transitivity. For any $t_1, t_2, t_3\in S$, if $t_1\preceq t_2$ and $t_2\preceq t_3$, then $t_1\preceq t_3$. 
$t_1 \prec t_2$ is defined as $t_1\preceq t_2$ and $t_1\neq t_2$. 

Recall that we use $\suf$ to represent the suffix relation between execution paths. $\suf$ is obviously a partial-order relation. 

\begin{proof}[Proof of Proposition~\ref{prop:relation mult is partially ordered}]
    The reflexivity is trivial. 
    The transitivity can be proved by the definition of `replacements' as described in Definition~\ref{def:Relation pmult} and the transitivity of relation $\psuf$.  
    Below we only prove the anti-symmetry. 

    For any finite sets $D_1, D_2$ of finite paths, if $D_1\mult D_2$ but $D_1\neq D_2$, let $f_{D_1,D_2} : D_1\to D_2$ be the function 
    defined such that for any $tr\in D_1$, either (1) $f_{D_1,D_2}(tr) \sufeq tr$; or (2) $tr$ is one of the replacements of a replaced element $f_{D_1,D_2}(tr)$ in $D_2$ with $tr\psuf f_{D_1,D_2}(tr)$. 

    For the anti-symmetry, 
    suppose $\cnt_1\mult\cnt_2$ and $\cnt_2\mult\cnt_1$ but $\cnt_1\neq \cnt_2$. 
    Let $tr\in \cnt_1$ but $tr\notin \cnt_2$. 
    Then from $\cnt_1\mult \cnt_2$, we have $tr\psuf f_{\cnt_1, \cnt_2}(tr)$. 
    If $f_{\cnt_1, \cnt_2}(tr)\in \cnt_1$, then we must have $f_{\cnt_1, \cnt_2}(tr)\pmult f_{\cnt_1, \cnt_2}(f_{\cnt_1, \cnt_2}(tr))$ because $f_{\cnt_1, \cnt_2}(tr)$ is already a replaced element in $\cnt_2$. 
    If $f_{\cnt_1, \cnt_2}(tr)\notin\cnt_1$, then by $\cnt_2\mult\cnt_1$, we have $f_{\cnt_1, \cnt_2}(tr)\psuf f_{\cnt_2, \cnt_1}(f_{\cnt_1, \cnt_2}(tr))$. 
    Continuing this process by considering $f_{\cnt_1, \cnt_2}(f_{\cnt_1, \cnt_2}(tr))$ or $f_{\cnt_2, \cnt_1}(f_{\cnt_1, \cnt_2}(tr))$ and further elements, we in fact can construct an infinite descent sequence
    like $tr\psuf f_{\cnt_1, \cnt_2}(tr)\psuf f_{\cnt_1, \cnt_2}(f_{\cnt_1, \cnt_2}(tr))\psuf ...$ w.r.t. relation $\suf$, which violates its well-foundedness. 
    So the only possibility is $\cnt_1 = \cnt_2$. 
    
    \ifx
    Since $f_{\cnt_1, \cnt_2}(tr)$ is the replaced element, so $f_{\cnt_1, \cnt_2}(tr)\notin \cnt_1$. 
    By $\cnt_2\mult\cnt_1$, we have $f_{\cnt_1, \cnt_2}(tr)\psuf f_{\cnt_2, \cnt_1}(f_{\cnt_1, \cnt_2}(tr))$. 
    Continuing this process, we can construct an infinite descent sequence
    $tr\psuf f_{\cnt_1, \cnt_2}(tr)\psuf f_{\cnt_2, \cnt_1}(f_{\cnt_1, \cnt_2}(tr))\psuf ...$ w.r.t. relation $\suf$, which violates its well-foundedness. 
    So the only possibility is $\cnt_1 = \cnt_2$. 
    \fi
    
\end{proof}

{\parindent 0pt
\textbf{Content of Lemma~\ref{lemma:infinite descent sequence}:}  
    In a cyclic proof (where there is at least one derivation path), 
    let $(\sigma:\phi, \sigma':\phi')$ be a step of a derivation trace over a derivation $(\nu, \nu')$ of an invalid derivation path, where $\phi, \phi'\in \DLF$.  
    For any set $\EX(\lm(\sigma), \phi)$ of $\sigma : \phi$ w.r.t. a counter-example mapping $\lm$ of $\nu$, 
    there exists a counter-example mapping $\lm'$ of $\nu'$ and a set $\EX(\lm'(\sigma'), \phi')$ of $\sigma':\phi'$ such that $\EX(\lm'(\sigma'), \phi')\mult \EX(\lm(\sigma), \phi)$. Moreover, if $(\sigma:\phi, \sigma':\phi')$ is a progressive step, then 
    $\EX(\lm'(\sigma'), \phi')\pmult \EX(\lm(\sigma), \phi)$. 
}

\begin{proof}[Proof of Lemma~\ref{lemma:infinite descent sequence}]
    Consider the rule application from node $\nu$, we only consider the cases when it is an instance of rule $([\alpha]R)$, rule $([\alpha]L)$, rule $(\Sub)$, 
    and rule $(\wedge R)$, and when the first element of the CP pair we consider is their target formula.  

    \textbf{Case for rule $([\alpha]R)$: }
    If from node $\nu$ rule $([\alpha]R)$ is applied with $\tau\dddef \sigma : [\alpha]\phi$ the target formula, 
    let $\tau' = (\sigma' : [\alpha']\phi)$ for some $\alpha'$ and $\sigma'$, so $\nu=(\Gamma\Rightarrow \tau, \Delta)$ and $\nu' = (\Gamma\Rightarrow \tau', \Delta)$. 
    (In this case, $(\nu, \nu')$ is already a progressive step. )
    Since $\lm$ is a counter-example of $\nu$, $\lm\not\models \tau$, so $\mex(\lm(\sigma), \alpha)\neq \emptyset$. Thus $\EX(\lm(\sigma), [\alpha]\phi)\neq \emptyset$. 
    By the soundness of rule $([\alpha]R)$ and the assumption that $\pfDLp\vdash \Gamma\Rightarrow (\alpha, \sigma)\xrightarrow{}(\alpha', \sigma'),\Delta$, for each path $tr = \lm(\sigma')s_1...s_n\in \EX(\lm(\sigma'), [\alpha']\phi)$ ($n\ge 0$), 
    path $\lm(\sigma)tr\in \EX(\lm(\sigma), [\alpha]\phi)$ has $tr$ as its proper suffix. 
    By Definition~\ref{def:Program Properties}, 
    $\EX(\lm(\sigma), [\alpha]\phi)$ and $\EX(\lm(\sigma'), [\alpha']\phi)$
    are also finite. 
    Therefore $\EX(\lm(\sigma'), [\alpha']\phi)\pmult \EX(\lm(\sigma), [\alpha]\phi)$. 
    
     \textbf{Case for rule $([\alpha]L)$: }
     If from node $\nu$ rule $([\alpha]L)$ is applied with $\tau\dddef \sigma : [\alpha]\phi$ the target formula, 
    let $\tau' = (\sigma' : [\alpha']\phi)$ for some $\alpha'$ and $\sigma'$, so $\nu=(\Gamma\Rightarrow \tau, \Delta)$ and $\nu' = (\Gamma\Rightarrow \tau', \Delta)$. 
    By the soundness of rule $([\alpha]L)$ and the assumption that $\pfDLp \vdash \Gamma\Rightarrow (\alpha, \sigma)\xrightarrow{}(\alpha', \sigma'),\Delta$, 
    for each path $tr = \lm(\sigma')s_1...s_n\in \EX(\lm(\sigma'), [\alpha']\phi)$ ($n\ge 0$), 
    path $\lm(\sigma)tr\in \EX(\lm(\sigma), [\alpha]\phi)$ has $tr$ as its proper suffix. 
     By Definition~\ref{def:Program Properties}, 
    $\EX(\lm(\sigma), [\alpha]\phi)$ and $\EX(\lm(\sigma'), [\alpha']\phi)$
    are also finite. 
    Therefore $\EX(\lm(\sigma'), [\alpha']\phi)\mult \EX(\lm(\sigma), [\alpha]\phi)$, where the equivalence relation $=$ holds only when $\EX(\lm(\sigma), [\alpha]\phi) = \emptyset$.
    When $(\tau, \tau')$ is progressive, which means that we also have the derivation $\pfDLp\vdash (\Gamma\Rightarrow \alpha \termi\sigma, \Delta)$, then $\mex(\lm(\sigma), \alpha)\neq \emptyset$. 
    This means that $\EX(\lm(\sigma), [\alpha]\phi)\neq \emptyset$. 
    Therefore $\EX(\lm(\sigma'), [\alpha']\phi)\pmult \EX(\lm(\sigma), [\alpha]\phi)$. 

    \ifx AN OLD PROOF
     Now consider two cases: 
    (1) If $\mex(\lm(\sigma), \alpha) = \emptyset$, then by the above $\mex(\lm(\sigma'), \alpha')$ is also empty. So by the definition of $\EX$, $\EX(\lm(\sigma), [\alpha]\phi) = \EX(\lm(\sigma'), [\alpha']\phi)$; 
    (2) If $\mex(\lm(\sigma), \alpha)\neq \emptyset$, especially, when $(\tau, \tau')$ is a progressive step with the support of the derivation $\pfDLp\vdash (\Gamma\Rightarrow \alpha' \termi\sigma', \Delta))$, 
    then $\EX(\lm(\sigma), [\alpha])\neq\emptyset$, 
    moreover, 
    we have $\EX(\lm(\sigma'), [\alpha']\phi)\pmult \EX(\lm(\sigma), [\alpha]\phi)$. 
    \fi
    
    \textbf{Case for rule $(\Sub)$: }
    If from node $\nu$ a substitution rule $(\Sub)$ is applied, 
    let $\tau = \Sub(\sigma) : \phi$ be the target formula of $\nu$, 
    then $\tau' = \sigma : \phi$. 
    By Definition~\ref{def:Substitution of Labels}, 
    for the label mapping $\lm$, 
    there exists a $\lm'$ such that 
    $\lm'(\sigma) = \lm(\Sub(\sigma))$. 
    So $\EX(\lm(\Sub(\sigma)), \phi) = \EX(\lm'(\sigma), \phi)$. 

    \ifx
    If from node $\nu$ a substitution rule $(\Sub)$ is applied, 
    let $\tau = \Sub(\sigma : \phi) = \Sub_p(\sigma) : \Sub(\phi)$ (for some $\Sub_p : \Prog\to \Prog$) be the target formula of $\nu$, 
    then $\tau' = \sigma : \phi$. 
    By Definition~\ref{def:Substitution of Labels}, 
    for the label mapping $\lm$, 
    there exists a $\lm'$ such that 
    $\lm\models \Sub(\sigma : \phi)$ iff $\lm'\models \sigma : \phi$. 
    So the invalidity of $\Sub(\sigma : \phi)$ implies the invalidity of $\sigma : \phi$ (, where $\lm'$ is a counter example mapping of $\sigma : \phi$). 
    Moreover, 
    \fi

    \textbf{Case for rule $(\wedge R)$: }
    If from node $\nu$ rule $(\wedge R)$ is applied, 
    let $\tau = \sigma : \phi\wedge \psi$ be the target formula of $\nu$, 
    and $\tau' = \sigma : \phi$. 
    By the definition of $\EX$ (Definition~\ref{def:counter-example set}), 
    we have $\EX(\lm(\sigma), \phi)\subseteq \EX(\lm(\sigma), \sigma:\phi\wedge \psi)$, 
    so $\EX(\lm(\sigma), \phi)\mult \EX(\lm(\sigma), \sigma:\phi\wedge \psi)$. 
\end{proof}

From the case of rule $([\alpha]L)$ in the above proof, 
we see that derivation $\pfDLp\vdash (\Gamma\Rightarrow \alpha \termi\sigma, \Delta))$ 
imposes $\EX(m(\sigma), [\alpha]\phi)\neq \emptyset$, which is the key to prove the strict relation $\pmult$ between $\EX(\lm(\sigma'), [\alpha']\phi)$ and $\EX(\lm(\sigma), [\alpha]\phi)$. 

{\parindent 0pt
\textbf{Content of Proposition~\ref{prop:lifting process 2}:} 
Given a sound rule of the form
$$
\begin{aligned}
    \infer[]
    {\Gamma\Rightarrow \Delta}
    {\Gamma_1\Rightarrow \Delta_1
    &...&
    \Gamma_n\Rightarrow \Delta_n}
\end{aligned}, \mbox{ $n\ge 1$}, 
$$
in which all formulas are unlabeled, then the rule 
$$
\begin{aligned}
    \infer[]
    {\sigma : \Gamma\Rightarrow \sigma : \Delta}
    {\sigma : \Gamma_1\Rightarrow \sigma : \Delta_1
    &...&
    \sigma : \Gamma_n\Rightarrow \sigma : \Delta_n}
\end{aligned}
$$
is sound for any label $\sigma\in \free(\Conf, \Gamma\cup \Delta\cup \Gamma_1\cup \Delta_1\cup...\cup\Gamma_n\cup\Delta_n)$. 
}

\begin{proof}[Proof of Proposition~\ref{prop:lifting process 2}]
    Let $A = \Gamma\cup \Delta\cup \Gamma_1\cup \Delta_1\cup...\cup\Gamma_n\cup\Delta_n$. 
    Assume sequents $\sigma : \Gamma_i\Rightarrow \sigma : \Delta_i$ $(1\le i\le n)$ are valid, we need to prove that sequent $\sigma : \Gamma\Rightarrow \sigma : \Delta$ is valid. 
    First, notice that each sequent $\Gamma_i\Rightarrow \Delta_i$ ($1\le i\le n$) is valid. 
    Because since $\sigma \in \free(\Conf, A)$, for any $s\in \Wd$ such that $s\models \Gamma_i$, there is a label mapping denoted by $\lm_s\in \LM$ with $s=_A \lm_s(\sigma)$ such that $\lm_s(\sigma)\models \Gamma_i$ (Definition~\ref{def:Same Effects}), so $\lm_s\models \sigma : \Gamma_i$ (By Definition~\ref{definition:Semantics of labeled Dlp Formulas}). By the validity of the sequent $\sigma : \Gamma_i\Rightarrow \sigma : \Delta_i$, $\lm_s\models \sigma : \phi$ for some $\phi\in \Delta_i$, which means $\lm_s(\sigma)\models \phi$. Hence $s\models \phi$. 
    From the validity of $\Gamma_i\Rightarrow\Delta_i$, we get that $\Gamma\Rightarrow \Delta$ is valid. 
    For any $\lm\in \LM$, if $\lm\models \sigma : \Gamma$, then $\lm(\sigma)\models \Gamma$ (Definition~\ref{definition:Semantics of labeled Dlp Formulas}). 
    By the validity of $\Gamma\Rightarrow\Delta$, we have $\lm(\sigma)\models \phi$ for some $\phi\in \Delta$. But this just means $\lm\models \sigma : \phi$. 
    Therefore, we have concluded that $\sigma : \Gamma\Rightarrow \sigma : \Delta$ is valid. 
    
    \ifx
    Let $A = \Gamma\cup \Delta\cup \Gamma_1\cup \Delta_1\cup...\cup\Gamma_n\cup\Delta_n$. 
    Assume formulas $\mfr{P}(\sigma : \Gamma_1\Rightarrow \sigma : \Delta_1)$,...,$\mfr{P}(\sigma : \Gamma_n\Rightarrow \sigma : \Delta_n)$ are valid, 
    we need to prove the validity of formula $\mfr{P}(\sigma : \Gamma\Rightarrow \sigma : \Delta)$. 
    First, notice that formula $\mfr{P}(\Gamma_i\Rightarrow \Delta_i)$ ($i\ge 1$) is valid, 
    because for any evaluation $\rho\in \Eval$, by Definition~\ref{def:free configurations}-\ref{item:free configurations cond 2} there exists a $\rho'\in \Eval$ satisfying that $(\rho', \sigma)\se_A \rho$. 
    So for any formula $\phi\in A$, $\rho\models \phi$ iff $\rho'\models \sigma : \phi$. 
    Therefore, formula $\mfr{P}(\Gamma\Rightarrow \Delta)$ is valid. 
    On the other hand, from that $\mfr{P}(\Gamma\Rightarrow \Delta)$ is valid, 
    we can get that $\mfr{P}(\sigma : \Gamma\Rightarrow \sigma : \Delta)$ is valid. 
    This is because for any evaluation $\rho\in \Eval$, by Definition~\ref{def:free configurations}-\ref{item:free configurations cond 1} there is a $\rho'\in \Eval$ such that $(\rho,\sigma)\se_A\rho'$, which means for any formula $\phi\in A$, $\rho\models\sigma : \phi$ iff $\rho'\models \phi$. 
    \fi
\end{proof}


{\parindent 0pt
\textbf{Content of Lemma~\ref{lemma:completeness of DLp a more general case}:} 
    Under the same conditions as in Theorem~\ref{theo:Completeness of DLp}, for any valid sequent of the form: $\Gamma\Rightarrow \sigma : [\la\alpha \ra]\phi$, $\pfDLp \vdash (\Gamma\Rightarrow \sigma : [\la\alpha \ra]\phi)$. 

\begin{proof}[Proof of Lemma~\ref{lemma:completeness of DLp a more general case}]
    Let $\nu\dddef (\Gamma\Rightarrow \sigma : [\la\alpha\ra]\phi)$. 
    We proceed by simultaneous induction on the number $M_\nu$ of the modalities $[\cdot]$ or $\la\cdot \ra$ in $\nu$ and the maximum number $N_{\alpha}$ of the different programs along the program sequences starting from $\alpha$ for all labels and contexts (see Definition~\ref{def:Expression Finiteness Property}). 

    \textbf{Base case}. $M_\nu = N_\alpha = 1$, so in the sequent $\nu$ there is only one modality, which is $[\la\alpha\ra]$, and program $\alpha$ is either $\ter$ or a loop program $\neq \ter$ that can only perform transitions of the form: $(\alpha, \sigma)\trans (\alpha, \sigma')$ for some $\alpha'$ under some context.
    The case for $\alpha = \ter$ is trivial, as by applying rule $([\ter])$ or rule $(\la\ter\ra)$ and rule $(\textit{Ter})$, we can directly obtain the result. 
    For the case when $\alpha\neq \ter$ is a loop program, the proof is just a special case to the proof for the step case as follows. 

    \textbf{Step case}. 
    This case is divided into two parts. Part I describes the process of constructing the derivation of sequent $\nu$ (named ``Proc'' below); Part II further proves that this derivation is cyclic.   
    
    \textit{Part I}: 
    Without loss of generality, suppose $\alpha$ is a loop program. 
    \ifx
    Starting from $(\alpha, \sigma)$ under the context $\Gamma$, we can categorize the core $\alpha$ sequences into the following two types: 
    $A = \{\Gamma : (\alpha, \sigma)...(\alpha,\sigma_i)\}_{i\ge 1}$, the set of all core $\alpha$-loop sequences
    and 
    $B = \{\Lambda_j : (\alpha, \sigma)(\beta, \varsigma_j)...\ |\ j\ge 1, \models (\Lambda_j\Rightarrow \Gamma)\}_{j\ge 1}$. 
    $A$ is the set of all $\alpha$-loop sequences starting from $(\alpha, \sigma)$ under context $\Gamma$, while $B$ is the set of the other $\alpha$ sequences starting from $(\alpha, \sigma)$ under $\Gamma$. 
    \fi
    For the label $\sigma\in \Conf$ and the context $\Gamma$, 
    by Definition~\ref{def:Well-behaved Loop Programs}, there exist a label $\sigma'$, a context $\Gamma'$ and a substitution $\eta$ satisfying that
    \begin{enumerate}[(a)]
    \item $\sigma = \eta(\sigma')$ and $\Gamma = \eta(\Gamma')$;
    \item $\models (\Gamma\Rightarrow \sigma\termi \alpha)$ implies $\models (\Gamma'\Rightarrow \sigma'\termi\alpha)$;
    \item for each $\alpha$-loop sequence: $\Gamma' : (\alpha, \sigma', \Gamma_1)...(\alpha, \sigma'', \Gamma_n)$ ($n\ge 1$), 
    there exist a context $\Gamma''$ and a substitution $\xi$ such that $\Gamma'' = \xi(\Gamma')$, $\sigma'' = \xi(\sigma')$ and $\models (\Gamma_n\Rightarrow \Gamma'')$.  
    \end{enumerate}
    From $\Gamma\Rightarrow \sigma : [\la\alpha\ra]\phi$, 
    we can have the following derivation named ``Proc'': 
    $$
    \begin{aligned}
        \infer[^{(\Sub)}]
        {1:\Gamma\Rightarrow \sigma : [\la\alpha\ra]\phi}
        {
            \infer[]
            {2:\Gamma'\Rightarrow \sigma' : [\la\alpha\ra\phi]}
            {
                \infer*[]
                {...}
                {
                    \infer[]
                    {...}
                    {
                        ...
                        &
                        \infer[^{([\la\alpha\ra]R)}]
                        {3:\Gamma_1\Rightarrow \sigma' : [\la\alpha\ra]\phi}
                        {
                            \infer*[]
                            {...}
                            {
                                \infer[^{([\la\alpha\ra]R)}]
                                {...}
                                {
                                    \infer[^{(\textit{Cut})}]
                                    {4:\Gamma_n\Rightarrow \sigma'' : [\la\alpha\ra]\phi}
                                    {
                                        \infer[]
                                        {8:\Gamma_n\Rightarrow \Gamma''}
                                        {(\mbox{Ind. Hypo.})}
                                        &
                                        \infer[^{(\textit{WkR})}]
                                        {\Gamma_n, \Gamma''\Rightarrow\sigma'' : [\la\alpha\ra]\phi}
                                        {
                                            \infer[^{(\Sub)}]
                                            {6:\Gamma''\Rightarrow\sigma'' : [\la\alpha\ra]\phi}
                                            {
                                                7:\Gamma'\Rightarrow\sigma' : [\la\alpha\ra]\phi
                                            }
                                        }
                                    }
                                }
                            }
                        }
                        &
                        ...
                        &
                        \infer*[]
                        {...}
                        {
                            \infer[]
                            {...}
                            {
                                \infer[^{([\la\alpha\ra]R)}]
                                {9:\Lambda\Rightarrow\varsigma : [\la\beta\ra]\phi}
                                {
                                    \infer[]
                                    {10:\Lambda\Rightarrow\varsigma' : [\la\beta'\ra]\phi}
                                    {(\mbox{Ind. Hypo.})}
                                }
                            }
                        }
                        &
                        ...
                    }
                }
            }
        }
    \end{aligned}
    $$
    The derivation from node 1 to 2 is according to (a). 
    From node 2, the context $\Gamma'$ is strengthened for the derivation of each $\alpha$ sequence starting from $(\alpha, \sigma')$. 
    This process can be realized by applying rule $(\textit{Cut})$ and the other rules for labeled propositional logical formulas as shown in Table~\ref{table:General Rules for LDL}.  
    Each $\alpha$-loop sequence: $\Gamma': (\alpha, \sigma', \Gamma_1)...(\alpha, \sigma'', \Gamma_n)$ ($n\ge 1$) corresponds to a derivation, named ``Sub-proc 1'', like the one from node 3 as shown in Proc. 
    From node 3 to 4 includes a series of derivation steps that symbolically executes the $\alpha$-loop sequence by applying the rule $([\alpha]R)$ or $(\la\alpha\ra R)$ (denoted by $([\la\alpha\ra]R)$) and also the other rules for strengthening the contexts $\Gamma_2,...,\Gamma_n$. 
    The (c) above provides the evidence for both the derivation from node 6 to 7, and the validation of node 8. 
    The derivation from node 2 to 10 is a general case of an $\alpha$ sequence: $\Gamma' : (\alpha, \sigma', \cdot)...(\beta, \varsigma, \Lambda)(\beta', \varsigma', \Lambda)$, where from $(\beta', \varsigma')$ under $\Lambda$ program $\alpha$ can never be reached. 
    We name the derivation like the one from node 10 here as ``Sub-proc 2''. 
    
    \textit{Part II}: 
    Now we show that the derivation Proc is actually a cyclic proof. 
    We firstly show that the whole proof Proc is a preproof (i.e. a finite tree structure with buds), which is based on the following 3 proof statements. 
    \begin{enumerate}[(1)]
        \item The derivation part as shown above in Proc is finite. 
        On one hand, by the finiteness of set $\pfDLp$, each derivation step must have finite premises;
        On the other hand, 
    by Lemma~\ref{lemma:Finite Loop Sequences}, from node 1 there is a finite number of $\alpha$ sequences (by selecting a set of contexts for each core $\alpha$ sequence). 
    This means that the number of the branches Sub-proc 1 is finite. 
    \item Each Sub-proc 1 is a preproof branch. On one hand, By that $\xi(\Gamma') = \Gamma''$ (see (c) above), the number of modalities in $\Gamma''$ is the same as that in $\Gamma'$. On the other hand, by Item~\ref{item:Simple Conditions} of Definition~\ref{def:Matching Proof System}, when we strengthen the context $\Gamma'$ during the derivation from node 3 to 4, 
    we can make sure that compared to $\Gamma'$ there is no more dynamic formulas added in $\Gamma_1,...,\Gamma_n$. Therefore, the number of modalities in the sequent $\Gamma_n\Rightarrow \Gamma''$ equals to that in sequent $\nu$, and is thus strictly less than $M_\nu$. 
    So, by induction hypothesis, $\pfDLp\vdash (\Gamma_n\Rightarrow \Gamma'')$. 
    \item  Each Sub-proc 2 is a cyclic proof branch.  
    For any derivation step like the one from node 9 to 10, since from $(\beta', \varsigma')$ under $\Lambda$ the program $\alpha$ can never be reached, clearly we have $N_{\beta'} < N_\alpha$, because except $\alpha$ itself, any program that $\beta'$ can reach can be reached by $\alpha$ (through the $\alpha$ sequence: $\Gamma' : (\alpha, \sigma', \cdot)...(\beta,\varsigma,\Lambda)(\beta',\varsigma',\Lambda)$).  
    So by induction hypothesis, $\pfDLp \vdash (\Lambda\Rightarrow \varsigma' : [\la\beta'\ra]\phi)$. 
    \end{enumerate}

    It remains to show that in Proc along every derivation path, there exists a progressive derivation trace. 
    Observing that in Proc, every derivation path must at least pass through a preproof branch Sub-proc 1 for infinite times. 
    So, it is enough to show that each Sub-proc 1 has a progressive derivation trace. 
    In a Sub-proc 1, consider two cases:
    \begin{enumerate}[(i)]
    \item If the modality is $[\alpha]$, by Definition~\ref{def:Progressive Step/Progressive Derivation Trace}, in every inference of rule $([\alpha]R)$, the CP pair on the right of both sequents is a progressive step. 
    \item If the modality is $\la\alpha\ra$, 
    since $\models (\Gamma\Rightarrow \sigma : \la\alpha\ra\phi)$, 
    $\models (\Gamma\Rightarrow \sigma\termi \alpha)$. 
    By (b) above, $\models (\Gamma'\Rightarrow \sigma' \termi \alpha)$. 
    By the completeness w.r.t. $\PTer$ (Item~\ref{item:assump 3} of Definition~\ref{def:Matching Proof System}), 
    $\pfDLp\vdash  (\Gamma'\Rightarrow \sigma' : \la\alpha\ra\phi)$. 
    So according to Definition~\ref{def:Program Properties}, in the first inference of rule $(\la\alpha\ra R)$,  the CP pair on the right of both sequents is a progressive step. 
    \end{enumerate}
    In both cases above, the progressive derivation trace is: $\sigma':[\la\alpha\ra]\phi$ in node 3, ..., $\sigma'':[\la\alpha\ra]\phi$ in node 4, $\sigma'':[\la\alpha\ra]\phi$, $\sigma'':[\la\alpha\ra]\phi$ in node 6, $\sigma':[\la\alpha\ra]\phi$ in node 7, ... as shown in Proc.  
    
\end{proof}

\section{A Cyclic Deduction of An Esterel Program}
\label{section:Example Two: A Synchronous Loop Program}
\ifx
In this section, we introduce another example of instantiations of $\DLp$ and display its cyclic deductions in system $\pfDLp$. 
We show how system $\pfDLp$ manages to verify a program whose ``loop structures'' are implicit. That is where $\DLp$ is really useful as the loop information can be tracked in the cyclic proof structures of $\DLp$.  
This example also depicts that how synchronous languages can benefit from symbolic-based reasoning in order to avoid additional program transformations. 
\fi


Esterel~\cite{Berry92} is a synchronous programming language for reactive systems. 
Below we first introduce the semantics of an Esterel program, then we explain why it needs extra program transformations in traditional verification frameworks.
Lastly, we explain how to express Esterel programs in $\DLp$ and give a cyclic deduction of this program. 

Note that the introduction we provide below is informal and does not cover all aspects of the semantics of Esterel. 
But it is enough to clear our point.

\subsection{An Esterel Program in $\DLp$}

We consider a synchronous program $\E$ of an instantiation $\Prog_\E$ of programs written in Esterel language~\cite{Berry92}:
$$
\E\dddef \{\Etrap\ A \parallel B\ \Eend\ \},
$$
where 
\begin{center}
    $\begin{aligned}
        A \dddef&\ \{\Eloop\ (\Eemit\ S(0)\ ;\ x := x - S\ ;\ \Eif\ x = 0\ \Ethen\ \Eexit\ \Eend\ ;\ \Epause)\ \Eend\ \}\\
        B \dddef&\ 
                \{\Eloop\ (\Eemit\ S(1)\ ;\ \Epause)\ \Eend\ \}. 
    \end{aligned}$
\end{center}

The behaviour of a synchronous program is characterized by a sequence of \emph{instances}.
At each instance, several (atomic) executions of a program may occur. 
The value of each variable is unique in an instance. 
When several programs run in parallel, 
their executions at one instance are thought to occur simultaneously. 
In this manner, the behaviour of a parallel synchronous program is deterministic.

In this example, the behaviour of the program $\E$ is illustrated as follows:
\begin{itemize}
    \item $x$, $S$ are two variables. $x$ is a local variable with $\mbb{Z}$ as its domain, $S$ is a ``signal'' whose domain is $\mbb{Z}\cup \{\bot\}$, with $\bot$ indicating the absense of a signal. 
    \item The key word $\Epause$ marks the end of an instance, when all signals are set to $\bot$, representing the state ``absence''.  
    \item Signal emission $\Eemit\ S(e)$ means assigning the value of an expression $e$ to a signal $S$ and broadcasts the value. $x:=x - S$ is the usual assignment as in FODL.  
    \item At each instance, program $A$ firstly emits signal $S$ with value $0$ and subtracts $x$ with the current value of $S$; 
then checks if $x = 0$. While program $B$ emits signal $S$ with value $1$. 
The value of signal $S$ in one instance should be the sum of all of its emitted values by different current programs. So the value of $S$ should be $1 + 0 = 1$. 
    \item The whole program $E$ continues executing until condition $x = 0$ is satisfied, when $\Eexit$ terminates the whole program by jumping out of the $\Etrap$ statement.  
\end{itemize}

\ifx
the key word $\Epause$ marks the end of an instance, when all signals are set to $\bot$, representing the state `absence'.  
The $\Eloop$ statements in programs $A$ and $B$ are executed repeatedly for an infinite number of times until some $\Eexit$ statement is encountered. 
Signal emission $\Eemit\ S(e)$ means assigning the value of an expression $e$ to a signal $S$ and broadcasts the value. 
At each instance, program $A$ firstly emits signal $S$ with value $0$ and subtracts $x$ with the current value of $S$; 
then checks if $x = 0$. 
While program $B$ emits signal $S$ with value $1$. 
The value of signal $S$ in one instance should be the sum of all of its emitted values by different current programs. 
So the value of $S$ should be $1 + 0 = 1$. 
The whole program $E$ continues executing until $x = 0$ is satisfied, when $\Eexit$ terminates the whole program by jumping out of the $\Etrap$ statement.  
\fi

\renewcommand{\arraystretch}{1.5}
\begin{table}[tb]
         \begin{center}
         \noindent\makebox[\textwidth]{%
         \scalebox{1}{
         \begin{tabular}{c | c | c | c }
         \toprule
         World $w$ & $w(x)$ & $w(S)$ & $n$th Instance\\
         \midrule
        $w_1$ & $3$ & $1$ & 1 \\
        $w_2$ & $2$ & $1$ & 2 \\
        $w_3$ & $1$ & $1$ & 3 \\
        $w_4$ & $0$ & $1$ & 4 \\
         \bottomrule
         \end{tabular}
              }
              }
          \end{center}
          \caption{Transitional Behaviours of Program $\E$ Starting From $w_1$}
          \label{table:Program Behaviours of Program E Starting From w1}
    \end{table}

Starting from an initial world $w_1$ with $w_1(x) = 3$ and $w_1(S) = 1$, 
we have an execution path $w_1w_2w_3w_4$ of program $E$ explained in Table~\ref{table:Program Behaviours of Program E Starting From w1}, where we omit the intermediate forms of programs during the execution.

\subsection{Prior Program Transformations in Esterel Programs}
In Esterel, the behaviour of a parallel program is usually not true interleaving. There exist data dependencies between its processes. 
For instance, in the program $\E$ above, 
the assignment $x := x - S$ can only be executed after all values of $S$ in both programs $A$ and $B$ are collected. In other words, $x := x - S$ can only be executed after $\Eemit\ S(0)$ and $\Eemit\ S(1)$. 

For a synchronous program like this, additional program transformations are mandatory (cf.~\cite{Gesell12}). 
We need to first transform the program $\E$, for example, into a sequential one as:
$$
\E' = \{\Etrap\ C\ \Eend\}, 
$$
where 
$$
C \dddef \{\Eloop\ \Eemit\ S(0)\ ;\ \Eemit\ S(1)\ ;\ x:=x - S\ ;\ \Eif\ x = 0\ \Ethen\ \Eexit\ \Eend\ ;\ \Epause)\ \Eend\}. 
$$
In $C$, we collect all micro steps happen in an instance from both $A$ and $B$, in a correct order. 
In~\cite{Gesell12}, $\E'$ is called an STA program. 
Note that such a prior transformation can be very heavy, since one can imagine that the behaviour of a parallel Esterel program can be very complex (e.g.~\cite{Berry1989-so}). 

\ifx
For example, in~\cite{Gesell12}, a Quartz program often needs to be firstly transformed into a canonical form called ``synchronous tuple assignment'', in order to perform a compositional reasoning in Hoare logic. 
As seen next in Section~\ref{section:cyclic Deduction of Program E}, 
our approach in the verification framework of $\DLp$ provides a more direct reasoning based on symbolic executions compared to~\cite{Gesell12}. 
\fi

\ifx
In a parallel Esterel program, the executions of concurrent programs are usually dependent on each other to maintain the consistency of simultaneous executions, imposing special orders of atomic executions in one instance. 
In the program $E$ above, for instance, 
the assignment $x := x - S$ can only be executed after all values of $S$ in both programs $A$ and $B$ are collected.
In other words, the assignment $x := x - S$ can only be executed after $\Eemit\ S(0)$ and $\Eemit\ S(1)$. 
This characteristic of Esterel programs makes direct compositional reasoning impossible because 
the executions between concurrent programs in an instance are not true interleaving. 
One has to symbolically execute the parallel program as a whole in order to obtain the right execution orders in each instance. 
\fi

\ifx
when initializing $x$ as value $3$, 
the values of all variables at the end of each instance are listed as follows:
\begin{itemize}
    \item instance 1: $(x = 3, S = 1)$;
    \item instance 2: $(x = 2, S = 1)$;
    \item instance 3: $(x = 1, S = 1)$;
    \item instance 4: $(x = 0, S = 1)$.
\end{itemize}
\fi

\subsection{Instantiation of $\DLp$ in Esterel Programs}
The instantiation process is similar to while programs. 
Let $\Prog_\E$ be the set of Esterel programs and $\Var_\E$ be the set of local variables and signals in Esterel. We assume a \PLK\ structure $\Kr_\E = (\Wd_\E, \trans, \mcl{I}_\E)$ for $\Prog_\E$, where each world $w\in \Wd_\E$ is a mapping $w : \Var_E\to (\mbb{Z}\cup \{\bot\})$ that maps each local variable to an integer and maps each signal to a value of $\mbb{Z}\cup \{\bot\}$.

A program configuration $\sigma\in \Conf_\E$ in $\Prog_\E$, as a label, is defined to capture the meaning of the structure
\begin{center}
$\{x_1\mapsto e_1 \Split ... \Split x_n\mapsto e_n\}$ ($n\ge 1$),
\end{center}
 where the only difference from that of a while program is that it is a stack (with $x_n\mapsto e_n$ the top element), allowing several local variables with the same name. 
For example, configuration $\{x\mapsto 5 \Split y\mapsto 1 \Split y \mapsto 2\}$ 
has two different local variables $y$, storing the values $1$ and $2$ respectively. 

Given a world $w\in \Wd_\E$, a label mapping $\lm_w\in \LM_\E$ (w.r.t. $w$) is defined such that for any configuration $\sigma\in \Conf_\E$ of the form: $\{x_1\mapsto e_1 \Split ... \Split x_n\mapsto e_n\}$, 
$\lm_w(\sigma)$ is a world satisfying that 
\begin{enumerate}[(1)]
    \item $\lm_w(\sigma)(x_i) = w(e_j)$ with $n\ge j\ge i\ge 1$ and $j$ the largest index for $x_i\mapsto e_j$ in $\sigma$ (i.e. the right-most value of variable $x_i$);
    \item $\lm_w(\sigma)(y) = w(y)$ for other variable $y\in \Var_\E$, 
\end{enumerate}
where $w(e)$ for an expression $e$ is defined similarly as in while programs (see Example~\ref{example:Label Mappings}). 
For example, we have $\lm_w(\{x\mapsto 5 \Split y\mapsto 1 \Split y \mapsto 2\})(y) = w(2) = 2$ for any $\lm_w$. 

We omit the details of the set $(\pfOper)_\E$ of rules for the operational semantics of Esterel programs, as they are too complex (cf.~\cite{Butucaru10}).

\ifx
Different from while programs, 
in program $\E$, $x, S\in \Var_E$ are two variables of different types. 
$x$ is a ``local variable'', with $\mbb{Z}$ as its domain. 
$S$ is a variable called ``signal'' whose domain is $\mbb{Z}\cup \{\bot\}$, with $\bot$ to express the absence of a signal. 
We assume a \PLK\ structure $\Kr_\E = (\Wd_\E, \trans, \mcl{I}_\E)$ for $\Prog_\E$, where each world $w\in \Wd_\E$ is a mapping $w : \Var_E\to (\mbb{Z}\cup \{\bot\})$ that maps each local variable to an integer and maps each signal to a value of $\mbb{Z}\cup \{\bot\}$. 

A program configuration $\sigma\in \Conf_\E$ in $\Prog_\E$, as a label, is defined to capture the meaning of the structure
\begin{center}
$\{x_1\mapsto e_1 \Split ... \Split x_n\mapsto e_n\}$ ($n\ge 1$),
\end{center}
 where the only difference from that of a while program is that it is a stack (with $x_n\mapsto e_n$ the top element), allowing several local variables with the same name. 
For example, configuration $\{x\mapsto 5 \Split y\mapsto 1 \Split y \mapsto 2\}$ 
has two different local variables $y$, storing the values $1$ and $2$ respectively. 

Given a world $w\in \Wd_\E$, a label mapping $\lm_w\in \LM_\E$ (w.r.t. $w$) is defined such that for any configuration $\sigma\in \Conf_\E$ of the form: $\{x_1\mapsto e_1 \Split ... \Split x_n\mapsto e_n\}$, 
$\lm_w(\sigma)$ is a world satisfying that 
\begin{enumerate}[(1)]
    \item $\lm_w(\sigma)(x_i) = w(e_j)$ with $n\ge j\ge i\ge 1$ and $j$ the largest index for $x_i\mapsto e_j$ in $\sigma$ (i.e. the right-most value of variable $x_i$);
    \item $\lm_w(\sigma)(y) = w(y)$ for other variable $y\in \Var_\E$, 
\end{enumerate}
where $w(e)$ for an expression $e$ is defined similarly as in while programs (see Example~\ref{example:Label Mappings}). 
For example, we have $\lm_w(\{x\mapsto 5 \Split y\mapsto 1 \Split y \mapsto 2\})(y) = w(2) = 2$ for any $\lm_w$. 
\fi

\subsection{A Cyclic Deduction of Program $\E$}
\label{section:cyclic Deduction of Program E}
In $\DLp$, program $\E$ can be directly reasoned about without additional program transformations. 
This is achieved because $\DLp$ supports a cyclic reasoning directly based on the operational semantics. 
During the following derivation, we see that the outside loop structure of $\E$ (which is $C$ above after the transformations) is actually reflected by the cyclic derivation tree itself.

We prove the property 
$$
\nu_2 \dddef \sigma_1 : x > 0\Rightarrow \sigma_1 : \la E\ra\true, 
$$
which says that under configuration $\sigma_1 = \{x\mapsto v, S\mapsto\bot\}$, with $v$ a fresh variable representing an initial value of $x$, if $x > 0$, then $E$ can finally terminate. 

\renewcommand{\arraystretch}{1}
\begin{table}[tb]
        \noindent\makebox[\textwidth]{%
        \scalebox{1}{
        \begin{tabular}{l|l}
        \toprule
        \begin{tabular}{l}
        \begin{tikzpicture}[->,>=stealth', node distance=3cm]
        \node[draw=none] (txt2) {
            $
                \infer[^{(\la\alpha\ra)}]
                {\mbox{$\nu_2 : 1$}}
                {\infer[^{(\la\alpha\ra)}]
                    {2}
                    {\infer[^{(\la\alpha\ra)}]
                        {3}
                        {\infer[^{(\sigma\textit{Cut})}]
                            {4}
                            {
                                \infer[^{(\textit{Wk} R)}]
                                {5}
                                {
                                    \infer[^{(\textit{Ter})}]
                                    {17}
                                    {}
                                }
                                &
                                \infer[^{(\sigma\vee L)}]
                                {6}
                                {
                                \infer[^{(\la\alpha\ra)}]
                                    {7}
                                    {
                                        \infer[^{(\la\alpha\ra )}]
                                        {12}
                                        {
                                            \infer[^{(\textit{LE})}]
                                            {13}
                                            {
                                                \infer[^{(\Sub)}]
                                                {14}
                                                {  
                                                    \infer[^{(\textit{LE})}]
                                                    {15}
                                                    {16}
                                                }
                                            }
                                        }
                                    }
                                &
                                \infer[^{(\la\alpha\ra)}]
                                    {8}
                                    {
                                    \infer[^{(\la\ter\ra)}]
                                        {9}
                                        {
                                        \infer[^{(\textit{Ter})}]
                                            {10}
                                            {
                                            }
                                        }
                                    }
                                }
                            }
                        }   
                    }
                }
            $
        };


        \path
        ;


        \draw[dotted,thick,red] ([xshift=-0.7cm, yshift=1.75cm]txt2.center) -- 
        ([xshift=-1.5cm, yshift=1.75cm]txt2.center) --
        ([xshift=-1.5cm, yshift=-1.75cm]txt2.center) --
        ([xshift=-1.25cm, yshift=-1.75cm]txt2.center);
        \end{tikzpicture}
        \end{tabular}
        &
        \begin{tabular}{l}
        Definitions of other symbols:\\
        $A \dddef
\Eloop\ (\Eemit\ S(0)\ ;\ x := x - S\ ;\ \Eif\ x = 0\ \Ethen\ \Eexit\ \Eend\ ;\ \Epause)\ \Eend$\\
        $B \dddef
        \Eloop\ (\Eemit\ S(1)\ ;\ \Epause)\ \Eend$\\
        $A' \dddef (\Eif\ x = 0\ \Ethen\ \Eexit\ \Eend\ ;\ \Epause)$\\
        $\sigma_1 \dddef \{x\mapsto v\Split S\mapsto \bot\}$\\
        $\sigma_2 \dddef \{x\mapsto v\Split S\mapsto 0\}$\\
        $\sigma_3 \dddef \{x\mapsto v\Split S\mapsto 1\}$\\
        $\sigma_4 \dddef \{x\mapsto v-1\Split S\mapsto 1\}$\\
        $\sigma_5 \dddef \{x\mapsto v-1\Split S\mapsto \bot\}$\\
        \end{tabular}
            
        \\
        \midrule
        \multicolumn{2}{l}{
            \begin{tabular}{l l l l}
            1: & $\sigma_1 : x > 0$ & $\Rightarrow$ & \ul{$\sigma_1 : \la \Etrap\ A\parallel B\ \Eend\ra\true$}
            \\
            2: & $\sigma_1 : x > 0$ & $\Rightarrow$ & \ul{$\sigma_2 : \la \Etrap \ ((x := x - S\ ; A')\ ;\ A)\parallel B)\ \Eend\ra\true$}
            \\
            3: & $\sigma_1 : x > 0$ & $\Rightarrow$ & \ul{$\sigma_3 : \la \Etrap \ ((x := x - S\ ; A')\ ;\ A)\parallel (\Epause\ ;\ B)\ \Eend\ra\true$}
            \\
            4: & $\sigma_1 : x > 0$ & $\Rightarrow$ & \ul{$\sigma_4 : \la \Etrap \ (A'\ ;\ A)\parallel (\Epause\ ;\ B)\ \Eend\ra\true$}
            \\
            5: &  $\sigma_1 : x > 0$ & $\Rightarrow$ & $\sigma_4 : \la \Etrap \ (A'\ ;\ A)\parallel (\Epause\ ;\ B)\ \Eend\ra\true, \sigma_1 : (x - 1 \neq 0\vee x - 1 = 0)$
            \\
            17: &  $\sigma_1 : x > 0$ & $\Rightarrow$ & $\sigma_1 : (x - 1 \neq 0\vee x - 1 = 0)$
            \\
            \midrule
            6: & $\sigma_1 : x > 0, \sigma_1 : (x - 1 \neq 0\vee x - 1 = 0)$ & $\Rightarrow$ & \ul{$\sigma_4 : \la \Etrap \ (A'\ ;\ A)\parallel (\Epause\ ;\ B)\ \Eend\ra\true$}
            \\
            7: & $\sigma_1 : x > 0, \sigma_1 : x - 1 \neq 0$ & $\Rightarrow$ & \ul{$\sigma_4 : \la \Etrap \ (A'\ ;\ A)\parallel (\Epause\ ;\ B)\ \Eend\ra \true$}
            \\
            12:& $\sigma_1 : x > 0, \sigma_1 : x - 1 \neq 0$ & $\Rightarrow$ & \ul{$\sigma_4 : \la \Etrap \ (\Epause\ ;\ A)\parallel (\Epause\ ;\ B)\ \Eend\ra \true$}
            \\
            13:& $\sigma_1 : x > 0, \sigma_1 : x - 1 \neq 0$ & $\Rightarrow$ & \ul{$\sigma_5 : \la \Etrap \ A\parallel B\ \Eend\ra \true$}
            \\
            14: & $\sigma_5 : x + 1 > 0, \sigma_5 : x \neq 0$ & $\Rightarrow$ & 
            \ul{$\sigma_5 : \la \Etrap \ A\parallel B\ \Eend\ra \true$}
            \\
            15: & $\sigma_1 : x + 1 > 0, \sigma_1 : x \neq 0$ & $\Rightarrow$ & \ul{$\sigma_1 : \la \Etrap \ A\parallel B\ \Eend\ra \true$}
            \\
            16:& $\sigma_1 : x > 0$ & $\Rightarrow$ & \ul{$\sigma_1 : \la \Etrap \ A\parallel B\ \Eend\ra \true$}
            \\
            \midrule
            8: & $\sigma_1 : x > 0, \sigma_1 : x - 1 = 0$ & $\Rightarrow$ & $\sigma_4 : \la \Etrap \ (A'\ ;\ A)\parallel (\Epause\ ;\ B)\ \Eend\ra \true$
            \\
            9: & $\sigma_1 : x > 0, \sigma_1 : x - 1 = 0$ & $\Rightarrow$ & $\sigma_5 : \la \ter\ra \true$
            \\
            10: & $\sigma_1 : x > 0, \sigma_1 : x - 1 = 0$ & $\Rightarrow$ & $\sigma_5 : \true$
            \\
        \end{tabular}
        }
        \\
        \bottomrule
        \end{tabular}
        }
        }
        \caption{Derivations of Property $\nu_2$}
        \label{figure:The derivation of Example 2}
\end{table}

The derivations of $\nu_2$ is depicted in Table~\ref{figure:The derivation of Example 2}. 
The symbolic executions of program $\E$ rely on rule 
$$
\begin{aligned}
    \infer[^{(\la\alpha\ra)}]
    {\Gamma\Rightarrow \sigma : \la\alpha\ra\phi, \Delta}
    {\Gamma\Rightarrow \sigma' : \la\alpha'\ra\phi, \Delta}
\end{aligned}, \mbox{ if $\pfDLp\vdash (\Gamma\Rightarrow (\alpha, \sigma)\trans (\alpha', \sigma'), \Delta)$}, 
$$
which can be derived by rules $([\alpha]L)$, $(\neg R)$ and $(\neg L)$ from Table~\ref{table:General Rules for LDL}. 
We omit all the side deductions of the program transitions and terminations in the instances of rule $(\la\alpha\ra)$. 

From node 2 to 3 is a progressive step, 
\ifx
where the derivation of  
the termination $\sigma_2\termi\Etrap \ ((x := x - S\ ; A')\ ;\ A)\parallel B\ \Eend$, 
depends on the operational semantics of Esterel. 
\fi
where to see that program $\sigma_2\termi\Etrap \ ((x := x - S\ ; A')\ ;\ A)\parallel B\ \Eend$ terminates, informally, we observe that the value of variable $x$ decreases by $1$ (by executing $x := x - S$) in each loop so that statement $\Eexit$ is finally reached. 
\ifx
Informally, this program does terminate as the value of variable $x$ decreases by $1$ (by executing $x := x - S$) in each loop so that statement $\Eexit$ is finally reached. 
\fi
From node 13 to 14 and node 15 to 16, rule
$$
\begin{aligned}
    \infer[^{(\textit{LE})}]
{\Gamma, \phi\Rightarrow \Delta}
{\Gamma, \phi'\Rightarrow\Delta},
\end{aligned}
 \ \ \mbox{if $\phi\to \phi'\in \Fmla$ is valid} 
$$
is applied, which can be derived by the following derivations:
$$
\begin{aligned}
    \infer[^{(\textit{Cut})}]
    {\Gamma, \phi\Rightarrow \Delta}
    {
        \infer[^{(\textit{Wk} L)}]
        {\Gamma, \phi, \phi'\Rightarrow\Delta}
        {\Gamma, \phi'\Rightarrow \Delta}
    &
        \infer[^{(\textit{Ter})}]
        {\Gamma, \phi\Rightarrow \phi', \Delta}
        {}
    }
\end{aligned}
$$
From node 14 to 15, rule $(\Sub)$
$$
\begin{aligned}
  \infer[^{(\Sub)}]
{\Gamma[e/x]\Rightarrow\Delta[e/x]}
{\Gamma\Rightarrow \Delta}  
\end{aligned}
$$
is applied, with $(\cdot)[e/x]$ an instantiation of the substitution $\Sub$ of labels (Definition~\ref{def:Substitution of Labels}). It is defined just as that in  the example of Section~\ref{section:Case Study}. 
Observe that sequent 14 can be written as: 
$$
\sigma_1[v-1/v] : x + 1 > 0, \sigma_1[v-1/v] : x\neq 0 \Rightarrow 
            \sigma_1[v-1/v] : \la \Etrap \ A\parallel B\ \Eend\ra \true. 
$$

Sequent 16 is a bud that back-links to sequent 1. 
The whole preproof is cyclic as 
the only derivation path: $1,2,3,4,6,7,12,13,14,15,16,1,...$ has 
a progressive trace whose elements are underlined in Table~\ref{figure:The derivation of Example 2}.

\section{An Encoding of Separation Logic in $\DLp$}
\label{section:An Encoding of Separation Logic in DLp}


We instantiate $\DLp$ to express separation logic~\cite{Reynolds02} --- an extension of Hoare logic for reasoning about shared mutable program data structures. 
The instantiated theory is called $\DLpSP$. 
Below we only deal with a part of separation logic, but it is enough to clear our point. 


\ifx
we encode a part of separation logic~\cite{Reynolds02} --- an extension of Hoare logic for reasoning about program data structures --- in which a type of more complex configurations are required to capture the notion of `heaps'. 
From this encoding, we shall see that a configuration interpretation can be defined more than just as substitutions of terms like $\app_\WP$ and $\app_\E$ in the examples in Section~\ref{section:Examples of Term Structures}. 
\fi

\textbf{Separation Logic}. 
In the following, we assume the readers are familiar with separation logic and we only give an informal explanations of its semantics.
For simplicity, we only introduce partial separation logic primitives: the atomic formula $e\allocto e'$ and critical operator $\sepc$ for formulas, and the atomic statements $x := \Alloc(e)$, $x := [e]$, $[e] := e'$ and $\disAlloc(e)$ for programs. 
We omit another critical operator $\sepi$ in formulas and the compositional programs that vary from case to case. 

\ifx
for formulas, we only introduce an atomic formula $e\allocto e'$ and a critical operator $\sepc$ (while omitting another critical operator $\sepi$); 
for programs, we only introduce atomic program statements: $x := \Alloc(e)$, $x := [e]$, $[e] := e'$ and $\disAlloc(e)$ that have direct influences on heaps, and ignore compositional programs that varies from case to case. 
After the introductions of the semantics, we give the encoding of separation logic into $\LDL$ as labeled formulas and explain how they actually capture the same meaning. 
\fi

Below we follow some conventions of notations: 
Given a partial function $f : A \partto B$, we use $\dom(f)$ to denote the domain of $f$. 
For a set $C$, partial function $f|_C : A\partto B$ is the function $f$ restricted on domain $\dom(f)\cap C$.  
$f[x\mapsto e]$ represents the partial function that maps $x$ to $e$, and maps the other variables in its domain to the same value as $f$ does. 

Let $\mbb{V} = \mbb{Z}\cup \textit{Addr}$ be the set of values, where 
$\textit{Addr}$ is a set of addresses. 
We assume $\textit{Addr}$ to be expressed with an infinite set of integer numbers. 
In separation logic, a store $s : V\to \mbb{V}$ is a function that maps each variable to a value of $\mbb{V}$, 
a heap $h : \textit{Addr}\partto \mbb{V}$ is a partial function that maps an address to a value of $\mbb{V}$, expressing that the value is stored in a memory indicated by the address. 
$\dom(h)$ is a finite subset of $\textit{Addr}$. 
A state is a store-heap pair $(s, h)$. 
The \emph{disjoint relation} $h_1\disj h_2$ is defined if $\dom(h_1)\cap\dom(h_2) = \emptyset$. 

Here we informally explain the semantics of each primitive. 
Given a state $(s,h)$, statement $x := \Alloc(e)$ allocates a memory addressed by a new integer $n$ in $h$ to store the value of expression $e$ (thus obtaining a new heap $h\cup \{(n, s(e))\}$ where $n\notin \dom(h)$), 
and assigns $n$ to $x$.
Statement $x := [e]$ assigns the value of the address $e$ in $h$ (i.e. $h(s(e))$) to variable $x$. 
$[e] := e'$ means to assign the value $e'$ to the memory of the address $e$ in $h$ (i.e. obtaining a new heap $h[s(e)\mapsto s(e')]$). 
$\disAlloc(e)$ means to de-allocate the memory of address $e$ in the heap 
(i.e. obtaining a new heap $h|_{\dom(h)\setminus\{s(e)\}}$). 
Formula $e\allocto e'$ means that value $e'$ is stored in the memory of address $e$. 
Given a state $(s, h)$, $s, h\models e\allocto e'$ is defined if $h(s(e)) = s(e')$. 
For any separation logical formulas $\phi$ and $\psi$, $s,h\models \phi\ \sepc\ \psi$ if 
there exist heaps $h_1, h_2$ such that $h = h_1\cup h_2$, $h_1\disj h_2$, and $s,h_1\models \phi$ and $s, h_2\models \psi$.

\begin{example}
    Let $(s, h)$ be a state such that $s(x) = 3, s(y) = 4$ and $h = \emptyset$, then 
    the following table shows the information of each state about focused variables and addresses during the process of the following executions: 
    $$
    (s, h)\xrightarrow{x := \Alloc(1)}
    (s_1,h_1)\xrightarrow {y := \Alloc(1)}
    (s_2,h_2)\xrightarrow {[y] := 37}
    (s_3,h_3)\xrightarrow {y := [x + 1]}
    (s_4,h_4)\xrightarrow{\disAlloc(x+1)}
    (s_5,h_5).
    $$
    \begin{center}
        \begin{tabular}{| c | l | c | l |}
        \toprule
        & \multicolumn{1}{c|}{Store} & & \multicolumn{1}{c|}{Heap}\\
        \hline
        $s$ & $x: 3$, $y: 4$ & $h$ & empty\\
        \hline
        $s_1$ & $x: 37$, $y: 4$ & $h_1$ & $37 : 1$\\
        \hline
        $s_2$ & $x : 37$, $y : 38$ & $h_2$ & $37: 1$, $38 : 1$\\
        \hline
        $s_3$ & $x : 37$, $y : 38$ & $h_3$ & $37 : 1$, $38 : 37$\\
        \hline
        $s_4$ & $x : 37$, $y : 37$ & $h_4$ & $37: 1$, $38 : 37$\\
        \hline
        $s_5$ & $x : 37$, $y : 37$ & $h_5$ & $37 : 1$\\
        \bottomrule
    \end{tabular}   
    \end{center}
Let $\phi\dddef (x\allocto 1\ \sepc\ y\allocto 1)$, $\psi\dddef (x\allocto 1\ \wedge\ y\allocto 1)$, 
we have $s_2, h_2\models \phi$ and $s_2, h_2\models \psi$, $s_5, h_5\models \psi$, but $s_5, h_5\not\models \phi$ since $x$ and $y$ point to the single memory storing value $1$.  
\end{example}

\renewcommand{\arraystretch}{1.5}
\begin{table}[htpb]
         \begin{center}
         \noindent\makebox[\textwidth]{%
         \scalebox{1}{
         \begin{tabular}{c}
         \toprule
         \\
         $
             \infer[^{1\ (\Alloc)}]
             {\Gamma \Rightarrow (x := \bff{cons}(e), (s, h))\trans (\ter, (s[x\mapsto n], h\cup \{(n,s(e))\})), \Delta}
             {}
         $, \ where $n$ is new w.r.t. $h$
         \\
         $
             \infer[^{(x := [e])}]
             {\Gamma \Rightarrow (x := [e], (s, h))\trans (\ter, (s[x\mapsto h(s(e))], h)), \Delta}
             {}
         $
         \\ 
         $
             \infer[^{([e] := e')}]
             {\Gamma \Rightarrow ([e] := e', (s, h))\trans (\ter, (s, h[s(e)\mapsto s(e')])), \Delta}
             {}
         $
         \\ 
            $
             \infer[^{(\disAlloc)}]
             {\Gamma \Rightarrow (\disAlloc(e), (s, h))\trans (\ter, (s, h|_{\dom(h)\setminus\{s(e)\}})), \Delta}
             {}
          $
        \\
         \bottomrule
         \end{tabular}
              }
              }
          \end{center}
          \caption{Partial Rules of $(\pfOper)_\Sep$ for Program Transitions of Atomic Statements in Separation Logic}
          \label{table:Encoding of A Part of Separation Logic}
    \end{table}

\textbf{Encoding of Separation Logic in $\LDL$}. 
In $\DLp$, let $\Prog_\Sep$ and $\Fmla_\Sep$ be the set of programs and formulas of separation logic. 
In the \PLK\ structure $\Kr_\Sep = (\Wd_\Sep, \trans_\Sep, \I_\Sep)$ of separation logic, $\Wd_\Sep = \{(s, h)\ |\ s : V\to \mbb{V}, h : \textit{Addr}\partto \mbb{V}\}$,  $\I_\Sep$ interprets each formula of $\Fmla_\Sep$ as explained above. 
We directly choose the store-heap pairs as the configurations of separation logic named $\Conf_\Sep$.  
In this case, we simply let $\LM\dddef \{\tau\}$, where $\tau : \Conf_\Sep\to \Wd_\Sep$ is a constant mapping satisfying that $\tau((s,h))\dddef (s,h)$ for any $(s, h)\in \Wd_\Sep$. 
Table~\ref{table:Encoding of A Part of Separation Logic} lists the rules for  the program transitions of the atomic statements in $(\pfOper)_\Sep$. 

To further derive the formulas like $\phi\ \sepc\ \psi$ in $\Fmla_\Sep$ into simpler forms, additional rules apart from $\pfLDLp$ need to be proposed. 
For example, we can propose a rule 
$$
\begin{aligned}
 \infer[^{(\sigma \sepc)}]
{\Gamma\Rightarrow (s, h_1\cup h_2) : \phi\ \sepc\ \psi, \Delta}
{
\Gamma\Rightarrow h_1\bot h_2, \Delta
&
\Gamma\Rightarrow (s,h_1) : \phi, \Delta
&
\Gamma\Rightarrow (s, h_2) : \psi, \Delta
}   
\end{aligned}
$$
to decompose the heap $h_1\cup h_2$ and the formula $\phi\ \sepc\ \psi$, or a frame rule
$$
\begin{aligned}
    \infer[^{(\sigma \textit{Frm})}]
    {\Gamma\Rightarrow (s, h) : \phi\ \sepc\ \psi, \Delta}
    {\Gamma\Rightarrow (s, h) : \phi, \Delta}
\end{aligned}, \mbox{ if no variables of $\dom(h)$ appear in $\psi$}, 
$$
to just decompose the formula $\phi\ \sepc\ \psi$. 
These rules are inspired from their
counterparts for programs in separation logic. 

In practice, the labels can be more explicit structures than the store-heap pairs shown here. Similar encoding can be obtained accordingly. 
From this example, we envision that the entire theory of separation logic can be embedded into $\LDL$, where additional rules like the above ones support a ``configuration-level reasoning'' of separation-logic formulas. 

\ifx
Table~\ref{table:Encoding of A Part of Separation Logic} lists the rules for  the program behaviours of the atomic  program statements. 
To capture the semantics of separation logical formulas, interpretations $\app_\Sep$ are defined in an inductive way according to their syntactical structures. 
Note that the formula $\app_\Sep((s,h), \phi\ \sepc\ \psi)$ requires variables $X, Y$ ranging over heaps. 



To tackle a formula like $\app_\Sep((s,h), \phi\ \sepc\ \psi)$, additional rules can be proposed. 
In this example, we can propose a rule
$$
\begin{aligned}
 \infer[^{(\sigma \sepc)}]
{\Gamma\Rightarrow (s, h_1\cup h_2) : \phi\ \sepc\ \psi, \Delta}
{
\Gamma\Rightarrow h_1\bot h_2, \Delta
&
\Gamma\Rightarrow (s,h_1) : \phi, \Delta
&
\Gamma\Rightarrow (s, h_2) : \psi, \Delta
}
\end{aligned}
$$
to decompose the heap's structure, or a rule
$$
\begin{aligned}
    \infer[^{(\sigma \textit{Frm})}]
    {\Gamma\Rightarrow (s, h) : \phi\ \sepc\ \psi, \Delta}
    {\Gamma\Rightarrow (s, h) : \phi, \Delta}
\end{aligned}, \mbox{ if no variables of $\dom(h)$ appear in $\psi$}, 
$$
to simplify the formula's structure. 
These rules are inspired from their
counterparts for programs in separation logic. 

In practice, configurations can be more explicit structures than store-heap pairs $(s, h)$. Similar encoding can be obtained accordingly. 
From this example, we envision that the entire theory of separation logic can actually be embedded into $\LDL$, where the additional rules like the above ones support a ``configuration-level reasoning'' of separation logical formulas. 

\fi

\end{document}